\documentclass[screen,acmsmall,nonacm]{acmart}

\usepackage{amssymb}
\usepackage{xcolor} 
\usepackage{stmaryrd}
\usepackage{balance}
\usepackage[Dfprime,bbsets]{math}

\title{Differential Equations as Fixpoints and Games}
\author{Noah {Abou El Wafa}}
\email{noah.abouelwafa@kit.edu}
\orcid{0000-0002-3987-9919}
\author{Andr\'e Platzer} 
\email{platzer@kit.edu}
\orcid{0000-0001-7238-5710}
\affiliation{%
  \institution{Karlsruhe Institute of Technology}
  \city{Karlsruhe}
  \country{Germany}
  }

\ccsdesc[500]{Mathematics of computing~Ordinary differential equations}
\ccsdesc[500]{Theory of computation~Timed and hybrid models}
\ccsdesc[500]{Theory of computation~Proof Theory}
\ccsdesc[500]{Theory of computation~Modal and temporal logics}
\ccsdesc[500]{Theory of computation~Programming logic}

\newif\iflongversion
\newif\ifinlineproofs 
\inlineproofsfalse
\inlineproofstrue 
\longversiontrue

\usepackage{environ}
\usepackage{mathtools}

\newcommand{\odo}[1]{}

\newcommand{\mutex}{\texorpdfstring{\ensuremath{\mu}}{mu}}

\iflongversion
    \usepackage[createShortEnv]{proof-at-the-end}
    \newcommand{\apprefexp}[2]{\Cref{#1} in \Cref{#2}\xspace}
    \newcommand{\appref}[1]{\Cref{#1}\xspace}
    
\else
    \newcommand{\citefullarx}{} 
    \usepackage[createShortEnv,conf={text link={\noexpand\citefullarx}}]{proof-at-the-end}
    \newcommand{\thelongversion}[2]{the long version~\cite{arxivversion}\xspace}
    
\fi
\ifinlineproofs
    \renewenvironment{proofE}{\begin{proof}}{\end{proof}}
\fi

\definecolor{semblue}{rgb}{0,0,0.7}
\definecolor{vgreen}{rgb}{.1,.5,0}
\definecolor{vred}{rgb}{.7,0,0}
\definecolor{vblue}{rgb}{.1,.15,.62}
\definecolor{vgray}{rgb}{.35,.35,.35}
\definecolor{semblue}{rgb}{0,0,0.7}
\definecolor{linkblue}{RGB}{0, 102, 204}

\newcommand{\pow}[1]{\mathcal{P}(#1)}
\newcommand{\setcomplement}[2]{#1\setminus#2}

\newcommand{\setdual}[1]{{#1}^{\mathrm{d}}}
\newcommand{\Intersection}{\bigcap}
\newcommand{\upwardclose}[1]{#1{\kern-2pt}\uparrow{}}

\newcommand{\grassign}{\;\;\;{::=}\;\;\;}
\newcommand{\grsep}{\mid}

\renewcommand{\|}{\grsep}

\newcommand{\term}{\theta}
\newcommand{\termb}{\delta}
\newcommand{\fml}{\varphi}
\newcommand{\fmlb}{\psi}
\newcommand{\fmlc}{\rho}
\newcommand{\fmld}{\eta}

\newcommand{\atfml}{p}

\newcommand{\ivars}{\mathcal{V}}
\newcommand{\ivar}{\mathrm{x}}
\newcommand{\ivarb}{\mathrm{y}}
\newcommand{\ivarc}{\mathrm{z}}
\newcommand{\ivarcode}{\mathrm{c}}
\newcommand{\ivarseq}{\bar{\ivar}}
\newcommand{\ivarbseq}{\bar{\ivarb}}
\newcommand{\ivarcseq}{\bar{\ivarc}}

\newcommand{\funcsymb}{f}

\newcommand{\ivarname}{individual variable\xspace}
\newcommand{\pvarname}{fixpoint variable\xspace}

\newcommand{\pvars}{\mathbb{V}}
\newcommand{\pvar}{X}
\newcommand{\pvarb}{Y}

\newcommand{\relsymb}{R}

\newcommand{\game}{\alpha}
\newcommand{\gameb}{\beta}

\newcommand{\gamesymb}{\mathfrak{a}}

\usepackage{amsmath}
\makeatletter
\newcommand{\xdashleftarrow}[2][]{\ext@arrow 3095\leftarrowfill@@{#1}{#2}}
\def\rightarrowfill@@{\arrowfill@@\relax\relbar\rightarrow}
\def\leftarrowfill@@{\arrowfill@@\dashleftarrow \relbar\relax}
\def\leftrightarrowfill@@{\arrowfill@@\leftarrow\relbar\rightarrow}
\def\arrowfill@@#1#2#3#4{%
    $\m@th\thickmuskip0mu\medmuskip\thickmuskip\thinmuskip\thickmuskip
        \relax#4#1
        \xleaders\hbox{$#4#2$}\hfill
        #3$%
}
\makeatother
\newcommand{\gamesymbat}[3][\gamesymb]{#2{\xleftarrow{#1}}#3}
\newcommand{\gamesymbatprop}[3][\gamesymb]{{#2{{\xdashleftarrow{#1}}}#3 }}

\renewcommand{\gamesymbat}[3][\gamesymb]{%
  #2{%
    \ensuremath{%
      \raisebox{0ex}{$%
        \xleftarrow{%
          \smash{\raisebox{-2pt}{\text{$\scriptscriptstyle\kern1pt#1$}}}%
        }$}%
    }%
  }\,#3%
}
\renewcommand{\gamesymbatprop}[3][\gamesymb]{%
  #2{%
    \ensuremath{%
      \raisebox{0ex}{{$%
        \xdashleftarrow{%
          \smash{\raisebox{-2pt}{\text{$\scriptscriptstyle\kern1pt#1$}}}%
        }$}}%
    }%
  }\,#3%
}

\newcommand{\lpossible}[1]{\langle#1\rangle}
\newcommand{\lnecessary}[1]{[#1]}

\newcommand{\gtest}[1]{{?}{#1}}
\newcommand{\gtestp}[1]{{?}{(#1)}}
\newcommand{\gcom}{;}
\newcommand{\gassign}[2]{{#1{\gassignsymb}#2}}

\newcommand{\gachoice}{\cup}

\newcommand{\gdchoice}{\cap}
\newcommand{\garepeat}[1]{#1{}^*}
\newcommand{\garepeatp}[1]{\garepeat{(#1)}}
\newcommand{\gdrepeat}[1]{#1^\times}
\newcommand{\gdrepeatp}[1]{\garepeat{(#1)}}
\newcommand{\garepeatn}[2]{#1^{\leq#2}}

\newcommand{\gassignsymb}{\coloneq}
\newcommand{\gndassignsymb}{*}
\newcommand{\gndassign}[1]{{#1{\gassignsymb}*}}
\newcommand{\gndassignrat}[1]{\gassign{#1}{*}_{\upharpoonright\rationals}}

\newcommand{\evdsep}{\,\&\,}

\newcommand{\gode}[2]{\D{#1}{=}#2}
\newcommand{\godeev}[3]{\D{#1}{=}#2\evdsep{#3}}

\newcommand{\godet}[2]{\D{#1}(\timevar){\overset{\scriptscriptstyle}{=}}#2}
\newcommand{\godetev}[3]{\D{#1}(\timevar){\overset{\scriptscriptstyle}{=}}#2\evdsep#3}
 
\newcommand{\termseq}{\bar{\term}}  
\newcommand{\termbseq}{\bar{\termb}}

\newcommand{\contreach}[4][]{#3\xrightarrow{#2}_{#1}#4}
\newcommand{\contreachin}[5][]{#3\xrightarrow{#2\evdsep#5}_{#1}#4}

\newcommand{\flfp}[2]{{\mu}{#1}.{#2}}

\newcommand{\fgfpnum}[3]{{\nu}^{#1}{#2}.{#3}}

\newcommand{\fgfp}[2]{{\nu}{#1}.{#2}}

\newcommand{\gdual}[1]{#1^{\mathrm{d}}}
\newcommand{\gdualp}[1]{\gdual{(#1)}}

\newcommand{\fosignature}{first-order signature\xspace}
\newcommand{\gamesignature}{game signature\xspace}
\newcommand{\gsig}{\mathcal{L}}

\newcommand{\gamesymbol}{action symbol\xspace}

\newcommand{\GameSymbol}{Action Symbol\xspace}

\newcommand{\fonstructure}{first-order neighborhood structure\xspace}
\newcommand{\FonStructure}{First-order neighborhood Structure\xspace}

\newcommand{\fonstr}{\mathcal{N}}
\newcommand{\fonstrdom}[1][\fonstr]{\abs{#1}}
\newcommand{\fonstrint}[2][\fonstr]{#2^{#1}}
\newcommand{\fonstrel}{u}
\newcommand{\fonstrelb}{v}

\newcommand{\fostates}{\mathcal{S}}
\newcommand{\fostate}{\omega}
\newcommand{\fostateb}{\eta}
\newcommand{\fosubset}{E}

\newcommand{\fosubsetvar}{X}

\newcommand{\fostatessubset}{S}
\newcommand{\fostatessubsetb}{T}
\newcommand{\fostatessubsetvar}{U}

\newcommand{\pnstrpintp}[2]{{#1}\llparenthesis{#2}\rrparenthesis}
\renewcommand{\pnstrpintp}[2]{{#1}\envelope{#2}}
\newcommand{\pnstrtrintp}[2]{\pnstrpintp{#1}{#2}}

\newcommand{\streplaceby}[3]{#1_{#2}^{#3}}
\newcommand{\strestrvar}[2]{#1{\upharpoonright}#2}
\newcommand{\varcomp}[1]{{#1}^c}

\newcommand{\envelope}[1]{\llbracket{#1}\rrbracket}

\newcommand{\interp}{interpretation\xspace}
\newcommand{\intp}{\mathcal{I}}
\newcommand{\intreplaceby}[3][\intp]{#1_{#2}^{#3}}
\renewcommand{\intreplaceby}[3][\intp]{#1[{#2}/{#3}]}
\renewcommand{\intreplaceby}[3][\intp]{#1[{#2}{\mapsto}{#3}]}

\newcommand{\Gamefuncs}[1]{\mathcal{G}(#1)}
\newcommand{\gamefunc}{F}
\newcommand{\gamefuncb}{G}

\newcommand{\Nfuncsmix}[2]{\mathcal{Q}(#1;#2)}
\newcommand{\nfunc}{F}

\newcommand{\mlfp}[2]{{\mu}{#1}.{#2}}
\newcommand{\mlfpp}[2]{\mlfp{#1}{(#2)}}
\newcommand{\mlfpi}[3]{{\mu^{#1}}{#2}.{#3}}
\newcommand{\mlfppi}[3]{{\mu^{#1}}{#2}.{(#3)}}
\newcommand{\mgfp}[3][]{{\nu^{#1}}{#2}.{#3}}

\newcommand{\ltden}[2][\fostate]{#1\envelope{#2}}
\newcommand{\glfden}[2][\fonstr]{#1\envelope{#2}}
\newcommand{\glgden}[2][\fonstr]{#1\envelope{#2}}
\newcommand{\lmden}[3][\fonstr]{#1\envelope{#3}^{#2}}

\newcommand{\glfdenprop}[2][\pnstr]{#1\envelope{#2}}
\newcommand{\glgdenprop}[2][\pnstr]{#1\envelope{#2}}
\newcommand{\lmdenprop}[3][\pnstr]{#1\envelope{#3}^{#2}}

\newcommand{\FirstorderMuCalculus}{First-order \mutex-Calculus\xspace}

\newcommand{\firstordermucalculus}{first-order \mutex-calculus\xspace}

\newcommand{\firstordergamelogic}{first-order game logic\xspace}
\newcommand{\Firstordergamelogic}{First-order game logic\xspace}
\newcommand{\FirstorderGameLogic}{First-order Game Logic\xspace}

\newcommand{\gamelogic}{game logic\xspace}

\newcommand{\fotoprops}{\flat}
\newcommand{\propstofo}{\sharp}

\newcommand{\gsigtoprop}[1][\gsig]{#1^\fotoprops}

\newcommand{\atfmltoprop}[1]{\smash{{#1}^\fotoprops}}

\newcommand{\fgentoprop}[1]{#1^\fotoprops}
\newcommand{\fgentopropp}[1]{(#1)^\fotoprops}

\newcommand{\fglfmltoprop}[1]{{#1}^\fotoprops}
\newcommand{\fglgtoprop}[1]{{#1}^\fotoprops}
\newcommand{\flmfmltoprop}[1]{{#1}^\fotoprops}

\newcommand{\fglgtopropp}[1]{\fglgtoprop{(#1)}}

\newcommand{\pglfmltofo}[1]{{#1}^{{\propstofo}}}
\newcommand{\pglgtofo}[1]{#1^{{\propstofo}}}
\newcommand{\plmfmltofo}[1]{#1{}^{{\propstofo}}}

\newcommand{\fonstrtons}[1][\fonstr]{#1^\fotoprops}

\newcounter{aligntablecounter}
\newenvironment{aligntable}[1][1]{%
    \newcommand{\eqsym}{=}
    \newcommand{\displayitem}[2]{&{##1}\eqsym##2}%
    \newcommand{\nextit}[2]{\displayitem{##1}{##2}\nexti}%
    \newcommand{\nexti}{%gndassign
        \stepcounter{aligntablecounter}%
        \ifnum\value{aligntablecounter}=#1%
            \setcounter{aligntablecounter}{0}\\%
        \else%
            &%
        \fi%
    }%
    \newcommand{\lastit}[2]{\displayitem{##1}{##2}}%
}{\setcounter{aligntablecounter}{0}}

\AtEndPreamble{\theoremstyle{remark}
\newtheorem{remark}[theorem]{Remark}}

\newcommand{\folmus}[1][\gsig]{\ensuremath{\smash{#1_\mu}}\xspace}
\newcommand{\folmusort}{\texorpdfstring{\folmus\kern-3pt\xspace}{first-order mu-calculus}}
\newcommand{\fogls}[1][\gsig]{\ensuremath{\mathrm{G}{\kern-0.15em}{#1}}\xspace}
\newcommand{\foglsort}{\texorpdfstring{\fogls}{first-order game logic}}
\newcommand{\foglss}[1][\gsig]{\ensuremath{\mathrm{G}{\kern-0.15em}{#1}_\mathrm{s}}\xspace}

\newcommand{\plmus}{\ensuremath{\smash{\mathsf{L}_\mu}}\xspace}
\newcommand{\pgls}{\ensuremath{\mathsf{GL}}\xspace}
\newcommand{\prlgl}{\ensuremath{\mathsf{rlGL}}\xspace}
\newcommand{\pglss}{\ensuremath{\mathsf{GL}_\mathrm{s}}\xspace}

\newcommand{\folmusone}[1][\gsig]{\ensuremath{\smash{#1_\mu^1}}\xspace}
\newcommand{\foglsone}[1][\gsig]{\ensuremath{\smash{\mathrm{G}{\kern-0.15em}{#1}^1}}\xspace}
\newcommand{\folmusstar}[1][\gsig]{\ensuremath{\smash{#1_\mu^*}}\xspace}
\newcommand{\foglsstar}[1][\gsig]{\ensuremath{\smash{\mathrm{G}{\kern-0.15em}{#1}^*}}\xspace}

\newcommand{\differentialgamelogic}{differential game logic\xspace}
\newcommand{\Differentialgamelogic}{Differential game logic\xspace}
\newcommand{\DifferentialGameLogic}{Differential Game Logic\xspace}
\newcommand{\differentialdynamiclogic}{differential dynamic logic\xspace}
\newcommand{\Differentialdynamiclogic}{Differential dynamic logic\xspace}
\newcommand{\DifferentialDynamicLogic}{Differential Dynamic Logic\xspace}
\newcommand{\differentialmucalculus}{differential \(\mu\)-calculus\xspace}

\newcommand{\systemicdlmus}{\ensuremath{\mathrm{d}\kern-1pt\mathsf{L}_{\exists\mu}}\xspace}
\newcommand{\systemicformula}{simple\xspace}
\newcommand{\angelicformulas}{existential\xspace}

\newcommand{\klary}[2]{\(#1\leftarrow#2\)~-ary}
\renewcommand{\klary}[2]{\((#1,#2)\)-ary}
\renewcommand{\klary}[2]{\({#1}{\shortleftarrow}{#2}\)}
\renewcommand{\klary}[2]{\({#2}\)-to-\({#1}\)}

\newcommand{\dLmu}{\ensuremath{\mathrm{d}\kern-1pt\mathsf{L}_\mu}\xspace}

\newcommand{\dsymb}[1][\odefof]{\mathrm{d}({#1\evdsep\evdfml})}

\newcommand{\dglsig}{\gsig^{\mathbb{R}}}

\newcommand{\logics}{\mathbb{L}}
\newcommand{\logicsa}{\mathbb{L}}
\newcommand{\logicsb}{\mathbb{K}}

\newcommand{\provrel}[3]{#1\vdash_{{}_{#2}}#3}
\newcommand{\valrel}[3]{#1\vDash_{{}_{#2}}#3}

\newcommand{\provlog}[2][]{\provrel{#1}{#2}{}}

\newcommand{\gentranslations}{K}
\newcommand{\gentranslationsb}{L}
\newcommand{\gentranslation}[2][\gentranslations]{{#2}^{#1}}
\newcommand{\gentranslationp}[2][\gentranslations]{{(#2)}^{#1}}
\newcommand{\gentranslationb}[2][\gentranslationsb]{{#2}^{#1}}
\newcommand{\gentranslationab}[1]{{#1}^{\gentranslations\gentranslationsb}}
\newcommand{\gentranslationba}[1]{{#1}^{\gentranslationsb\gentranslations}}

\newcommand{\semparset}{\mathbb{P}}
\newcommand{\semparsetg}[1]{\mathbb{S}(#1)}

\newcommand{\sempar}{p}
\newcommand{\semparb}{q}

\newcommand{\gensem}[4][\semantics]{#2,#3\vDash_{#1}#4}
\newcommand{\gensemsemival}[3][\semantics]{#2\vDash_{#1}#3}
\newcommand{\gensemval}[2][\semantics]{\vDash_{#1}#2}
\newcommand{\gensemb}[3][\semantics]{#2\envelope{#3}_{#1}}

\newcommand{\folmusem}[3]{\gensem[\folmus]{#1}{#2}{#3}}

\newcommand{\foglsem}[3]{\gensem[\fogls]{#1}{#2}{#3}}

\newcommand{\provfolmu}[2][]{\provrel{#1}{\folmus}{#2}}
\newcommand{\provfogl}[2][]{\provrel{#1}{\fogls}{#2}}
\newcommand{\valfolmu}[2][]{\valrel{#1}{\folmus}{#2}}
\newcommand{\valfogl}[2][]{\valrel{#1}{\fogls}{#2}}

\newcommand{\gvar}{x}

\newcommand{\gvarseq}{\bar{x}}
\newcommand{\gamebseq}{\bar{\gameb}}

\newcommand{\glfps}{\mathsf{r}}
\newcommand{\ggfps}{\rotatebox[origin=c]{180}{\ensuremath{\mathsf{r}}}}
\newcommand{\glfp}[2]{\glfps{}#1. #2}{}
\newcommand{\ggfp}[2]{\ggfps{}#1. #2}

\newcommand{\foaxiomsinprop}{\mathcal{F}}

\newcommand{\provpglplus}[1]{\provrel{\foaxiomsinprop}{\pgls}{#1}}

\newcommand{\provplmuplusassumption}[2]{\provrel{\foaxiomsinprop,#1}{\plmus}{#2}}
\newcommand{\provpglplusassumption}[2]{\provrel{\foaxiomsinprop,#1}{\pgls}{#2}}
\newcommand{\provprlglplusassumption}[2]{\provrel{\foaxiomsinprop,#1}{\prlgl}{#2}}

\newcommand{\plmuplusdiff}{\smash{\mathsf{L}_\mu^*}}
\newcommand{\pglplusdiff}{\pgls^*}
\newcommand{\prlglplusdiff}{\prlgl^*}
\newcommand{\plmuplusone}{\smash{\mathsf{L}_\mu^1}}%
\newcommand{\pglplusone}{\pgls^1}
\newcommand{\prlglplusone}{\prlgl^1}

\newcommand{\plmupluss}{\plmus+\foaxiomsinprop}

\newcommand{\glrlglinc}{\mathsf{i}}
\newcommand{\rlglremym}{\mathsf{r}}
\newcommand{\plmutorlglsym}{\mathsf{g}}
\newcommand{\pgltolmus}{\mathsf{f}}

\newcommand{\plmutorlgl}[1]{#1{}^{\plmutorlglsym}}
\newcommand{\pgltolmu}[1]{#1{}^{\pgltolmus}}
\newcommand{\rlglrem}[1]{#1{}^{\rlglremym}}

\newcommand{\Playerone}{Angel\xspace}
\newcommand{\Playertwo}{Demon\xspace}

\newcommand{\ud}{\mathrm{d}}

\newcommand{\Rc}{\mathcal{R}}

\newcommand{\cont}{c_\top}
\newcommand{\conf}{c_\bot}

\newcommand{\gasab}[1]{{\sim{}\kern-3pt #1}}
\newcommand{\gdsab}[1]{\gdual{{\sim{\kern-3pt}#1}}}

\newcommand{\semantics}{\mathcal{S}}
\newcommand{\semanticsc}{\mathcal{S}}
\newcommand{\semanticsd}{\mathcal{T}}

\usepackage{paralist}
\usepackage{enumitem} 
\newenvironment*{caselist}[1][- Case]{%
    \newcommand{\case}[1]{\item[\noindent\emph{#1} ##1:]}
    \begin{inparaitem}
        }{%
    \end{inparaitem}
}

\newcommand{\fmlfreesetto}[3][\fml]{\fmlreplacepvarby[#1]{#2}{#3}}
\newcommand{\fmlreplacevarby}[3][\fml]{#1\tfrac{#3}{#2}}
\newcommand{\fmlreplacepvarby}[3][\fml]{#1[{#2}/{#3}]} 
\renewcommand{\fmlreplacepvarby}[3][\fml]{#1[{#2}{\mapsto}{#3}]}
\newcommand{\fmlreplacepvarpby}[3][\fml]{\fmlreplacepvarby[(#1)]{#2}{#3}}
\newcommand{\fmlreplacevarpby}[3][\fml]{(#1)\tfrac{#3}{#2}}
\newcommand{\gamereplacevarby}[3][\game]{#1\tfrac{#3}{#2}}
\newcommand{\gamereplacevarpby}[3][\game]{(#1)\tfrac{#3}{#2}}

\newcommand{\pstatepre}{{\bar{\fonstrel}}}
\newcommand{\pstatepreel}[1]{{{\fonstrel}_{#1}}}
\newcommand{\pstatepreatvar}[2][\pstatepre]{^#1_{#2}}
\newcommand{\replpstate}[2][\fostate]{#1{}#2}

\newcommand{\compactset}{K}
\newcommand{\derbound}{k_1}
\newcommand{\secondderbound}{k_2}
\newcommand{\derbounds}{k}
\newcommand{\supnorm}[2][\compactset]{\norm[{#1}]}
\renewcommand{\supnorm}[2][\compactset]{\abs{#2}_{#1}}
\newcommand{\realnorm}[1]{\abs{#1}}

\newcommand{\reachrelsymbol}{\Rc}
\newcommand{\reachrelofin}[2]{\reachrelsymbol({#1{\evdsep}#2})}
\newcommand{\growthcond}[3][\evdset]{\mathcal{B}_{#2\evdsep#1,#3}}

\newcommand{\evdset}{C}
\newcommand{\evdsetbound}[2][\evdset]{{#1}_{\upharpoonright#2}}

\newcommand{\evdfml}{\psi}

\newcommand{\splitsymbol}{\mathcal{S}}
\newcommand{\splitop}[3][\evdset]{\splitsymbol_{#2\evdsep#1,#3}}

\newcommand{\rfnstr}{\mathbb{R}}

\newcommand{\epsnbhd}[2][\varepsilon]{B_{#1}(#2)}

\newcommand{\singularsemantics}{semantics\xspace}

\newcommand{\ivarsall}{\bar{\mathrm{v}}}
\newcommand{\ivarsallel}[1]{\mathrm{v}_{#1}}

\newcommand{\synequiv}{{\;}{\equiv}{\;}}
\newcommand{\notsynequiv}{{\;}{\not\equiv}{\;}}

\newcommand{\bigand}{\bigwedge}

\newcommand{\finseqcodefml}{\Theta}
\newcommand{\defassumption}{\Delta_\fml}

\newcommand{\dLmurat}{\ensuremath{\mathrm{d}\mathcal{L}_{\mu\rationals}}\xspace}
\newcommand{\dLmuratminus}{\ensuremath{\mathrm{d}\mathcal{L}_{\mu\rationals}^-}\xspace}
\newcommand{\dGLrat}{\ensuremath{\dGL_{\rationals}}\xspace}

\newcommand{\provdlmu}{\provlog{\mathrm{d}\kern-1pt\mathcal{L}_{\kern-1pt\mu}}}
\newcommand{\provdl}{\provlog{\dL}}
\newcommand{\provdgl}{\provlog{\dGL}}
\newcommand{\provdglplus}{\provlog[\dL]{\dGL}}

\newcommand{\rationaldifferentialgamelogic}{rational differential game logic\xspace}
\newcommand{\Rationaldifferentialgamelogic}{Rational differential game logic\xspace}

\newcommand{\valdlmu}{\vDash_{\dLmu}}
\newcommand{\valdgl}{\vDash_{\dGL}}
\newcommand{\smallGame}{{\scriptscriptstyle\Game}}

\newcommand{\synlie}[2]{\mathcal{L}_{#1}(#2)}
\newcommand{\synliedouble}[1]{\mathcal{L}^2_{#1}}

\newcommand{\syntacticboundsnoevd}[1][\odef]{\mathsf{B}_{#1}}
\newcommand{\syntacticbounds}[2][\odef]{\mathsf{B}_{#1,#2}}
\newcommand{\growthbounds}[2][\odef]{\mathsf{G}_{#1,#2}}
\newcommand{\synsplitmu}[2][\odef]{{\mathrm{S}}^\mu_{#1,#2}}
\newcommand{\synsplitgame}[1][\odef]{{\mathrm{S}}^\smallGame_{#1,\evdfml}}
\newcommand{\synsplitgamenolabel}[1][\odef]{{\mathrm{S}}_{#1,\evdfml}}
\newcommand{\synsplitmunolabel}[2][\odef]{{\mathrm{S}}_{#1,#2}}

\newcommand{\midpointseq}{\bar{u}}

\newcommand{\synreachrelsymbol}{{\mathrm{R}}}

\newcommand{\nsynreachrelex}[3][]{\synreachrelsymbol^{#1}_{#2,#3}}

\newcommand{\nsynapproxreachrel}[3][]{\synreachrelsymbol^{#1}_{#2,#3}}
\newcommand{\nsynreachapproxrelnobound}[2][]{\synreachrelsymbol^{#1}_{#2}}
\newcommand{\synlipbd}[2][\odef]{\mathrm{Lip}_{#1,#2}}
\newcommand{\synsolex}[2][\odef]{\mathrm{Ev}_{#1,#2}}
\newcommand{\synapprox}[3][]{\mathrm{Ex}_{#2,#3}^{#1}}
\newcommand{\fmlepsminus}[3][\odevarsx]{#2_{#1\pm#3}}

\newcommand{\conderbounds}{\kappa}

\newcommand{\dGLheadline}{\texorpdfstring{\dGL}{Differential Game Logic}}
\newcommand{\dLmuheadline}{\texorpdfstring{\dLmu}{Differential mu-Calculus}}

\newcommand{\generalparamname}{global\xspace}
\newcommand{\localparamname}{local\xspace}

\newcommand{\pref}[1]{{\ref{#1})}}

\newcommand{\fostatessubsetafterdas}[3][\fostatessubset]{#1_{#2}^{#3}}
\newcommand{\fostatessubsetafterdasp}[3][\fostatessubset]{(#1)_{#2}^{#3}}

\newcommand{\ciff}{\quad\text{iff}\quad}

\newcommand{\fostatessubsetboundeff}[3][\fostatessubset]{#1\cap\fostatesame{#2}{#3}}

\newcommand{\fostatesame}[2]{I_{#1}^{#2}}

\newcommand{\fostatessubsetcomp}{\fostatessubset^c}

\newcommand{\contextaxiomname}{context\xspace}
\newcommand{\excontextaxiomname}{extended context\xspace}

\newcommand{\iref}[1]{{\ref{#1})}}

\newcommand{\realat}[3]{\mathsf{at}(#1,#2,#3)}

\newcommand{\ivarnat}{n}

\newcommand{\lforallnat}[1]{\lforall{{#1{\in}\naturals}}\,}

\newcommand{\lexistsnat}[1]{\lexists{#1{\in}\naturals}\,}

\newcommand{\finps}{\mathsf{PS}}
\newcommand{\finpspred}[1]{#1\in\mathsf{PS}}

\newcommand{\lforallfinps}[1]{\lforall{{#1{\in}\finps}}\,}

\newcommand{\elof}[2]{{#1}_{({#2})}}

\newcommand{\loglfplfp}[4]{[l\mathsf{F}_{#1,#2}#3](#4)}

\newcommand{\LFP}[1][\fosig]{\ensuremath{\mathsf{LFP}}\xspace}

\newcommand{\lfptofml}[2][\mu]{{#2}^{{{#1}}}}
\newcommand{\lfptofmlp}[2][\mu]{\lfptofml[#1]{(#2)}}

\newcommand{\lfpden}[2][\fonstr]{#1\envelope{#2}}

\newcommand{\allexp}[2][\ivarbseq]{{#2}^{#1}}
\newcommand{\allexpp}[2][\ivarbseq]{{(#2)}^{#1}}

\newcommand{\ivarforpvar}[1]{\chi_{#1}}
\newcommand{\sharpivarforpvar}[1]{\xi_{#1}}

\newcommand{\predof}[1]{\check{#1}}
\newcommand{\predofreal}[1]{\predof{#1}}
\newcommand{\predofivar}[1]{\predof{\chi}_{#1}}

\newcommand{\reccode}[1]{\mathsf{T}_{#1}}

\newcommand{\progcode}{\sigma}

\newcommand{\progcodequote}[1]{\raisebox{0.2ex}{$\ulcorner$}\kern-2pt#1\kern-2pt\raisebox{0.2ex}{$\urcorner$}}
\newcommand{\applyprogcodetoprog}[2][\sigma]{\progseqeval{#1}#2}
\newcommand{\applyprogcodetoprogdia}[2][\sigma]{\progseqevaldia{#1}#2}

\newcommand{\forallprogcode}[1]{\lbox{#1}}
\newcommand{\forallprogcodep}[1]{\forallprogcode{(#1)}}
\newcommand{\forsomeprogcode}[1]{\ldiamond{{#1}}}
\newcommand{\forsomeprogcodep}[1]{\forsomeprogcode{(#1)}}

\newcommand{\upcorners}[1]{\raisebox{0.2ex}{$\ulcorner$}\kern-2pt#1\kern-2pt\raisebox{0.2ex}{$\urcorner$}}
\newcommand{\locorners}[1]{\raisebox{-0.4ex}{$\llcorner$}\kern-2pt#1\kern-2pt\raisebox{-0.4ex}{$\lrcorner$}}

\newcommand{\progseq}{\bar{a}}
\newcommand{\emptyprogseq}{\emptyset}

\newcommand{\progseqeval}[1]{\dbox{#1}}
\newcommand{\progseqevaldia}[1]{\ddiamond{#1}}

\newcommand{\finseq}{u}
\newcommand{\finseqlen}[1][\finseq]{\abs{#1}}

\newcommand{\existsatmost}[2][1]{\exists^{\leq#1}#2{\;}}
\newcommand{\existsunique}[1]{\exists!#1{\;}}

\newcommand{\selection}{replacement\xspace}

\newcommand{\splitilength}[1][m]{\Delta^{-1}_{#1}}
\newcommand{\scalesplitilength}[1][m]{\floor{\splitilength[#1]s}}
\newcommand{\splitilengthinv}[1][m]{\Delta_{#1}}
\newcommand{\splitilengthdoubinv}[1][m]{\Delta_{#1}^{2}}

\newcommand{\setepsminus}[2]{#1_{-#2}}

\newcommand{\fmlcl}[2][\odevarsx]{\mathrm{cl}_{#1}({#2})}
\newcommand{\fmlint}[2][\odevarsx]{\mathrm{int}_{#1}({#2})}

\newcommand{\timevar}{\mathrm{t}}
\newcommand{\timecalcvar}{\tau}

\newcommand{\odevarsx}{\ivarseq}
\newcommand{\odevarsb}{\ivarbseq}
\newcommand{\odevarsc}{\bar{z}}

\newcommand{\odef}{\mathrm{F}}
\newcommand{\odefof}[1][\odevarsx]{\odef(#1)}
\newcommand{\odefofsth}[2][\odevarsx]{#2(#1)}
\newcommand{\odefi}[1]{\odef_{#1}}

\newcommand{\odefinterp}{\envelope{\odef}}
\newcommand{\evdfmlinterp}{\envelope{\evdfml}}

\newcommand{\stdode}{\godeev{\odevarsx}{\odefof}{\evdfml}}

\newcommand{\stdodewo}{\gode{\odevarsx}{\odefof}}
\newcommand{\stdodewith}[1]{\godeev{\odevarsx}{\odefof}{#1}}

\newcommand{\stdodet}{\godetev{\odevarsx}{\odefof}{\evdfml}}
\newcommand{\stdodetwith}[1]{\godetev{\odevarsx}{\odefof}{#1}}
\newcommand{\stdodetbdd}{\godetev{\odevarsx}{\odefof}{\evdfml\land \timevar{\leq} T}}

\newcommand{\stdodetclbdd}{\godetev{\odevarsx}{\odefof}{\fmlcl{\evdfml}\land\timevar{\leq} T}}
\newcommand{\stdodetintbdd}{\godetev{\odevarsx}{\odefof}{\fmlint{\evdfml}\land\timevar{\leq} T}}
\newcommand{\stdodetwo}{\godet{\odevarsx}{\odefof}{}}

\newcommand{\stdoderat}{\godeev{\odevarsx(\timevar)}{\odefof}{\evdfml}_{\upharpoonright\rationals}}
\newcommand{\stdoderatwo}{\gode{\odevarsx(\timevar)}{\odefof}_{\upharpoonright\rationals}}

\newcommand{\xel}{x}
\newcommand{\yel}{y}

\newcommand{\varisseq}[2]{#1{=}\upcorners{#2}}
\newcommand{\seqasvar}[2]{{\gassign{#1}{\locorners{#2}}}}
\newcommand{\seqassigntovar}[2]{{\gassign{#1}{\upcorners{#2}}}}
\newcommand{\varincode}[2]{#1{\in}#2}

\usepackage{tikz}
\usetikzlibrary{positioning}

\AtEndPreamble{%
    \usepackage{cleveref}%
    \crefname{enumi}{}{}
    \Crefname{enumi}{}{}%
    \usepackage[prefixflatinterpret,bracketmodalinterpret,fixformat,silentconst,sidenotecalculus,longseqcontext,seqoptional,boldmquant]{logic}%
    \usepackage[prefixflatinterpret,bracketmodalinterpret,fixformat,silentconst,differentialdL,simplenames]{dL}

    \definecolor{dark-gray}{rgb}{.35,.35,.35}
    \renewcommand*{\irrulename}[1]{\text{\textcolor{dark-gray}{\upshape#1}}}
    \renewcommand{\linferenceRuleLabelPrefix}{}%
    \renewcommand{\linferenceRuleNameSeparation}{\hspace{1em}}%
    \renewcommand{\ltrue}{\top}
    \renewcommand{\lfalse}{\bot}
    \renewcommand{\mand}{{\,}{\&}{\,}}
    \newcommand{\limplied}{\leftarrow}
    \newcommand{\univaxsymbol}{C}
    \hypersetup{%
        colorlinks=true,
    }%
}

\begin{document}

\begin{abstract}
    Games and fixpoints are unified by showing that first-order game logic $\fogls$ and the first-order modal $\mu$-calculus $\folmus$ are provably equiexpressive and equivalent, thereby fully aligning their expressive and deductive power.
    That is, there is a semantics-preserving translation from $\fogls$ to $\folmus$, and vice versa.
    Both translations are provability-preserving, while equivalences with roundtrip translations are provable in both calculi.
    This is to be contrasted with the propositional case, where game logic is strictly less expressive than the modal $\mu$-calculus (without adding sabotage games).
 
    The extensions with differential equations, differential game logic (\dGL) and differential modal $\mu$-calculus, are also proved equiexpressive and equivalent.
    Moreover, as the continuous dynamics are definable by fixpoints or via games, ODEs can be axiomatized completely and, as a consequence, infinitesimally robust properties of ODEs can be decided via proof search.
    Rational gameplay provably collapses the games into single player games to yield a strong arithmetical completeness theorem for \dGL with rational-time ODEs.
\end{abstract}
\maketitle

\keywords{Modal-logic,Mu-calculus,Expressiveness,Differential equations,Fixpoint,Game Logic,Completeness}

\pratendSetLocal{category=proofs}

%########################################################

\section{Introduction}

Whereas the core challenges in the analysis and verification of \emph{discrete} computer programs are loops and recursion, the analysis and verification of \emph{simultaneously discrete and continuous} cyber-physical systems introduce the additional complexity of reasoning about
differential equations.
While repetition in program behavior \cite{DBLP:journals/computing/Clarke79} has long been understood in terms of fixpoints and games to great effect, the continuous dynamics of physical systems, modeled with differential equations, are difficult to handle using discrete computation and symbolic reasoning.
This article both unifies the fixpoint and game perspectives for discrete programs and extends them to cover the continuous dynamics through both a fixpoint and a game understanding of differential equation reachability.
The complex primitives of fixpoints, adversarial behavior and differential equations are shown to be naturally closely related and to provide complementary perspectives.

The contexts of this article are the two general-purpose, first-order interpreted program logics: the $\mu$-calculus as a fixpoint logic, and game logic.
It is shown that, despite their different roots, their characterizations via fixpoints and via games are fundamentally the same.
That is, first-order game logic $\fogls$ and the first-order modal $\mu$-calculus $\folmus$ are equiexpressive and equivalent\footnote{%
    This is in contrast to the propositional case \cite{DBLP:conf/focs/Parikh83}, where expressiveness has a subtle but wide gap \cite{DBLP:journals/mst/BerwangerGL07}, and completeness, despite significant attention \cite{DBLP:conf/lics/EnqvistHKMV19}, remains elusive \cite{kloibhofer2023note}, so that the addition of sabotage games is needed to complete game logic and establish equivalence and completeness \cite{DBLP:conf/lics/WafaP24}.
}, which completely aligns both their expressive power and their deductive power.
This is established via semantics- and provability-preserving translations that support provable roundtrip translations.
The syntactic provability of the correctness of these translations lifts the semantic equiexpressiveness proofs to complete syntactic proof transfers.
Every proved property of one syntactically lifts to a proved property of the other.
As a consequence, relative completeness theorems for one fragment readily transfer to the other.

The general perspective in this paper supports general atomic game symbols, which makes it possible to lift these findings to the presence of various additional dynamics.
Discrete, continuous, and adversarial dynamics are all shown to be fixpoints.
Most interestingly, differential equations can be characterized coinductively by fixpoints and, thus, also by games, providing a \emph{global} and \emph{discrete} perspective on differential equations.
Instead of approximating differential equations by a sequence of steps, whose precision is linear in the time step, the fixpoint provably approximates the function globally and symmetrically.
This nondeterministic dynamic definition of continuous reachability allows for approximations with precision exponential in the number of fixpoint iterations.
Beyond the theoretical interest, the axiomatization enables strong (absolute) completeness results for systems where safety is not contingent on infinitesimal behavior on the boundary.
Moreover, the truth of robust properties is decidable via proof search, which can be used to automate verification.

The equivalence of games and fixpoints via logic extends to show that differential game logic \dGL (with discrete, continuous, and adversarial dynamics) \cite{DBLP:journals/tocl/Platzer15} is equiexpressive and equivalent to the differential $\mu$-calculus \dLmu (with discrete and continuous dynamics).
For computer programs described in Hoare calculus \cite{DBLP:journals/cacm/Hoare69}, Cook showed that the proof calculus is relatively complete \cite{DBLP:journals/siamcomp/Cook78}.
Harel extended this to show the arithmetic completeness of interpreted dynamic logic \cite{DBLP:conf/icalp/Harel78} by an equivalent semantic reduction to the assertion language.
For the more complex adversarial and continuous dynamics, this approach cannot work to show arithmetic completeness of \differentialgamelogic, as it has significant additional expressive power \cite{DBLP:journals/tocl/Platzer15}.
However, when the choices of the players are restricted to rational values, $\dGL$ is shown to be equiexpressive and relatively complete.
This means games with \emph{rational} choices, in contrast to general games, do \emph{not add} complexity to the game-free version.
To summarize, while \differentialgamelogic is conceptually significantly richer than \differentialdynamiclogic, with surprisingly small restrictions, such as robustness for continuous properties or rational choices for games, neither the adversarial nor the continuous dynamics fundamentally increase the expressive and deductive complexity.

\subsubsection*{Contributions}

First-order game logic, the first-order \(\mu\)-calculus, and their abstract \gamesymbol{s} are introduced.
In contrast to their propositional counterparts, both logics are proved to be logically equivalent, and the fixpoint variable hierarchy of the first-order $\mu$-calculus is shown to collapse.
The proof techniques are general and show the power of provable roundtrip translations between logics and propositional reductions of first-order questions to purely semantic propositional properties.

The general theory is applied to prove new theoretical properties of \differentialgamelogic.
Via a fixpoint axiomatization of ODEs, the \(\mu\)-calculus correspondence is used to prove equiexpressiveness of \differentialgamelogic to its ODE-free fragment.
This provides a new way for continuous dynamics in hybrid games to be handled completely.
Moreover, the fixpoint characterization is shown to have practical value, as it delivers practical completeness and decidability of robust safety and reachability properties of differential equations.
Although adversarial behavior of \dGL significantly increases its expressiveness over the non-adversarial fragment \dL \cite{DBLP:conf/lics/Platzer12b,DBLP:journals/tocl/Platzer15}, this difference is shown to vanish when restricted to rational play.
This shows that games and fixpoints interact differently with the uncountable than with the countable.

\begin{table}[htbp]\newcommand{\tablequiv}{\color{dark-gray}\hfill\(\sim\)\hspace{0.5em}}%
    \newcommand{\tablenotquiv}{\color{dark-gray}\hfill\(\not\sim\)\hspace{0.5em}}
    \centering
    \begin{tabular}{p{3.1cm} ll p{7.05cm}}
        \toprule
         & \textbf{Fixpoint} & \textbf{Game} & \textbf{Notes} \\
        \midrule
        Propositional
        \newline
        \textcolor{dark-gray}{(\cite{DBLP:journals/tcs/Kozen83,DBLP:conf/focs/Parikh83,DBLP:conf/lics/WafaP24})}
                         & \(\plmus\)\tablenotquiv
                         & \(\pgls\)
                      & Interpreted over neighborhood structures \newline
                      \(\plmus\) more expressive than \(\pgls\) \\
        \addlinespace
        First-order
        \newline
        \textcolor{dark-gray}{\Cref{sec:fo-gl-and-lmu} -- \Cref{sec:equiexpressivenessandequivalence}}
                         & \(\folmus[\gsig]\)\tablequiv
                         & \(\fogls[\gsig]\) 
                        & Extend propositional version to first-order \newline
                        Parametric in signature \(\gsig\)
                         \\
        \addlinespace
        Differential
        \newline
        \textcolor{dark-gray}{\Cref{sec:differential-equations-as-fixpoints} -- \Cref{sec:completion-for-robust-properties}}
                         & \(\dLmu\)\tablequiv 
                         & \(\dGL\) 
                                                 & Instantiate first-order cases for signature \(\dglsig\)
                         \newline
                         Interpreted with fixed semantics over \(\rfnstr\)
                         \\
                         \addlinespace
                         Differential Dynamic
                         \newline
        \textcolor{dark-gray}{\Cref{sec:rational-completeness}}
                         & \(\systemicdlmus\)\tablequiv
                         & \(\dL\)
                         & Restrict \(\dLmu\) and \(\dGL\) to non-adversarial interaction \\
                         \addlinespace
                         Differential Rational
                         \newline
        \textcolor{dark-gray}{\Cref{sec:rational-completeness}}
                                          & \(\dLmurat\)\tablequiv
                                          & \(\dGLrat\)
                                          & Restrict \(\dLmu\) and \(\dGL\) to rational choices
                                          \newline
                                          Equiexpressive with \(\systemicdlmus\) and \(\dL\)
                                          \\
        \bottomrule
    \end{tabular}
    \caption{Overview of the considered fixpoint and game logics. Propositional \(\mu\)-calculus and game logic are not equiexpressive as indicated by \(\not\sim\). The first order variants are equiexpressive and equivalent as indicated by \(\sim\). }\label{tab:logic-fixpoints-games}
\end{table}

\subsubsection*{Outline}

In \Cref{sec:logics}, the notions of general expressive power and (proof-theoretic) equivalence between logics are introduced.
The extension of first-order primitives with \gamesymbol{s} for state change and the first-order extensions of game logic and the \(\mu\)-calculus are formulated in \Cref{sec:fo-gl-and-lmu}.
\Cref{sec:proof-calculi} introduces proof calculi for these logics, and the equiexpressiveness (\Cref{prop:glmufoequiexp}) and equivalence (\Cref{prop:foequiv}) of \fogls and \folmus are established in \Cref{sec:equiexpressivenessandequivalence}.
\Differentialgamelogic is presented as an interpreted version of game logic, and a novel fixpoint axiomatization of ODEs (\Cref{prop:ode-as-fixpoint-theorem}) is developed in \Cref{sec:differential-equations-as-fixpoints}.
This is applied in \Cref{sec:completion-for-robust-properties} to prove completeness and decidability for robust safety and reachability properties (\Cref{prop:robust-completeness}).
Finally, the completeness of \dGL with rational-time ODEs relative to the base logic (\Cref{thm:relativecomplete}) is proved via provable equiexpressiveness (\Cref{prop:expressiveness-rational}) in \Cref{sec:rational-completeness}. \Cref{sec:related-work} discusses related work.

\subsubsection*{Overview}
For reference, \Cref{tab:logic-fixpoints-games} concisely summarizes the logics that will be introduced in this article.
In all rows, except for the first one (on propositional logics), the fixpoint logic and the corresponding game logic are equiexpressive and equivalent.
Moreover, \Cref{sec:lfp-and-folmu} relates \(\folmus\) to least fixpoint logic (\LFP), which is an important extension of first-order logic with fixpoints.

%########################################################

\section{Logics: Expressive and Deductive Power}\label{sec:logics}

Formulating the results of this paper benefits from a formal notion of equivalence of logics.
This section introduces the abstract notions of logics, semantics and translations.
An important general condition of equivalence is identified: the correctness of translations needs to be provable \emph{in the logical calculi} via roundtrip translations.
This condition is the crucial ingredient when it comes to proof transfers.

A \emph{logic} \(\logics\) is abstractly viewed as consisting of a (computable) set of formulas \(\logics\) and a proof calculus \(\provlog{\logics}\), which is abstractly viewed as a (semi-computable) reflexive, transitive provability relation \(\provlog[\fmlb]{\logics}\fml\) on formulas \(\fml,\fmlb\in\logics\).
It is also assumed that the set of formulas of a logic is closed under propositional connectives and that the proof calculus proves all propositional tautologies and admits the modus ponens proof rule:
\[
    \cinferenceRule[mp|MP]{modus ponens}
    {
        \linferenceRule[sequent]
        {\fml & \fml\limply\fmlb}
        {\fmlb}
    }{}
\]
That is, \(\provlog[\fmlc]{\logics}{\fmlb}\) for any formula \(\fmlc\) with \(\provlog[\fmlc]{\logics}{\fml}\) and \(\provlog[\fmlc]{\logics}{\fml\limply\fmlb}\).
Write \(\provlog{\logics}\fml\) for \(\provlog[\ltrue]{\logics}\fml\).

The \singularsemantics of a logic may depend on parameters.
For example, in model-theoretic semantics of first-order logic, a formula can be interpreted in a structure for different values of variables, which can be viewed as parameters.
Parameters can be \emph{\generalparamname} or \emph{\localparamname}.
In first-order logic, the structure is usually viewed as a \generalparamname parameter and variable assignment as \localparamname parameters.
Similarly, in modal logic, Kripke models are typically \generalparamname parameters, while states are \localparamname parameters.
As these examples show, \localparamname parameters can depend on the values of the \generalparamname parameters.
Formally, a (denotational) \emph{\singularsemantics} \(\semantics\) of \(\logics\) with \generalparamname parameter set \(\semparset\) and \localparamname parameter sets \(\semparsetg{\sempar}\) for \(\sempar\in\semparset\) is of the type
\[
    \semantics:\prod_{\sempar\in\semparset}(\logics\to\pow{\semparsetg{\sempar}}),
\]
where \(\pow{\semparsetg{\sempar}}\) denotes the power set of the \localparamname parameter space.
Write \(\gensemb[\semantics]{\sempar}{\fml}\) for \(\semantics(\sempar)(\fml)\).
For \(\sempar\in\semparset\) and \(\semparb\in\semparsetg{\sempar}\) write \(\gensem[\semantics]{\sempar}{\semparb}{\fml}\) if \(\semparb\in\gensemb[\semantics]{\sempar}{\fml}\).
If \(\gensem[\semantics]{\sempar}{\semparb}{\fml}\) holds for all \(\semparb\in\semparsetg{\sempar}\), write \(\gensemsemival[\semantics]{\sempar}{\fml}\).
Moreover, if \(\gensemsemival[\semantics]{\sempar}{\fml}\) holds for all \(\sempar\in\semparset\), write \(\gensemval[\semantics]{\fml}\) and say \(\fml\) is \emph{valid} for \(\semantics\).
A logic \(\logics\) is
\begin{enumerate}
    \item (globally) \emph{sound} with respect to $\semantics$ if
          \(\gensemval[\semantics]{\fml}\) whenever \(\provlog[\fmlb]{\logics}\fml\) and \(\gensemval[\semantics]{\fmlb} \).
    \item \emph{locally} sound with respect to $\semantics$ if
          \(\provlog[\fmlb]{\logics}\fml\) implies \(\mforall{\sempar\in\semparset}( \gensemsemival[\semantics]{\sempar}{\fmlb}\mimply \gensemsemival[\semantics]{\sempar}{\fml}).\)
\end{enumerate}
Note that local soundness implies soundness.
The choice of which parameters are \localparamname and which are \generalparamname is crucial to the notion of local soundness.
The logic \(\logics\) is said to be \emph{complete} with respect to $\semantics$ if \(\gensemval[\semantics]{\fml}\) implies \(\provlog{\logics}{\fml}\).

Consider two logics \(\logicsa,\logicsb\) with the same \generalparamname and \localparamname parameters.
A \emph{translation} \(\gentranslations:\logicsa\to\logicsb\) from \(\logicsa\) to \(\logicsb\) is a function mapping each formula \(\fml\) of \(\logicsa\) to a formula \(\gentranslation{\fml}\) of \(\logicsb\).
A translation \(\gentranslations\) is \emph{sound} with respect to \(\sempar\in\semparset\) and the respective semantics \(\semanticsc\) of \(\logicsa\) and \(\semanticsd\) of \(\logicsb\) if \(\gensemb[\semanticsc]{\sempar}{\fml}=\gensemb[\semanticsd]{\sempar}{\gentranslation{\fml}}\) for all \(\fml\).
Equivalently,  \(\gentranslations\) is \emph{sound} if \(\gensem[\semantics]{\sempar}{\semparb}{\fml}\) exactly if \(\gensem[\semanticsd]{\sempar}{\semparb}{\gentranslation{\fml}}\).
If such a sound translation exists, \(\logicsb\) is said to be \emph{at least as expressive} over \(\sempar\) as \(\logicsa\).
If, in addition, \(\logicsa\) is at least as expressive as \(\logicsb\)  over \(\sempar\), then \(\logicsa\) and \(\logicsb\) are said to be \emph{equiexpressive over \(\sempar\)}.
The logics \(\logicsa\) and \(\logicsb\) are \emph{equiexpressive} if they are equiexpressive over all \generalparamname parameters.
The choice of which parameters are local and which are \generalparamname is crucial to the notion of equiexpressiveness.

Besides their expressiveness, the \emph{deductive power} of the proof calculi is of interest.
The following definition captures the concept that two logics have the same deductive power.
\begin{definition}\label{def:logics-equivalent}
    An \emph{equivalence} of two logics \(\logicsa,\logicsb\) is a pair of translations \(\gentranslations:\logicsa\to\logicsb\) and \(\gentranslationsb:\logicsb\to\logicsa\) such that
    \begin{enumerate}
        \item \(\provlog{\logicsa}\fml \;{\mimply}\; \provlog{\logicsb}\gentranslation{\fml}\) and
              \(\provlog{\logicsb}\fmlb \;{\mimply}\; \provlog{\logicsa}\gentranslationb{\fmlb}\) (proof translation)\label{prop:logicequivalencethereonly}
        \item \(\provlog{\logicsa}\fml\lbisubjunct\gentranslationab{\fml}\)
              and \(\provlog{\logicsb}\fmlb\lbisubjunct\gentranslationba{\fmlb}\) (provable roundtrip)\label{prop:logicequivalencethereandback}
    \end{enumerate}
\end{definition}

The first condition (proof translation) \iref{prop:logicequivalencethereonly} requires that both logics prove the same formulas, up to translation.
It is possible that the proofs for \(\fml\) and \(\gentranslation{\fml}\) are very different. For example, in the case of sabotage game logic and the modal \(\mu\)-calculus \cite{DBLP:conf/lics/WafaP24}, the length of the proof of the translation to the \(\mu\)-calculus can be non-elementary in the length of the original game proof, while the reverse embedding into games is linear.

The second condition (roundtrip) \iref{prop:logicequivalencethereandback} requires that the two translations are inverse to each other up to (provable) logical equivalence.
This crucially ensures that the correctness of the translation depends \emph{only} on the proof calculus and makes it possible to transfer all relative provability properties.
The role of the two conditions is best illustrated by the useful consequence of equivalence that completeness can be transferred.

\begin{theoremE}[Completeness Transfer][normal]\label{prop:equivalence-transfers-completeness}
    Let \(\gentranslations,\gentranslationsb\) be an equivalence between \(\logicsa\) and \(\logicsb\), and let \(\gentranslationsb\) be sound for \(\semanticsd\) and~\(\semanticsc\).
    If \(\logicsa\) is complete for \(\semanticsc\), then \(\logicsb\) is complete for \(\semanticsd\).
\end{theoremE}

\begin{proofE}
    Suppose \(\gensemval[\semanticsd]{\fml}\), so that \(\gensemval[\semanticsc]{\gentranslation[\gentranslationsb]{\fml}}\) as \(\gentranslationsb\) is a sound translation.
    By completeness of \(\logicsa\) it follows that \(\provlog{\logicsa}\gentranslation[\gentranslationsb]{\fml}\) is provable.
    By condition \iref{prop:logicequivalencethereonly} this transfers to \(\provlog{\logicsb}\gentranslation[\gentranslationsb\gentranslations]{\fml}\).
    Since \(\provlog{\logicsb}\gentranslation[\gentranslationsb\gentranslations]{\fml}\limply\fml\) by the roundtrip condition~\iref{prop:logicequivalencethereandback}, modus ponens \irref{mp} derives \(\provlog{\logicsb}\fml\).
\end{proofE}

\begin{remark}
    Equivalence of logics is closely related to equivalence of categories.
    A logic \(\logicsa\) can be viewed as a category with the formulas of \(\logicsa\) as objects such that there is a unique arrow from \(\fml\) to \(\fmlb\) if and only if \(\provlog{\logicsa}\fml\limply\fmlb\).
    A translation is \emph{functorial} if it preserves implications, i.e., \(\provlog{\logicsb}\gentranslation{\fml}\limply\gentranslation{\fmlb}\) if \(\provlog{\logicsa}\fml\limply\fmlb\).
    An equivalence of two logics consisting of functorial translations is an equivalence of the two corresponding categories and vice versa.
    This notion of equivalence is \emph{local}.
    Moreover, functorial equivalences can be composed.

    Alternatively, by defining the category so that there is an arrow from \(\fml\) to \(\fmlb\) if and only if \(\provlog[\fml]{\logicsa}\fmlb\), the resulting notion of equivalence would be \emph{global}.
    For logics satisfying the deduction theorem (if \(\provlog[\fml]{\logicsa}{\fmlb}\) then \(\provlog{\logicsa}{\fml\limply\fmlb}\)) and modus ponens, the local and global equivalence notions coincide.
    However, in the absence of the deduction theorem, such as in most modal logics, the notion of local equivalence is more interesting.
\end{remark}

%########################################################

\section{First-Order Game Logic and \mutex-Calculus}\label{sec:fo-gl-and-lmu}

\Firstordergamelogic and the \firstordermucalculus add dynamics to ordinary first-order logic in two ways.
On the one hand, they add \emph{atomic games}, which are a very general notion of state change, generalizing the quantifiers of first-order logic, which have a limited form of single variable universal or existential state change.
On the other hand, both logics allow for new (and different) combinations of such state changes, through fixpoints or games, respectively.

General atomic games, which \firstordergamelogic and the \firstordermucalculus have in common, are introduced by extending a first-order signature and providing game semantics in \Cref{sec:action-symbols}.
The dynamics of fixpoints and  games are introduced in \Cref{sec:fo-gl,sec:fo-lmu}, respectively.
Particular atomic games for assignment are discussed and related to the usual first-order logic quantifiers in \Cref{sec:quantifier-and-assignment-actions}.
The subtleties of substitution are discussed in \Cref{sec:free-and-bound}.

\subsection{\GameSymbol Dynamics}\label{sec:action-symbols}

In \firstordergamelogic it is natural to start from \emph{atomic game} primitives.
This is in contrast to other first-order modal logics, where the basic modalities are interpreted as relations~\cite{Fitting1998-MELFML}.
Defining atomic games is more subtle, as the interactive nature of the gameplay needs to be considered.
At the same time, the variables that the game depends on and affects need to be discernible.
For this purpose, \gamesymbol{s} abstract from the particular interpretation of an atomic game just as proposition and function symbols abstract propositions and functions in first-order logic.

\subsubsection{First-order Modal Signature}

{Fix a set \(\ivars\) of \ivarname{s}.}  
A  \emph{\fosignature} consists of
function symbols~\(\funcsymb\) and relation symbols~\(\relsymb\) with fixed arities.
A \emph{\gamesignature}~\(\gsig\) is a \fosignature together with a list of \gamesymbol{s}~\(\gamesymb\) with fixed \emph{pairs} of arities.
An \gamesymbol is an abstract representation of a game and the arities refer to the number of variables that the game depends on and the number of variables that the game affects.
As usual, first-order \(\gsig\)-terms \(\term\) and \emph{atomic} \(\gsig\)-formulas \(\atfml\) are defined by the grammar:
\begin{align*}
    \term
     & \grassign \ivar \| \funcsymb (\term_1,\ldots,\term_n)
    \\
    \atfml
     & \grassign \term_1 = \term_2 \| \relsymb(\term_1,\ldots,\term_n)
\end{align*}%
where \(\ivar\) is an \ivarname{}, \(\funcsymb\) is an \(n\)-ary \(\gsig\)-function symbol and \(\relsymb\) is an \(n\)-ary \(\gsig\)-relation symbol.
The notation \(\termseq,\termbseq\) is used for finite sequences of terms of the form \(\term_1,\ldots,\term_\ell\) and \(\ivarseq,\ivarbseq\) stand for sequences of pairwise distinct variables \(\ivar_1,\ldots,\ivar_k\).

\subsubsection{Effectivity Functions}

The semantics of \gamesymbol{s} requires some basic definitions.
Let \(\Gamefuncs{X}\) be the set of monotone (with respect to \(\subseteq\)) functions \(\gamefunc:\pow{X}\to\pow{X}\) called \emph{effectivity functions} \cite{DBLP:conf/lics/EnqvistHKMV19}.
Effectivity functions assign to a goal region \(W\subseteq X\) the set of states \(\gamefunc(W)\) from which \Playerone can win a given game with the goal of reaching a state in \(W\).
The monotonicity condition requires that if \Playerone can win the game with the goal \(W\) from a state, then \Playerone can, of course, also force the game into any superset \(V\supseteq W\) of \(W\) from the same state.

\begin{definition}
    For effectivity functions \(\gamefunc,\gamefuncb \in \Gamefuncs{X}\) define
    \begin{enumerate}
        \item the pointwise \emph{union} \((\gamefunc\cup\gamefuncb)(\fostatessubset) = \gamefunc(\fostatessubset)\cup\gamefuncb(\fostatessubset)\),
        \item the \emph{dual} \(\setdual{\gamefunc}(\fostatessubset) = \setcomplement{X}{\gamefunc(\setcomplement{X}{\fostatessubset})}\) and
        \item the \emph{least (pre-)fixpoint} \(\mlfp{\fostatessubset}{\gamefunc(\fostatessubset)} = \Intersection\{\fostatessubset\subseteq X:\gamefunc(\fostatessubset)\subseteq \fostatessubset\}.\)
    \end{enumerate}
\end{definition}
Note that \(\mlfp{\fostatessubset}{\gamefunc(\fostatessubset)}\) is in fact the least fixpoint of \(\gamefunc\) by the Knaster-Tarski theorem \cite{DBLP:journals/pjm/Tarski55}, since \(\gamefunc\) is monotone.
Effectivity functions are closely related to \emph{monotone neighborhood functions}, which are functions \(\nfunc:X\to\pow{\pow{Y}}\) such that every \(\nfunc(x)\) is an \emph{upward closed} family of sets (i.e., if \(x\in X\), \(W\in\nfunc(x)\) and \(V\supseteq W\) then \(V\in\nfunc(x)\)).
Let \(\Nfuncsmix{X}{Y}\) be the set of monotone neighborhood functions \(\nfunc:X\to\pow{\pow{Y}}\).
In terms of games, a monotone neighborhood function can be viewed as a function assigning to an initial state the set of all sets into which \Playerone can force the game to go.
In this way, neighborhood functions and effectivity functions are naturally equivalent.

\subsubsection{\FonStructure{s}}

A first-order neighborhood structure is a first-order structure with additional interpretations of the \gamesymbol{s} in the signature.
Similarly to \(n\)-ary predicate symbols \(\relsymb\), which appear in the form \(\relsymb(\termseq)\) in a formula, an \klary{k}{\ell} \gamesymbol{} \(\gamesymb\) may appear in the form \(\gamesymbat[\gamesymb]{\ivarseq}{\termseq}\) in a formula, where \(\ivarseq\) is of length \(k\) and \(\termseq\) is of length \(\ell\).
Intuitively this can be read as the game \(\gamesymb\) being played with parameters \(\termseq\) and affecting the values of the variables~\(\ivarseq\).
A structure interprets a \gamesymbol \(\gamesymb\) as a monotone neighborhood function \(\fonstrdom^\ell\to \pow{\pow{\fonstrdom^k}}\), which assigns to all parameter values the set of those (\(\gamesymb\)-achievable) sets into which \Playerone can force the game to go.

\begin{definition}
    An \emph{\(\gsig\)-\fonstructure~\(\fonstr\)} consists of a non-empty domain~\(\fonstrdom\) and interpretations for the \(\gsig\)-symbols:
    \begin{enumerate}
        \item \(\fonstrint{\funcsymb}:\fonstrdom^n\to\fonstrdom\) for \(n\)-ary \(\gsig\)-function symbols \(\funcsymb\),
        \item \(\fonstrint{\relsymb}\in\pow{\fonstrdom^n}\) for \(n\)-ary \(\gsig\)-relation symbols \(\relsymb\) and
        \item \(\fonstrint{\gamesymb}\in\Nfuncsmix{\fonstrdom^{\ell}}{\fonstrdom^k}\) for \klary{k}{\ell} \(\gsig\)-\gamesymbol{s}~\(\gamesymb\).
    \end{enumerate}
\end{definition}

A state \(\fostate\) is a function \(\ivars\to\fonstrdom\) assigning values to \ivarname{s} and~\(\fostates\) denotes the set of all states.
The notation \(\strestrvar{\fostate}{\ivarseq}=(\fostate(\ivar_1), \ldots,\fostate(\ivar_k))\in \fonstrdom^k\) is the restriction of the state \(\fostate\) to variables from \(\ivarseq\).
For a tuple \(\pstatepre=(\pstatepreel{1},\ldots,\pstatepreel{k})\in\fonstrdom^k\) and a sequence \(\ivarseq=(\ivar_1,\ldots,\ivar_k)\) of \ivarname{s}, let \(\replpstate[\fostate]{\pstatepreatvar{\ivarseq}}\) be the state that agrees with \(\fostate\) everywhere, except \(\replpstate[\fostate]{\pstatepreatvar{\ivarseq}}(\ivar_i)=\pstatepreel{i}\) for \(1\leq i \leq k\).

\subsection{First-order Game Logic}
\label{sec:fo-gl}

\subsubsection{Syntax}
\label{sec:syntax}

The syntax of~\(\gsig\)-game logic (\fogls) \emph{formulas} $\fml$ and \emph{games} $\game$ is given by
\begin{align*}
    \fml
     & \grassign \atfml \| \lnot \fml \| \fml_1\land\fml_2 \| \lpossible{\game}\fml
    \\
    \game
     & \grassign \gamesymbat{\ivarseq}{\termseq} \| \gtest{\fml} \| \game_1\gachoice\game_2 \| \game_1\gcom\game_2 \| \garepeat{\game} \| \gdual{\game}
\end{align*}
where~\(\atfml\) is an atomic~\(\gsig\)-formula, \(\ivarseq\) is a \(k\)-sequence of \ivarname{s}, \(\termseq\) is an \(\ell\)-sequence of \(\gsig\)-terms and~\(\gamesymb\) is a \klary{k}{\ell}~\gamesymbol of \(\gsig\).
The logical constants \(\top,\bot\) and connectives \(\fml\lor\fmlb\) and \(\lnecessary{\game}\fml\) are definable as usual by \(\ivar=\ivar\), \(\ivar\neq\ivar\), \(\lnot(\lnot\fml\land\lnot\fmlb)\) and \(\lnot\lpossible{\game}\lnot\fml\), respectively.
Note that the usual first-order quantifiers are \emph{not} included in the basic grammar.
Just as basic equality can be viewed as merely a fixed relation symbol with fixed interpretation in ordinary first-order logic, it is convenient to introduce quantifiers as fixed \gamesymbol{s} with particular fixed interpretation in \Cref{sec:quantifier-and-assignment-actions}.
Further note that the \(\gsig\) in the notation \fogls indicates the signature.

\subsubsection{Semantics}
\label{sec:semantics-gl}

The semantics of terms, formulas and games in \(\gsig\)-first-order game logic is defined with respect to an \(\gsig\)-\fonstructure~\(\fonstr\).
The semantics of a term \(\term\) is defined as usual to denote an element \(\ltden{\term}\in\fonstrdom\) of the domain with respect to a state \(\fostate\in\fostates\).
The \emph{semantics of \fogls formulas \(\fml\)} is defined as a subset \(\glfden{\fml}\subseteq \fostates\) by induction on formulas.
As this notation indicates, the structures \(\fonstr\) are the \generalparamname parameters and the states are the \localparamname parameters of the semantics in the terminology of \Cref{sec:logics}.
For propositional connectives the definition of the semantics is as usual and for modalities \(\glfden{\lpossible{\game}\fml}\) is defined by \(\glgden{\game}(\glfden{\fml})\), where
the \emph{semantics of games \(\game\)} is defined by mutual recursion to be the \emph{monotone} function \(\glgden{\game}:\pow{\fostates}\to\pow{\fostates}\)
{%
        \allowdisplaybreaks
        \begin{aligntable}[2]
            \begin{align*}
                \nextit{ \glgden{\game\gachoice\gameb}                           }{ \glgden{\game}\cup\glgden{\gameb} }
                \nextit{ \glgden{\gtest{\fml}}(\fostatessubset)                     }{ \glfden{\fml}\cap\fostatessubset }
                \nextit{ \glgden{\game\gcom\gameb}                               }{ \glgden{\game}\circ\glgden{\gameb} }
                \nextit{ \glgden{\garepeat{\game}}(\fostatessubset)                 }{ \mlfpp{\fostatessubsetvar}{\fostatessubset\cup\glgden{\game}(\fostatessubsetvar)}}
                \nextit{ \glgden{\gdual{\game}}                                     }{ \setdual{\glgden{\game}} }
                \lastit{ \glgden{\gamesymbat{\ivarseq}{\termseq}}(\fostatessubset)  }{ \{\fostate \in \fostates : \mexists{U\in\fonstrint{\gamesymb}(\ltden[\fostate]{\termseq})}\;                   \forall{\pstatepre{\in} U}\;\replpstate[\fostate]{\pstatepreatvar{\ivarseq}}\in\fostatessubset\} }.
            \end{align*}%
        \end{aligntable}%
    }

Except for the new \singularsemantics of atomic games, most of the \singularsemantics is similar to the \singularsemantics of propositional \gamelogic \cite{DBLP:conf/focs/Parikh83,DBLP:conf/lics/WafaP24}.
The diamond formula \(\lpossible{\game}\fml\) says that \Playerone has a winning strategy in \(\game\) to achieve $\fml$ and \(\glgden{\game}\) is a winning region function, which assigns to every goal region \(\fostatessubset\subseteq\fostates\) the set of states from which player \Playerone can force the game to end in a state in \(\fostatessubset\) (or win prematurely).
In the interpretation of atomic games of the form \(\gamesymbat{\ivarseq}{\termseq}\) \Playerone can force this game into a goal region \(\fostatessubset\) starting in a state \(\fostate\in\fostates\), if there is an  \(\gamesymb\)-achievable set \(U\in\fonstrint{\gamesymb}(\ltden[\fostate]{\termseq})\) from the parameters \(\ltden[\fostate]{\termseq}\) such that any possible final state \(\replpstate[\fostate]{\pstatepreatvar{\ivarseq}}\) with \(\pstatepre\in U\) is in the goal region~\(\fostatessubset\).

Since \(\game\gachoice\gameb\) denotes \Playerone's choice between playing the game \(\game\) or \(\gameb\), \Playerone can force the game \(\game\gachoice\gameb\) into \(\fostatessubset\) exactly if she can force it into \(\fostatessubset\) in \(\game\) \emph{or} in \(\gameb\).
The dual operator \(\gdual{}\), intuitively, switches the roles of the players, so that any choice and test taken by \Playerone in \(\game\) becomes player \Playertwo's in \(\gdual{\game}\) and vice versa.
Composed games \(\game;\gameb\) are played sequentially, playing $\gameb$ after $\game$.
The winning region of \(\game;\gameb\) for the goal~\(W\) is  the winning region \(\glgden{\game}(V)\) of \(\game\), where the goal is the winning region \(V=\glgden{\gameb}(W)\) of the subsequent game \(\gameb\).
In a test game \(\gtest{\fml}\) \Playerone loses prematurely if the formula \(\fml\) is not satisfied and no state change happens if it is.
The repetition game \(\garepeat{\game}\) is played repeatedly, where \Playerone gets to choose after every round of $\game$ whether to continue, yet she loses if she never chooses to stop.

The monotonicity of the \singularsemantics of games is crucial.
Intuitively, it means that if \Playerone can win the game with the goal region \(\fostatessubset\), she can win into any larger goal region \(\fostatessubsetb\supseteq\fostatessubset\) as well.

\begin{propositionE}[][normal]\label{prop:games-monotone}
    The function \(\glgden{\game}\) of a game \(\game\) is monotone and \(\glgden{\garepeat{\game}}(\fostatessubset)\) is the least fixpoint of \(\fosubsetvar\mapsto \fostatessubset \cup \glgden{\game}(\fosubsetvar)\).
\end{propositionE}

\begin{proofE}
    Immediate by induction on \(\game\).
\end{proofE}

\subsection{\FirstorderMuCalculus}\label{sec:fo-lmu}

\subsubsection{Syntax}

{Fix a set \(\pvars\) of \pvarname{s}.}
The syntax of~\emph{\(\gsig\)-\firstordermucalculus (\(\folmus\)) formulas} is given by
\begin{align*}
    \fml \grassign \pvar \| \atfml \| \lnot \fml \| \fml_1\land\fml_2 \| \lpossible{\gamesymbat{\ivarseq}{\termseq}}\fml \| \flfp{\pvar}{\fml}
\end{align*}%
where \(\pvar\in\pvars\) is a \pvarname,~\(\atfml\) is an atomic~\(\gsig\)-formula, \(\ivarseq\) is a \(k\)-sequence of \(\gsig\)-variables, \(\termseq\) is an \(\ell\)-sequence of \(\gsig\)-terms and~\(\gamesymb\) is a \klary{k}{\ell}~\(\gsig\)-\gamesymbol.
As usual, \(\pvar\) can appear only positively in \(\fml\) when it is bound by \(\flfp{\pvar}{\fml}\) in order to ensure well-definedness of the semantics.
For a \gamesignature \(\gsig\) write \(\folmus\) for the \firstordermucalculus in the signature \(\gsig\).
Again quantifiers are introduced as specific \gamesymbol{s}.

\subsubsection{Semantics}

The \emph{semantics of formulas of the \firstordermucalculus} is defined as a set of states \(\lmden{\intp}{\fml}\subseteq\fostates\) with respect to an \interp \(\intp:\pvars\to\pow{\fostates}\).
For atomic formulas the definition is as usual and the semantics of the remaining connectives is defined recursively as follows:
\begin{aligntable}[3]
    \begin{align*}
        \nextit{\lmden{\intp}{\pvar}}{\intp(\pvar)\quad}
        \nextit{\lmden{\intp}{\fml_1\land\fml_2}}{\lmden{\intp}{\fml_1}\cap\lmden{\intp}{\fml_2}}
        \nextit{\lmden{\intp}{\flfp{\pvar}{\fml}}}{\mlfp{\fostatessubsetvar}{\lmden{\intreplaceby{\pvar}{\fostatessubsetvar}}{\fml}}}
        \nextit{\lmden{\intp}{\lnot\fml}}{\fostates\setminus\lmden{\intp}{\fml}}
        \lastit{\lmden{\intp}{\lpossible{\gamesymbat{\ivarseq}{\termseq}}\fml}}{
            \glgden{\gamesymbat{\ivarseq}{\termseq}}(\lmden{\intp}{\fml})
        }.
    \end{align*}%
\end{aligntable}%
Here \(\intreplaceby{\pvar}{\fostatessubsetvar}\) is the \interp that agrees with \(\intp\) everywhere, except that \(\intreplaceby{\pvar}{\fostatessubsetvar}(\pvar)=\fostatessubsetvar\).
If a formula does not have free \pvarname{s}, the \interp \(\intp\) is dropped.
The semantics of \(\lmden{\intp}{\flfp{\pvar}{\fml}}\) is a least fixpoint by monotonicity:

\begin{propositionE}[][normal]
    For all \folmus formulas \(\fml\) the map \(\fostatessubsetvar\mapsto\lmden{\intreplaceby{\pvar}{\fostatessubsetvar}}{\fml}\) is monotone if \(\pvar\) appears only positively in \(\fml\).
\end{propositionE}

\begin{proofE}
    Immediate by induction on \(\fml\).
\end{proofE}

\subsection{Quantifier and Assignment Actions}\label{sec:quantifier-and-assignment-actions}

As \emph{first-order} variants, the \firstordermucalculus and \firstordergamelogic should have universal and existential quantifiers.
While these could be added separately, it is convenient and illuminating to introduce quantifiers as particular \gamesymbol{s} instead.

\subsubsection{The Quantifier \GameSymbol}
Let \(\gndassignsymb\) be a  \klary{1}{0} \gamesymbol \(\gndassignsymb\) called the quantifier \gamesymbol or nondeterministic assignment \gamesymbol.
This symbol is interpreted over \emph{any} \fonstructure \(\fonstr\) as
\(\fonstrint{\gndassignsymb}(\emptyset) = \pow{\fonstrdom}\setminus\{\emptyset\}.\)
Instances of \(\gndassignsymb\)-atomic games \(\gamesymbat[*]{\ivar}{}\) are also written \(\gndassign{\ivar}\).
The semantics is such that \(\fostate \in \glgden{\gndassign{\ivar}}(\fostatessubset)\) if and only if there is \(\fonstrel\in\fonstrdom\) with \(\streplaceby{\fostate}{\ivar}{\fonstrel}\in\fostatessubset\), which coincides with the typical nondeterministic assignment in the literature \cite{DBLP:journals/tocl/Platzer15,DBLP:conf/stoc/HarelMP77}.
Interestingly, the quantifier \gamesymbol{s} are exactly the usual first-order quantifiers.
In fact \(\lexists{\ivar}\fml\) can be viewed as alternative notation for \(\lpossible{\gndassign{\ivar}}\fml\) and the semantics is exactly as for the existential quantifier:
\[\lmden{\intp}{\lexists{\ivar}\fml} = \{\fostate : \mexists{\fonstrel\in\fonstrdom} \streplaceby{\fostate}{\ivar}{\fonstrel}\in\lmden{\intp}{\fml}\}.\]
The universal quantifier \(\lforall{\ivar}\fml\) is defined by \(\lnot\lexists{\ivar}{\lnot\fml}\) as usual.

As quantifiers are fundamental in the first-order context, in the sequel, these will be assumed to exist in \fogls and \folmus.
\textbf{In the following it is assumed that every \gamesignature contains the quantifier symbol \(\gndassignsymb\) and every \fonstructure interprets it as described above.}
This is analogous to the assumption, common in first-order logic, that every signature contains equality \(=\) as a relation symbol with fixed interpretation.

Viewing quantifiers as \gamesymbol{s} shows that \gamesymbol{s} generalize quantifier symbols.
Even Mostowski quantifiers \cite{Mostowski1957OnAG} are \gamesymbol{s}.
The \firstordermucalculus can be viewed as first-order logic \emph{with} least fixpoints and generalized quantification (in the form of atomic game modalities).
In contrast, \firstordergamelogic can be viewed as the logic \emph{of} generalized quantifiers (which are defined by games with repetitions).
The equiexpressiveness and equivalence of \folmus and \fogls in \Cref{sec:equiexpressivenessandequivalence} unify these different perspectives logically.

The \firstordermucalculus \emph{without} any \gamesymbol{s} other than \(\gndassignsymb\) is equiexpressive with least fixpoint logic (\LFP) \cite{DBLP:journals/bsl/DawarG02}.
(See \appref{sec:lfp-and-folmu} for details.)
The difference between the \firstordermucalculus and least fixpoint logic is that in \LFP fixpoints are finitary predicate symbols, whereas in the \firstordermucalculus they are predicates on the state (so on \emph{all variables}).
Nonetheless, the expressible properties coincide.

\subsubsection{Deterministic Assignment}\label{sec:deterministic-assignment}

Another important \gamesymbol is the \klary{1}{1} deterministic assignment \(\gassignsymb\), which deterministically assigns the value of \(\term\) to variable \(\ivar\).
This state-change primitive is foundational for describing deterministic computer programs in first-order dynamic logic.
The semantics of deterministic assignment is
\[\fonstrint[\fonstr]{(\gassignsymb)}(\fonstrel) = \{\fosubset \subseteq \fonstrdom : \fonstrel\in\fosubset\}.\]
Writing \(\gassign{\ivar}{\term}\) to mean \(\gamesymbat[\gassignsymb]{\ivar}{\term}\), observe that the semantics is \(\glgden{\gassign{\ivar}{\term}}(\fostatessubset) = \{\fostate : \streplaceby{\fostate}{\ivar}{\ltden[\fostate]{\term}}\in \fostatessubset\}\) as expected \cite{DBLP:journals/tocl/Platzer15,DBLP:conf/stoc/HarelMP77}.

The \gamesymbol \(\gassignsymb\) does not need to be added separately, as it is syntactically definable from \(\gndassignsymb\) in \fogls and \folmus.
In \fogls any formula \(\lpossible{\gassign{\ivar}{\term}}\fml\) can be written equivalently without \(\gassignsymb\) using a fresh variable \(\ivarb\) (i.e., \(\ivarb\notin\freevars{\fml}\cup\freevars{\term}\) ) as:
\(\lpossible{\gndassign{\ivarb};\gtest{\ivarb=\term};\gndassign{\ivar};\gtest{\ivar=\ivarb}}\fml\).
In \folmus the formula
\(\lpossible{\gndassign{\ivarb}}(\ivarb=\term\land\lpossible{\gndassign{\ivar}}(\ivar=\ivarb\land\fml))\)
is equivalent to \(\lpossible{\gassign{\ivar}{\term}}\fml\), when \(\ivarb\) is a fresh variable.
In what follows the \gamesymbol \(\gassignsymb\) is treated as an abbreviation.

\subsection{Free Variables, Bound Variables and Substitutions}\label{sec:free-and-bound}

In the context of \firstordergamelogic and the \firstordermucalculus, as generally in first-order logic, the concepts of free and bound variables and substitutions are of critical importance for axiomatizations.
The fixpoint variables, repetition games and atomic games of the form \(\gamesymbat{\ivarseq}{\termseq}\) introduce additional subtleties.
This is outlined here while full definitions and detailed proofs of the relevant properties are in \appref{sec:free-and-bound-variables}.

\subsubsection{Variables in \folmusort}

An \ivarname \(\ivar\) is (syntactically) free in an \folmus formula \(\fml\) if it appears in \(\fml\) in a place not within the scope bound by an atomic game \(\gamesymbat{\ivarbseq}{\termseq}\) where \(\ivar\in\ivarbseq\).
A \pvarname \(\pvar\) is (syntactically) free in \(\fml\) if it appears in \(\fml\) in a place not in the scope of a fixpoint quantifier \(\flfp{\pvar}{\fmlb}\).
The set of (syntactically) free variables \(\freevars{\fml}\) of an \folmus formula \(\fml\) consists of all its free \ivarname{s} \(\ivar\) \emph{and} all its free \pvarname{s} \(\pvar\).
The formula \(\fmlreplacevarby[\fml]{\ivar}{\term}\) is the result of substituting the term \(\term\) for the \ivarname \(\ivar\) in a formula \(\fml\) and this is \emph{defined} so that no free variables are captured.
Importantly, substitution into \pvarname{s} \(\fmlreplacevarpby[\pvar]{\ivar}{\term} = \lpossible{\gassign{\ivar}{\term}}\pvar\) results in additional deterministic assignments.
Similarly, substitutions in atomic games may also introduce deterministic assignments.
Substitution for \pvarname{s} is simpler and has the property that \(\lmden{\intp}{\fmlreplacepvarby{\pvar}{\fmlb}}=\lmden{\intreplaceby{\pvar}{\lmden{\intp}{\fmlb}}}{\fml}\) as long as \(\pvar\) is free for \(\fmlb\) in \(\fml\).

\subsubsection{Variables in \foglsort}

In \firstordergamelogic subtleties arise when substituting into composite games.
As in other contexts \cite{DBLP:conf/cade/Platzer18} it is important to consider the variables that a game can potentially bind and those it necessarily binds.
The definition is such that substitution is \emph{always} allowed and variable capture is impossible.
For example the substitution \(\gamereplacevarpby[{\game;\gameb}]{\ivar}{\term}\) is defined by \(\gamereplacevarby[{\game}]{\ivar}{\term};\gameb\) if \(\game\) \emph{necessarily} binds \(\ivar\).
However, if \(\game\) only possibly binds \(\ivar\) or some variables that are free in \(\term\), then the substitution must be defined to be \(\gassign{\ivar}{\term};\game;\gameb\).
In cases like this, substitution may introduce additional deterministic assignments, just as in \(\folmus\).
This is fundamentally necessary, since the state-based nature of \(\folmus\) and \(\fogls\) makes purely syntactic substitution impossible in some cases.
Deterministic assignments provide an elegant way to nonetheless define a \emph{generalized form of substitution} for all formulas.
In most but not all cases, such as for purely first-order formulas, the generalized substitution does not require deterministic assignments.
In particular, the generalized notion presented here coincides with the usual notion for first-order logic.

%########################################################

\section{Proof Calculus}\label{sec:proof-calculi}

This section introduces proof calculi for \fogls and \folmus.
Since the two logics share a common core, a proof calculus for the shared fragment is introduced first in \Cref{sec:generalcalculus}, followed by proof calculi for \fogls and \folmus in \Cref{sec:gl-calculus,sec:mu-calculus}, respectively.
A deduction theorem is proved in \Cref{sec:contextual-induction-mon-deduction}.

\subsection{A Basic First-order Modal Calculus for Game Actions}\label{sec:generalcalculus}
The basic first-order modal calculus is a Hilbert-style proof calculus for the shared fragment of \fogls and \folmus.
It consists of all propositional tautologies as axioms and the usual axioms for equality together with modus ponens \irref{mp} and the following axioms and rule:

\begin{center}
    \begin{calculuscollection}%
        \begin{calculus}
            \cinferenceRule[mon|M]{monotonicity rule}
            {
                \linferenceRule[sequent]
                {\fml\limply\fmlb}
                {\lpossible{\gamesymbat{\ivarseq}{\termseq}}\fml\limply\lpossible{\gamesymbat{\ivarseq}{\termseq}}\fmlb}
            }{}
        \end{calculus}
        \qquad\qquad
        \begin{calculus}%
            \cinferenceRule[existsaxiom|$\exists$]{existential quantifier axiom}
            {
                \linferenceRule[impl]
                {\fmlreplacevarby{\ivar}{\term}}
                {\lexists{\ivar}\fml}
            }{}
        \end{calculus}
        \\
        \begin{calculus}
            \cinferenceRule[univ|\univaxsymbol]{universality axiom}
            {
                \linferenceRule[impl]
                {(\fmlb\land \lpossible{\gamesymbat{\ivarseq}{\termseq}}\fml)}
                {\lpossible{\gamesymbat{\ivarseq}{\termseq}}(\fmlb\land\fml)}\quad
            }{$\freevars{\fmlb}\subseteq \varcomp{\ivarseq}$}
        \end{calculus}
    \end{calculuscollection}%
\end{center}

\noindent
The monotonicity rule \irref{mon} captures the monotonicity of the interpretations of \gamesymbol{s} and it generalizes the usual monotonicity property of modal logic to the first-order setting.
The existential introduction axiom \irref{existsaxiom} extends the usual axiom from the first-order Hilbert-calculus.
For first-order formulas, substitution for \folmus and \fogls coincides with the usual notion of substitution so that \irref{existsaxiom} is the standard axiom.
However, \irref{existsaxiom} applies to arbitrary \folmus and \fogls formulas, such as \(\fml\synequiv\pvar\), where it is crucial for soundness that \(\fmlreplacevarby[\pvar]{\ivar}{\term}\) is defined as \(\lpossible{\gassign{\ivar}{\term}}\pvar\).
The \contextaxiomname axiom \irref{univ} captures the restricted bounding behavior of games.
It generalizes the first-order theorem \((\fmlb\land\lforall{x}\fml)\limply\lforall{x}(\fmlb\land\fml)\), which holds when \(\ivar\) is not free in \(\fmlb\), from quantifiers binding a single variable to atomic games binding multiple variables.
In \firstordergamelogic, axiom~\irref{univ} can be adapted to retain contextual information (\Cref{sec:contextual-induction-mon-deduction}).

The usual \(\lexists\) quantifier rule \irref{existsrule} is derivable from rule \irref{mon} with context axiom \irref{univ}:

\begin{center}
    \begin{calculuscollection}%
        \begin{calculus}%
            \dinferenceRule[existsrule|$G_\exists$]{existential quantifier rule}
            {
                \linferenceRule[sequent]
                {\fmlb\limply\fml}
                {\lexists{\ivar}\fmlb\limply\fml}\qquad
            }{$\freevars{\fml}\subseteq\ivars\setminus\{\ivar\}$}
        \end{calculus}
    \end{calculuscollection}
\end{center}

Note that the side-conditions of \irref{univ} and \irref{existsrule} ensures that (in the case of \folmus) the formula \(\fml\) does not have free fixpoint variables, which is critical for soundness.
For example \(\lexists{\ivar}\pvar\limply\pvar\) is not valid in \folmus, even though \(\pvar\limply\pvar\) is valid.
For purely first-order formulas, this side-condition is vacuously satisfied.
The proof calculus, thus, is an extension of the complete first-order Hilbert proof calculus.
As such the proof calculus is complete for the first-order fragment.

\subsection{Proof Calculus for \FirstorderGameLogic}\label{sec:gl-calculus}
The proof calculus for \fogls is an extension of the basic first-order modal calculus from \Cref{sec:generalcalculus}.
It adds the standard axioms and rules of propositional game logic \cite{DBLP:conf/focs/Parikh83} in the following form:

\begin{center}
    \begin{calculuscollection}
        \begin{calculus}
            \cinferenceRule[dual|$\lpossible{d}$]{duality axiom}
            {
                \linferenceRule[viuqe]
                {\lpossible{\gdual{\game}}\fml}
                {\lnot \lpossible{\game}\lnot\fml}
            }{}
            \cinferenceRule[test|$\lpossible{?}$]{test axiom}
            {
                \linferenceRule[viuqe]
                {\lpossible{\gtest{\fmlb}}\fml}
                {(\fmlb\land\fml)}
            }{}
        \end{calculus}%
        \qquad%
        \begin{calculus}
            \cinferenceRule[choice|$\lpossible{\cup}$]{choice axiom}
            {
                \linferenceRule[viuqe]
                {\lpossible{\game\gachoice\gameb}\fml}
                {\lpossible{\game}\fml\lor\lpossible{\gameb}\fml}
            }{}
            \cinferenceRule[composition|$\lpossible{{;}}$]{composition axiom}
            {
                \linferenceRule[viuqe]
                {\lpossible{\game;\gameb}\fml}
                {\lpossible{\game}\lpossible{\gameb}\fml}
            }{}
        \end{calculus}%
        \qquad%
        \begin{calculus}
            \cinferenceRule[diaind|I${}_*$]{diamond induction axiom}
            {
                \linferenceRule[sequent]
                {(\fml\lor\lpossible{\game}\fmlb)\limply\fmlb}
                {\lpossible{\garepeat{\game}}\fml\limply\fmlb}
            }{}
            \cinferenceRule[loop|$\lpossible{*}$]{loop axiom}
            {
                \linferenceRule[impl]
                {\fml\lor\lpossible{\game}\lpossible{\garepeat{\game}}\fml}
                {\lpossible{\garepeat{\game}}\fml}
            }{}
        \end{calculus}
    \end{calculuscollection}%
\end{center}

\noindent
Write \(\provfogl[\fml]{\fmlb}\) if there is a proof of \(\fmlb\) from \(\fml\) in the \(\fogls\) calculus.
Axiom \irref{dual} captures the meaning of the \(\gdual{}\)-operator, which flips the roles of the players and axioms \irref{choice}, \irref{test} and \irref{composition} equivalently reduce game operations to propositional logic.
The loop axiom \irref{loop} enables unrolling a loop for a single step and axiom \irref{diaind} is an induction principle: essentially if \(\fml\limply\fmlb\) and \(\fmlb\) is only \(\game\)-reachable from states satisfying \(\fmlb\), then \(\fmlb\) is true in every state from which \(\fml\) is \(\garepeat{\game}\)-reachable.

\begin{theoremE}[\fogls Soundness][normal]\label{prop:fogl-sound}
    The \fogls proof calculus is sound, i.e., \(\provfogl[\fml]{\fmlb}\) implies \(\valfogl[\fml]{\fmlb}\).
\end{theoremE}

\begin{proofE}
    Soundness of the calculus follows inductively from the soundness of each axiom and rule as usual.
    For most of the axioms and rules the proof of soundness is straightforward.
    For \irref{existsaxiom} soundness follows by \apprefexp{prop:substgl}{sec:free-and-bound-variables}.
    Note that \irref{existsaxiom} does not need a side condition due to the definition of the substitution \(\fmlreplacevarby[\fml]{\ivar}{\term}\).
    Soundness of the monotonicity rule \irref{mon} is immediate from \Cref{prop:games-monotone}.

    For axiom \irref{univ} take some state \(\fostate\in\glfden{\fmlb\land\lpossible{\gamesymbat{\ivarseq}{\termseq}}\fml}\).
    By definition of the semantics of atomic games there is some \(\pstatepre\in\fonstrint{\gamesymb}(\ltden[\fostate]{\termseq})\)
    such that
    \(\replpstate[\fostate]{\pstatepreatvar{\ivarseq}}\in\glfden{\fml}\).
    Since \(\freevars{\fmlb}\subseteq\varcomp{\ivarseq}\)
    also \(\strestrvar{\replpstate[\fostate]{\pstatepreatvar{\ivarseq}}}{\freevars{\fmlb}}=\strestrvar{\fostate}{\freevars{\fmlb}}\).
    It follows from \(\fostate\in\glfden{\fmlb}\) by \apprefexp{prop:coincidencegl}{sec:free-and-bound-variables} that also \(\replpstate[\fostate]{\pstatepreatvar{\ivarseq}}\in\glfden{\fml\land\fmlb}\).
    Consequently, \(\fostate\in\glfden{\lpossible{\gamesymbat{\ivarseq}{\termseq}}(\fmlb\land\fml)}\) as required.

    The soundness of game axioms \irref{dual}, \irref{composition}, \irref{test}, \irref{choice} and \irref{loop}, as well as the induction rule \irref{diaind} is as in the literature \cite{DBLP:journals/tocl/Platzer15}.
\end{proofE}

\subsection{Proof Calculus for \FirstorderMuCalculus}\label{sec:mu-calculus}
The proof calculus for \folmus extends the basic first-order modal calculus from \Cref{sec:generalcalculus} with the standard fixpoint axiom \irref{fp} and induction rule \irref{muind} of the propositional \(\mu\)-calculus \cite{DBLP:journals/tcs/Kozen83}.

\begin{center}
    \begin{calculuscollection}
        \begin{calculus}
            \cinferenceRule[fp|$\mu$]{fixpoint axiom}
            {
                \linferenceRule[impl]
                {\fmlfreesetto{\pvar}{\flfp{\pvar}{\fml}}}
                {\flfp{\pvar}{\fml}}
            }{}
        \end{calculus}
        \qquad
        \begin{calculus}
            \cinferenceRule[muind|I${}_\mu$]{fixpoint induction rule}
            {
                \linferenceRule[sequent]
                {\fmlfreesetto{\pvar}{\fmlb}\limply\fmlb}
                {\flfp{\pvar}{\fml}\limply\fmlb}
            }{}
        \end{calculus}
    \end{calculuscollection}
\end{center}

\noindent
Write \(\provfolmu[\fml]{\fmlb}\) if there is a proof of \(\fmlb\) from \(\fml\) in this calculus.
Axiom \irref{fp} says that \(\flfp{\pvar}{\fml}\) is a pre-fixpoint and \irref{muind} is Park's fixpoint induction rule.
Soundness is similar to  \Cref{prop:fogl-sound}.

\begin{theoremE}[\folmus Soundness][]
    The \folmus proof calculus is sound, i.e., \(\provfolmu[\fml]{\fmlb}\) implies \(\valfolmu[\fml]{\fmlb}\).
\end{theoremE}

\begin{proofE}
    Soundness for most axioms and rules is simple.
    Axiom~\irref{existsaxiom} is sound by \apprefexp{prop:substlmu}{sec:free-and-bound-variables}.
    For axiom \irref{univ} take some state \(\fostate\in\lmden{\intp}{\fmlb\land\lpossible{\gamesymbat{\ivarseq}{\termseq}}\fml}\).
    By definition of the semantics of atomic games there is some \(\pstatepre\in\fonstrint{\gamesymb}(\ltden[\fostate]{\termseq})\)
    such that
    \(\replpstate[\fostate]{\pstatepreatvar{\ivarseq}}\in\lmden{\intp}{\fml}\).
    Since \(\freevars{\fmlb}\subseteq\varcomp{\ivarseq}\)
    also \(\strestrvar{\replpstate[\fostate]{\pstatepreatvar{\ivarseq}}}{\freevars{\fmlb}}=\strestrvar{\fostate}{\freevars{\fmlb}}\).
    It follows from \(\fostate\in\lmden{\intp}{\fmlb}\) by \apprefexp{prop:coincidencelmu}{sec:free-and-bound-variables} that also \(\replpstate[\fostate]{\pstatepreatvar{\ivarseq}}\in\lmden{\intp}{\fml\land\fmlb}\).
    Consequently, \(\fostate\in\lmden{\intp}{\lpossible{\gamesymbat{\ivarseq}{\termseq}}(\fmlb\land\fml)}\) as required.
\end{proofE}

\subsection{Contextual Induction, Monotonicity, Deduction Theorem}\label{sec:contextual-induction-mon-deduction}

In \fogls, the \contextaxiomname axiom \irref{univ} can be generalized  from instances of atomic games \(\gamesymbat{\ivarseq}{\termbseq}\) to arbitrary (composite) games $\game$.

\begin{center}
    \begin{calculus}
        \cinferenceRule[axcontext|\univaxsymbol${}^+$]{extended universality axiom}
        {
            \linferenceRule[impl]
            {(\fmlb\land \lpossible{\game}\fml)}
            {\lpossible{\game}(\fmlb\land\fml)}\quad
        }{$\freevars{\fmlb}\cap{\boundvars{\game}}=\emptyset$}
    \end{calculus}
\end{center}

\begin{lemmaE}[][all end]\label{prop:context-axiom-derivable}
    The axiom \irref{axcontext} is derivable in \fogls from \irref{univ}.
\end{lemmaE}

\noindent The \excontextaxiomname axiom \irref{axcontext} is derivable in \fogls from \irref{univ}. (See \Cref{sec:properties-of-proof-calculi}).
There are stronger versions \irref{monstrong}, \irref{diaindstrong} of the monotonicity rule \irref{mon} and the induction rule \irref{diaind} rule with context, respectively.

\begin{lemmaE}\label{prop:mon-box-derivable}
    The two rules \irref{monstrong} and \irref{diaindstrong} are derivable for $\freevars{\fmlc}\cap(\boundvars{\game})=\emptyset$.

    \begin{center}
        \begin{calculuscollection}%
            \begin{calculus}%
                \cinferenceRule[monstrong|M${}_c$]{strong monotonicity rule}
                {
                    \linferenceRule[sequent]
                    {\fmlc\limply(\fml\limply\fmlb)}
                    {\fmlc\limply (\lpossible{\game}\fml\limply\lpossible{\game}\fmlb)}
                }{}
            \end{calculus}
            \qquad
            \begin{calculus}
                \cinferenceRule[diaindstrong|I${}_c$]{strong box induction axiom}
                {
                    \linferenceRule[sequent]
                    {\fmlc\limply((\fml\lor\lpossible{\game}\fmlb)\limply\fmlb)}
                    {\fmlc\limply(\lpossible{\garepeat{\game}}\fml\limply\fmlb)}
                }{}
            \end{calculus}
        \end{calculuscollection}
    \end{center}
\end{lemmaE}

For the proof see \Cref{sec:properties-of-proof-calculi}.
Moreover, all instances of \irref{univ} and \irref{mon} are derivable from \irref{monstrong}.
By this interderivability, \(\fogls\) can be axiomatized with the rule \irref{monstrong} in place of \irref{univ} and \irref{mon}.
The advantage of the contextual rules \irref{axcontext}, \irref{monstrong} and \irref{diaindstrong} is that they retain as much information as possible.
A form of the deduction theorem can be proved as a consequence.

\begin{theoremE}[\(\fogls\) Deduction Theorem][normal] \label{prop:gl-deduction-theorem}
    For  \fogls formulas \(\fmlc, \fmlb\) such that \(\freevars{\fmlc}=\emptyset\):
    \[\fmlc\provfogl{\fmlb}\quad\mimply\quad\provfogl{\fmlc\limply\fmlb}.\]
\end{theoremE}

\noindent The converse implication holds by \irref{mp}.

\begin{proofE}
    The proof of the deduction theorem is, as usual, by induction on the length of the derivation of~\(\fmlc\provfogl{\fmlb}\). If the last step of this derivation is an axiom, there is nothing to show and if the derivation ends with \irref{mp} the induction step is straightforward.
    Suppose the last step in the derivation of~\(\fmlc\provfogl{\fmlb}\) is an application of \irref{mon} that shows
    \(\fmlc\provfogl{\lpossible{\gamesymbat{\ivarseq}{\termseq}}\fml\limply\lpossible{\gamesymbat{\ivarseq}{\termseq}}\fmlb}\).
    By the induction hypothesis \(\provfogl{\fmlc\limply (\fml\limply\fmlb)}\) and hence \irref{monstrong} concludes
    \(\provfogl{\fmlc\limply(\lpossible{\gamesymbat{\ivarseq}{\termseq}}\fml\limply\lpossible{\gamesymbat{\ivarseq}{\termseq}}\fmlb)}\).
    Suppose the last step is an instance of \irref{diaind} that deduces \(\fmlc\provfogl{\lpossible{\garepeat{\game}}\fml\limply\fmlb}\).
    Then by the induction hypothesis \(\provfogl{\fmlc\limply((\fml\lor\lpossible{\game}\fmlb)\limply\fmlb)}\) and \irref{diaindstrong} derives \(\provfogl{\fmlc\limply(\lpossible{\garepeat{\game}}\fml\limply\fmlb)}\).
\end{proofE}

%########################################################

\section{\folmusort{} and \foglsort: Equiexpressiveness and Equivalence}\label{sec:equiexpressivenessandequivalence}

For this section fix a \gamesignature~\(\gsig\) including the \gamesymbol{} \(\gndassignsymb\).
While game logic and the modal \(\mu\)-calculus in their propositional versions \pgls and \plmus are \emph{not} equiexpressive \cite{DBLP:journals/mst/BerwangerGL07}, their first-order versions \fogls and \folmus \emph{are} equiexpressive \emph{and} equivalent.
This can, even more surprisingly, be shown by a reduction to the propositional case.
The propositional reduction of \firstordergamelogic can model \emph{sabotage games} \cite{DBLP:conf/lics/WafaP24}, which allow one player to observably affect future plays of the game by sabotaging the opponent.
The expressive power of sabotage is sufficient to model \plmus on the propositional level.
Definitions of propositional game logic (\pgls) and the propositional modal \(\mu\)-calculus \plmus are in the literature \cite{DBLP:conf/focs/Parikh83,DBLP:journals/tcs/Kozen83,DBLP:conf/lics/WafaP24}.

The reduction to the propositional case is presented in \Cref{sec:propinterpretation}.
Equiexpressiveness and equivalence are proved in \Cref{sec:equiexpressiveness} and \Cref{sec:equivalence}, respectively.

\subsection{Propositional Interpretation}\label{sec:propinterpretation}
The first-order logics \folmus and \fogls can be viewed as propositional logics by abstraction.
For example, any atomic first-order formula can be viewed as a single fixed proposition symbol and similarly for \gamesymbol{s} (\Cref{sec:propositionalsyntax}).
A corresponding transformation of a first-order neighborhood structure into a propositional neighborhood structure retains the semantics (\Cref{sec:propositionalsemantics}) and first-order proofs are equivalent to certain propositional proofs (\Cref{sec:propositionalproofs}).
Thus the first-order \(\mu\)-calculus and game logic can be viewed as instances of the propositional versions by abstracting from variables and atomic formulas.
Crucially, this abstraction requires propositional atomic \gamesymbol{s} to abstract the quantifiers of first-order logic.

\subsubsection{Syntactic Propositional Abstraction}\label{sec:propositionalsyntax}

To define the propositional abstraction formally, pick for any \emph{atomic} \(\gsig\)-formula \(\atfml\) a fresh proposition symbol \(\atfmltoprop{\atfml}\).
For any \emph{atomic} \fogls-game \(\game\) (i.e., \(\gamesymbat[\gamesymb]{\ivarseq}{\termseq}\)) pick a fresh \emph{propositional} game symbol written~\(\gamesymbatprop[{\gamesymb}]{\ivarseq}{\termseq}\).
Let \(\gsigtoprop\) be the \emph{propositional} neighborhood signature consisting of all the fresh proposition symbols \(\atfmltoprop{\atfml}\) and the propositional game symbols \({\gamesymbatprop[{\gamesymb}]{\ivarseq}{\termseq}}\).
For \fogls-formulas \(\fml\) and games \(\game\) let \(\fglfmltoprop{\fml}\) be the \fogls[\gsigtoprop] formula and \(\fglgtoprop{\game}\) the propositional game obtained by replacing all atomic formulas \(\atfml\) by \(\atfmltoprop{\atfml}\)
and all atomic games \(\gamesymbat[\gamesymb]{\ivarseq}{\termseq}\) by \({\gamesymbatprop[{\gamesymb}]{\ivarseq}{\termseq}}\).
Similarly, for an \folmus-formula \(\fml\) let \(\flmfmltoprop{\fml}\) be the \folmus[\gsigtoprop] formula obtained from \(\fml\) by replacing all atomic subformulas \(\atfml\) by~\(\atfmltoprop{\atfml}\) and all atomic games \(\gamesymbat[\gamesymb]{\ivarseq}{\termseq}\) by \({\gamesymbatprop[{\gamesymb}]{\ivarseq}{\termseq}}\).

Conversely, for any propositional \fogls[\gsigtoprop] formula \(\fml\) or game \(\game\) let \(\pglfmltofo{\fml}\) and \(\pglgtofo{\game}\) be obtained by replacing all \(\atfmltoprop{\atfml}\) by \(\atfml\) and all appearances \({\gamesymbatprop[{\gamesymb}]{\ivarseq}{\termseq}}\) by atomic games \(\gamesymbat[\gamesymb]{\ivarseq}{\termseq}\).
Similarly, for any propositional \folmus[\gsigtoprop] formula \(\fml\) let \(\plmfmltofo{\fml}\) be the \(\folmus\) formula obtained by replacing all \(\atfmltoprop{\atfml}\) by \(\atfml\) and all symbols \({\gamesymbatprop[{\gamesymb}]{\ivarseq}{\termseq}}\) by atomic games \(\gamesymbat[\gamesymb]{\ivarseq}{\termseq}\).
The operations are syntactic inverses, so \(\pglfmltofo{{\fglfmltoprop{\fml}}}\equiv\fml\) and \(\fglfmltoprop{{\pglfmltofo{\fmlb}}}\equiv\fmlb\) for \fogls formulas and \pgls formulas \(\fmlb\).
Likewise, \(\plmfmltofo{\flmfmltoprop{\fml}}\equiv\fml\) and \(\fmlb\equiv\flmfmltoprop{\plmfmltofo{\fmlb}}\) for all \folmus formulas \(\fml\) and all \plmus formulas \(\fmlb\).

\subsubsection{Semantics of Propositional Abstraction}\label{sec:propositionalsemantics}

Corresponding to every \(\gsig\)-\fonstructure~\(\fonstr\) a propositional \(\gsigtoprop\)-neighborhood structure~\(\fonstrtons\) is defined where the states of \(\fonstrtons\) are the elements of \(\fostates\) and the interpretation of proposition symbols \(\atfmltoprop{\atfml}\) and \gamesymbol{}s \(\gamesymbatprop[\gamesymb]{\ivarseq}{\termseq}\) are
\(
\pnstrpintp{\fonstrtons}{\atfmltoprop{\atfml}}=\lmden{}{\atfml}
\)
and
\(
\pnstrtrintp{\fonstrtons}{\gamesymbatprop[{\gamesymb}]{\ivarseq}{\termseq}}=\glgden{\gamesymbat[{\gamesymb}]{\ivarseq}{\termseq}}
\).
This abstracts from the first-order elements of \fogls and \folmus to the propositional core.
The soundness of this reduction is clear:
\begin{lemmaE}[Propositional Abstraction]\label{prop:soundnessflat}
    Let \(\fonstr\) be an \(\gsig\)-\fonstructure, \(\fml\) an \folmus formula, \(\fmlb\) a \fogls formula and \(\game\) a \fogls game. Then
    \[\lmden[\fonstr]{\intp}{\fml} = \lmdenprop[\fonstrtons]{\intp}{\flmfmltoprop{\fml}}
        \qquad
        \glfden[\fonstr]{\fmlb} = \glfdenprop[\fonstrtons]{\fglfmltoprop{\fmlb}}
        \qquad
        \glgden[\fonstr]{\game} = \glgdenprop[\fonstrtons]{\fglgtoprop{\game}}.
    \]
\end{lemmaE}
\noindent
Hence \(\fotoprops\) is a semantics-preserving translation with syntactic inverse \(\propstofo\), which shows that the two first-order logics and their propositional counterparts are equiexpressive \emph{over the corresponding \(\gsigtoprop\)-neighborhood structure~\(\fonstrtons\)}.
This equivalence does \emph{not} hold over general structures.
The difficulty that originates with the high expressiveness of the first-order extensions is shifted from the syntactic side to the semantic side.

\subsubsection{Propositional Abstraction of Proofs}
\label{sec:propositionalproofs}

The proof calculi for \fogls and \folmus can also be abstracted to the propositional level, because they crucially add \emph{only axioms} (no rules) over the propositional proof calculus for propositional \pgls and \plmus respectively.
Unlike rules, it is straightforward to translate axioms between equiexpressive logics.
The original proof calculus for propositional \pgls with respect to the symbols from \(\gsigtoprop\) consists of the translated versions of all rules and axioms, except for \irref{univ} and \irref{existsaxiom} and the equality axioms \cite{DBLP:conf/lics/WafaP24}.
The propositionally abstracted (\(\fotoprops\)) versions of these axioms are:

\begin{center}
    \begin{calculuscollection}
        \begin{calculus}
            \cinferenceRule[propexistsaxiom|${\exists^\flat}$]{propositional existential quantifier axiom}
            {
                \linferenceRule[impl]
                {\fgentopropp{\fmlreplacevarby{\ivar}{\term}}}
                {\lpossible{\gamesymbatprop[\gndassignsymb]{\ivarseq}{\term}}\fgentopropp{\fml}}
            }{}
        \end{calculus}
        \qquad
        \begin{calculus}
            \cinferenceRule[propuniv|\univaxsymbol${}^\flat$]{propositional universality axiom}
            {
                \linferenceRule[impl]
                {(\fgentoprop{\fmlb}\land \lpossible{\gamesymbatprop{\ivarseq}{\termseq}}\fgentoprop{\fml})
                }{
                    \lpossible{\gamesymbatprop{\ivarseq}{\termseq}}(\fgentoprop{\fmlb}\land\fgentoprop{\fml})~~~}
            }{{$\freevars{\fmlb}{\subseteq} \varcomp{\ivarseq}$}}
        \end{calculus}%
    \end{calculuscollection}%
\end{center}%
These axioms align the propositional reduction fully with the propositional calculus.
Let \(\foaxiomsinprop\) be the set of all instances of the axioms \irref{propexistsaxiom} and \irref{propuniv} and the propositional \(\fotoprops\) abstractions of the axioms for equality.
Then \(\provpglplus{\fml}\) holds if there is a proof of~\(\fml\) in the propositional proof calculus for \(\pgls\) from the axioms in \(\foaxiomsinprop\).
The extended propositional calculus \(\plmupluss\) is defined analogously as the extension of \(\plmus\) with the axioms from \(\foaxiomsinprop\).
Because \emph{only axioms} are added over the propositional calculus, it is not hard to translate from first-order proofs to propositional proofs and vice versa.

\begin{propositionE}[][normal]
    The propositional reduction preserves proofs:
    \begin{enumerate}[wide,labelindent=0pt]
        \item
              (%     
              \(\provfogl[{\fmlb}]{\fml}\)
              \;\(\Leftrightarrow\)\;
              \(\provpglplusassumption{\fglfmltoprop{\fmlb}}{\fglfmltoprop{\fml}}\)%
              )
              and
              (%
              \(\provfogl[\pglfmltofo{\fmlb}]{\pglfmltofo{\fml}}\)
              \;\(\Leftrightarrow\)\;
              \(\provpglplusassumption{\fmlb}{\fml}\)%
              )\label{propredongames}
        \item
              (\(\provfolmu[{\fmlb}]{\fml}\) \;\(\Leftrightarrow\)\; \(\provplmuplusassumption{\flmfmltoprop{\fmlb}}{\flmfmltoprop{\fml}}\))
              and
              (\(\provfolmu[\plmfmltofo{\fmlb}]{\plmfmltofo{\fml}}\) \;\(\Leftrightarrow\)\; \(\provplmuplusassumption{\fmlb}{\fml}\))\label{propredonmu}
    \end{enumerate}\label{prop:propprooflifting}
\end{propositionE}

\begin{proofE}
    The forward implication of the first equivalence of (\ref{propredongames}) is immediate, by translating the proof tree with the reduction \(\fotoprops\).
    Similarly, the backward implication of the second equivalence of (\ref{propredongames}) can be obtained by translating the proof tree with the expansion \(\propstofo\).
    The remaining directions of (\ref{propredongames}) follow as the translations \(\fotoprops\) and \(\propstofo\) are syntactic inverses.
    Item (\ref{propredonmu}) is shown analogously.
\end{proofE}

Consequently, the propositional reduction \(\fotoprops\) and expansion \(\propstofo\) form an equivalence (\Cref{sec:logics}) between \fogls and its propositional counterpart \pgls with the additional propositional axioms \(\foaxiomsinprop\) translated from \fogls.
Similarly, \folmus is equivalent to the propositional version~\plmus with axioms \(\foaxiomsinprop\) translated from \folmus.

\odo{The F's should be different?}

\subsection{Equiexpressiveness}\label{sec:equiexpressiveness}

When considering \folmus \emph{as a logic}, the set of formulas is taken to be the set of all formulas of \firstordermucalculus without free \pvarname{s}.
This is important, since equiexpressiveness is only meaningful between logics with semantics that share the same \generalparamname and \localparamname parameter sets.
Restricting to formulas with \(\freevars{\fml}\subseteq \ivars\) does not restrict the generality, as free \pvarname{s} can equivalently be viewed as relation symbols in the signature~\(\gsig\).

The propositional reductions of \fogls and \folmus from \Cref{sec:propinterpretation} are the key to showing the equiexpressiveness of the two logics.
Since the propositional modal \(\mu\)-calculus is easily seen to be at least as expressive as propositional game logic \cite{DBLP:conf/focs/Parikh83}, the challenge lies in showing that \fogls can express everything that \folmus can express.
This is proven by distinguishing two cases.
Assume first that the domain of the structure is a singleton set.
In this case the expressive power of \folmus is limited by the fact that all dynamics are constant, so they are easily expressible in~\fogls.

\begin{lemmaE}\label{prop:soundnessoneel}
    There is a translation \(G_1: \folmus\to\fogls\), which is sound with respect to all \fonstructure{s} \(\fonstr\) containing only \emph{one} element.
\end{lemmaE}

\begin{proofE}
    Define the translation \(G_1\) as follows
    \begin{aligntable}[3]
        \begin{align*}
            \nextit{\gentranslation[G_1]{\pvar}}{\lfalse}
            \nextit{\gentranslationp[G_1]{\lnot\fml}}{\lnot\gentranslation[G_1]{\fml}}
            \nextit{\gentranslationp[G_1]{\lpossible{\gamesymbat{\ivarseq}{\termseq}}\fml}}{\lpossible{\gamesymbat{\ivarseq}{\termseq}}\gentranslation[G_1]{\fml}}
            \nextit{\gentranslation[G_1]{\atfml}}{\atfml}
            \nextit{\gentranslationp[G_1]{\fml\land\fmlb}}{\gentranslation[G_1]{\fml}\land\gentranslation[G_1]{\fmlb}}
            \lastit{\gentranslationp[G_1]{\flfp{\pvar}{\fml}}}{\gentranslation[G_1]{\fml}}.
        \end{align*}%
    \end{aligntable}%
    Let \(\intp\) be the interpretation with \(\intp(\pvar)=\emptyset\) for all variables \(\pvar\in\pvars\).
    By induction on \(\fml\) show \(\lmden{\intp}{\fml}=\glfden{\gentranslation[G_1]{\fml}}\).
    The only interesting case is for fixpoint formulas \(\flfp{\pvar}{\fml}\).
    Since there is only one element, there are two cases \(\glfden{\gentranslation[G_1]{\fml}}=\emptyset\) or \(\glfden{\gentranslation[G_1]{\fml}} = \fonstrdom^\ivars\).
    In the former case \(\intreplaceby[\intp]{\pvar}{\glfden{\gentranslation[G_1]{\fml}}} =\intp\) and by induction hypothesis
    \(\lmden{\intp}{\fml}=\glfden{\gentranslation[G_1]{\fml}}.\)
    Hence, as it is the least fixpoint  \(\lmden{\intp}{\flfp{\pvar}{\fml}}\subseteq\glfden{\gentranslation[G_1]{\fml}}=\emptyset\).
    In the latter case \(\fonstrdom^\ivars = \glfden{\gentranslation[G_1]{\fml}} = \lmden{\intp}{\fml} = \lmden{\intp}{\flfp{\pvar}{\fml}}\) by monotonicity.
\end{proofE}

In case the structure contains at least two elements, the propositional reduction of \fogls can mirror the sabotage games of Game Logic with Sabotage (\pglss) \cite{DBLP:conf/lics/WafaP24}.
Since \pglss is equiexpressive and equivalent to the propositional \(\mu\)-calculus, the propositional reduction lifts to equiexpressiveness and equivalence on the first-order level over structures with at least two elements.
To refer to two different elements in a structure fix two distinct constant symbols \(\cont,\conf\).
Just as sabotage does in the propositional context, first-order variables and the two values \(\cont,\conf\) (corresponding roughly to the possible saboteurs) can encode the static forced repetition of fixpoints by the nondeterministic repetition of games.
Intuitively, a fixpoint formula \(\flfp{\pvar}{\fml}\) is inductively translated into a game 
\(\gassign{\ivarc}{\cont};\garepeatp{\gtest{\ivar=\cont};\gassign{\ivarc}{\conf};\tilde{\fml}};\gtest{\ivar=\conf}\),
where \(\ivarc\) is a fresh \ivarname and \(\tilde{\fml}\) is a game corresponding to \(\fml\) which sets \(\gassign{\ivarc}{\cont}\) when it reaches \(\pvar\).
The idea behind this translation was presented as \(\natural\) in the embedding of \pglss to \plmus \cite{DBLP:conf/lics/WafaP24}.
The variables (or sabotage) ensure that the nondeterministic choices in the iteration \(\garepeat{}\) are resolved as required by the fixpoint variables in \(\fml\).

\begin{lemmaE}\label{prop:soundnesstwoel}
    Let \(\gsig\) be a \gamesignature containing \(\cont,\conf\).
    There is a translation \(G_2: \folmus\to\fogls\), which is sound with respect to all \(\fonstr\) with \(\fonstrint{\cont}\neq\fonstrint{\conf}\).
\end{lemmaE}

\begin{proofE}
    Right-linear game logic is the extension of propositional game logic, with recursive games of the form \(\glfp{\gvar}{\game}\) and game variables \(\gvar\).
    Intuitively, \(\glfp{\gvar}{\game}\) is played like \(\gvar\) and repeated, whenever the game variable \(\gvar\) is reached in \(\game\).
    By convention no game variable is bound more than once.
    Recall the sound translation \(\plmutorlglsym:\plmus\to\prlgl\) from sentences of the propositional \(\mu\)-calculus to sentences of right linear game logic \(\prlgl\) \cite{DBLP:conf/lics/WafaP24}.
    Fix an \ivarname \(\sharpivarforpvar{\gvar}\) for every game variable and define a translation \(\rlglremym:\prlgl\to\pgls\) homomorphically, except that
    \begin{align*}
         & \rlglrem{\gvar}\synequiv \fglgtoprop{\gassign{\sharpivarforpvar{\gvar}}{\cont}}
         &
         & \rlglrem{(\glfp{\gvar}{\game})} \synequiv \fglgtopropp{\gassign{\sharpivarforpvar{\gvar}}{\cont}};\garepeatp{\fglgtopropp{\gtest{\sharpivarforpvar{\gvar} =\cont};\gassign{\sharpivarforpvar{\gvar}}{\conf}};\rlglrem{\game}};\fglgtopropp{\gtest{\sharpivarforpvar{\gvar} =\conf}}
    \end{align*}
    for game variables \(\gvar\) and fixpoint games \(\glfp{\gvar}{\game}\) (and dually for \(\ggfp{\gvar}{\game}\)).
    Then \(\provprlglplusassumption{\fglgtopropp{\cont\neq\conf}}\fml\lbisubjunct\rlglrem{\fml}\) is provable for all \(\prlgl\) formulas by induction on \(\prlgl\) games.%
    \footnote{
    More generally, by structural induction on \(\prlgl\) games \(\game\) prove that for all \(\prlgl\) formulas \(\fml\)
    \[\textstyle
    \provprlglplusassumption{\fglgtopropp{{\cont\neq\conf}}}
    \lpossible{\game\tfrac{\gamebseq}{\gvarseq}}\fml\lbisubjunct
    \lpossible{\fglgtopropp{\gassign{\sharpivarforpvar{\gvarseq}}{\conf}};\rlglrem{\game}}(
    \bigand_{i=1}^m(\fglgtopropp{\sharpivarforpvar{\gvar_i} =\cont}\limply \lpossible{\gameb_i}\fml)\land(\bigand_{i=1}^m \fglgtopropp{\sharpivarforpvar{\gvar_i} =\conf}\limply\fml))
    \]
    where \(\gvarseq\synequiv\gvar_1,\ldots,\gvar_m\) is the sequence of free game variables in \(\game\) and \(\gamebseq\synequiv\game_1,\ldots,\game_m\)
    is a sequence of closed \(\prlgl\) games.
    }
    The translation \(\fml^{G_2}\synequiv \plmfmltofo{\rlglrem{\plmutorlgl{\fglfmltoprop{\fml}}}}\)
    is sound with respect to all \(\fonstr\) with \(\fonstrint{\cont}\neq\fonstrint{\conf}\).
\end{proofE}

Since singleton structures are definable in first-order logic and the universe is nonempty, the two translations \(G_1\) and \(G_2\) can be combined to obtain the equiexpressiveness of \folmus and \fogls.

\begin{theoremE}[Fixpoints and Game Equiexpressiveness][normal]\label{prop:glmufoequiexp}
    Logics \folmus and \fogls are equiexpressive.
\end{theoremE}

\begin{proofE}
    There is a sound translation \(\pgltolmus:\pgls\to\plmus\) \cite{DBLP:conf/lics/WafaP24}.
    With \Cref{prop:soundnessflat} this can be lifted to a sound translation \(F:\fogls\to\folmus\) defined by \(\fml^F\synequiv\plmfmltofo{{\pgltolmu{\fglfmltoprop{\fml}}}}\).

    For the converse translation assume, without loss of generality, that \(\cont,\conf\) are not part of the signature and instead view them as fresh variables in~\(\ivars\).
    The translation \(G:\folmus\to\fogls\) with
    \[\fml^G\synequiv((\lforall{\ivar,\ivarb}\; \ivar=\ivarb)\limply\gentranslation[G_1]{\fml})\land \lforall{\cont,\conf}(\cont\neq\conf\limply\gentranslation[G_2]{\fml}).\]
    is sound by \Cref{prop:soundnesstwoel,prop:soundnessoneel}.
\end{proofE}

A consequence of the equiexpressiveness is the collapse of the \pvarname hierarchy in \folmus.
This contrasts with the propositional case, where this question remained open until the  \pvarname{s} hierarchy was shown to be strict \cite{DBLP:journals/mst/BerwangerGL07}.
That result established that the propositional \(\mu\)-calculus is strictly more expressive than propositional game logic, thereby settling their comparative expressiveness.

\begin{theoremE}[\folmus Variable Hierarchy][normal]\label{prop:variable-collapse}
    Any \folmus formula is (provably) equivalent to a formula with two fixpoint variables.
\end{theoremE}

\begin{proofE}
    By \Cref{prop:glmufoequiexp} it suffices to show that any \fogls formula is equivalent to an \folmus formula with two fixpoint variables.
    This is possible as the propositional translation \(\pgltolmus\) \cite{DBLP:conf/lics/WafaP24} can be chosen to ensure that at most two fixpoint variables are used \cite[Section 6.4.2]{pauly01}.
\end{proofE}

\begin{remark}
    Least fixpoint logic (\LFP) is an important first-order fixpoint logic with many applications in finite model theory and descriptive complexity \cite{DBLP:journals/bsl/DawarG02}.
    Unlike \(\folmus\), which adds fixpoints over the entire state space along with modalities for state change, \LFP extends first-order logic minimally by adding fixpoint constructors for \emph{finitary relation symbols} and does not include modalities.
    Nonetheless, \LFP is expressively equivalent to \folmus with nondeterministic assignments as shown in \Cref{prop:lfp-is-follmu} in \Cref{sec:lfp-and-folmu}.
    Thus, despite their very different modes of expression, \LFP and \fogls are equiexpressive, by \Cref{prop:glmufoequiexp}.
    \LFP defines \emph{properties} using a fixpoint operation, whereas \fogls describes \emph{state-change dynamics} and structured iteration.
    Interestingly, the translation and \Cref{prop:variable-collapse} show that two relation symbols suffice for all fixpoint formulas in \LFP \cite{DBLP:journals/bsl/DawarG02}.
\end{remark}

\subsection{Equivalence}\label{sec:equivalence}

The proof calculi for \fogls and \folmus are sufficient to syntactically prove that the translations between the two logics constitute an equivalence.
Again the difficulty is going from \folmus to \fogls.
Just like the translation, the formal proof splits in two with the disjunction \(\provfogl{\lforall{\ivar,\ivarb}\; \ivar=\ivarb}\lor\lexists{\cont,\conf}\cont\neq\conf\).
The case for a non-singleton structure is based on the propositional equivalence which is lifted to the first-order case. In the case of a singleton structure the dynamics must be trivial and the following two technical lemmas for local substitution are useful.

\begin{lemmaE}\label{prop:gonesubsitution}
    \({\provfogl{\gentranslation[G_1]{\fml}\limply\gentranslationp[G_1]{\fmlfreesetto{\pvar}{\fmlb}}}}\)
    for all \folmus formulas \(\fml,\fmlb\), such that \(\pvar\) is only positive in~\(\fml\).
\end{lemmaE}

\begin{proofE}
    Since the translation \(G_1\) translates every fixpoint variable \(\pvar\) to \(\bot\) this is a consequence of monotonicity of \(\fml\) in \(\pvar\) and can be seen by a straightforward induction on \(\fml\).
\end{proofE}

\begin{lemmaE} \label{prop:gonesubsitutiontwo}
    Let \(\ivarseq\) contain all variables from \(\freevars{\fmlb}\cup\freevars{\fmlc}\):
    \[\provfogl{\lforall{\ivarseq}(\gentranslation[G_1]{\fmlb}\lbisubjunct\gentranslation[G_1]{\fmlc}) \limply (\gentranslationp[G_1]{\fmlfreesetto{\pvar}{\fmlb}}\lbisubjunct \gentranslationp[G_1]{\fmlfreesetto{\pvar}{\fmlc}})}.\]
\end{lemmaE}

\begin{proofE}
    By induction on \folmus formulas \(\fml\).
    The only interesting case is for formulas of the form \(\lpossible{\gamesymbat{\ivarbseq}{\termseq}}\tilde\fml\), which follows with rule \irref{monstrong} from the induction hypothesis, since \(\freevars{\fmlb}\cup\freevars{\fmlc}\subseteq\ivarseq\).
\end{proofE}

Now it can be shown that the two logics for games and for fixpoints are \emph{deductively} equivalent.
Despite being different in the mode of expression and contrary to the propositional case, \fogls and \folmus are interchangeable when viewed as logics.

\begin{theoremE}[Fixpoint and Game Equivalence][normal]\label{prop:foequiv}
    The logics \folmus and \fogls are equivalent.
\end{theoremE}

\begin{proofE}
    Let \({F:\fogls\to\folmus}\) and \({G:\folmus\to\fogls}\) be the equi\-expressiveness translations from the proof of \Cref{prop:glmufoequiexp}.

    To analyze the disjunction in the definition of \(G\) let \(\folmusstar\) and \(\foglsstar\) be the extensions of \(\folmus\) and \(\fogls\) by the axiom \(\cont\neq\conf\), respectively.
    On the propositional level, let \(\plmuplusdiff,\pglplusdiff,\prlglplusdiff\) be the logics extending \(\plmus,\pgls,\prlgl\) by the axioms \(\foaxiomsinprop+\fglgtopropp{\cont\neq\conf}\).
    Analogously,  \(\folmusone\) and \(\foglsone\) extend \(\folmus\) and \(\fogls\) with the axiom \(\lforall{\ivar,\ivarb}\; \ivar=\ivarb\) and \(\plmuplusone,\pglplusone,\prlglplusone\) extend \(\plmus,\pgls,\prlgl\) with \(\foaxiomsinprop+\fglgtopropp{\lforall{\ivar,\ivarb}\; \ivar=\ivarb}\).

    Recall that the translations \(F\) and \(G\) are based on the propositional translations \(\pgltolmus:\pgls\to\plmus\) and \(\plmutorlglsym:\plmus\to\prlgl\).
    In fact, \(\pgltolmus\) is the restriction of a translation \(\pgltolmus:\prlgl\to\plmus\) and the pair \(\pgltolmus,\plmutorlglsym\) forms a functorial equivalence between \(\prlgl\) and \(\plmus\) \cite{DBLP:conf/lics/WafaP24}.
    Moreover, \(\pgltolmus,\plmutorlglsym\) naturally also form functorial equivalences between \(\prlglplusdiff\) and \(\plmuplusdiff\) and \(\prlglplusone\) and \(\plmuplusone\).
    Recall also the translation \(\rlglremym:\prlgl\to\pgls\) from \Cref{prop:soundnesstwoel}.
    As shown in the proof of \Cref{prop:soundnesstwoel}, it can be seen that the translation \(\rlglremym\) and the inclusion \(\glrlglinc:\pgls\to\prlgl\) form an equivalence between \(\prlglplusdiff\) and \(\pglplusdiff\).

    Note that \(\provfogl{\fml}\) if and only if \(\provrel{}{\foglsone}{\fml}\) and \(\provrel{}{\foglsstar}{\fml}\), by the deduction theorem (\Cref{prop:gl-deduction-theorem}) and since \(\provfogl{\lforall{\ivar,\ivarb}\; \ivar{=}\ivarb\lor(\lexists{\conf,\cont}\,\conf{\neq}\cont)}\).
    Hence, it suffices to prove that the translations \(F\) and \(G\) form equivalences between \(\folmusone\) and \(\foglsone\) and between \(\folmusstar\) and \(\foglsstar\).
    First, consider the translations \(F,G\) between \(\folmus\) and \(\fogls\). These form an equivalence of \(\folmusstar\) and \(\foglsstar\), since \(G\) and~\(F\) are the compositions of the following functorial equivalences:

    \begin{center}
        \begin{tikzpicture}[node distance=1cm, >=stealth]

            \node (a) {\raisebox{-6pt}{$\folmusstar$}};
            \node (b) [right=of a] {\raisebox{-8pt}{$\plmuplusdiff$}};
            \node (c) [right=of b] {$\prlglplusdiff$};
            \node (d) [right=of c] {$\pglplusdiff$};
            \node (e) [right=of d] {\raisebox{-6pt}{$\foglsstar$.}};

            \draw[->] ([yshift=3pt]a.east) -- node[above, font=\scriptsize] {$\fotoprops$} ([yshift=3pt]b.west);
            \draw[->] ([yshift=3pt]b.east) -- node[above, font=\scriptsize] {$\plmutorlglsym$} ([yshift=3pt]c.west);
            \draw[->] ([yshift=3pt]c.east) -- node[above, font=\scriptsize] {$\rlglremym$} ([yshift=3pt]d.west);
            \draw[->] ([yshift=3pt]d.east) -- node[above, font=\scriptsize] {$\propstofo$} ([yshift=3pt]e.west);

            \draw[<-] ([yshift=-3pt]a.east) -- node[below, font=\scriptsize] {$\propstofo$} ([yshift=-3pt]b.west);
            \draw[<-] ([yshift=-3pt]b.east) -- node[below, font=\scriptsize] {$\pgltolmus$} ([yshift=-3pt]c.west);
            \draw[<-] ([yshift=-3pt]c.east) -- node[below, font=\scriptsize] {$\glrlglinc$} ([yshift=-3pt]d.west);
            \draw[<-] ([yshift=-3pt]d.east) -- node[below, font=\scriptsize] {$\fotoprops$} ([yshift=-3pt]e.west);

            \draw[->] ([yshift=0pt]a.north) to[bend left=17] node[below, font=\small] {$G$} ([yshift=0pt]e.north);
            \draw[->] ([yshift=-2pt]e.south) to[bend left=17] node[above, font=\small] {$F$} ([yshift=-2pt]a.south);

        \end{tikzpicture}
    \end{center}

    Next, prove that \(F\) and \(G\) form an equivalence between \(\folmusone\) and \(\foglsone\).
    Begin by verifying \iref{prop:logicequivalencethereonly} of \Cref{def:logics-equivalent}.
    Suppose first \(\provrel{}{\foglsone}{\fml}\) so that \(\provrel{}{\smash{\pglplusone}}{\fglfmltoprop{\fml}}\) by \Cref{prop:propprooflifting}.
    From the equivalence property of the propositional translation \(\pgltolmus\) \cite[Theorem 6.7]{DBLP:conf/lics/WafaP24}, it follows that \(\provrel{}{\smash{\plmuplusone}}{\pgltolmu{\fglfmltoprop{\fml}}}\).
    Then \(\provrel{}{\folmusone}{\gentranslation[F]{\fml}}\) follows from \Cref{prop:propprooflifting} since \(\gentranslation[F]{\fml}\synequiv\plmfmltofo{{\pgltolmu{\fglfmltoprop{\fml}}}}\).

    Next assume \(\provrel{}{\folmusone}{\fml}\).
    Proceed by induction on the derivation and distinguish cases based on the last applied rule.
    Most cases are straightforward and only the cases where the last step is an instance of \irref{fp} or \irref{muind} are of interest.
    For \irref{fp} it suffices to derive \(
    \provrel{}{\foglsone}{\gentranslation[G_1]{\fml}}\lbisubjunct\gentranslationp[G_1]{\fmlfreesetto{\pvar}{\fml}}\)
    where \(\pvar\) appears only positively in \(\fml\).
    The \(\limply\) implication derives by \Cref{prop:gonesubsitution}.
    For \(\leftarrow\) first-order reasoning derives
    \(\provrel{}{\foglsone}{\lforall{\ivarcseq}\gentranslation[G_1]{\fml}\lor\lforall{\ivarcseq}\lnot\gentranslation[G_1]{\fml}}\) where \(\ivarcseq\) contains all variables free in \(\fml\), since the domain is necessarily a singleton.
    Hence, it suffices to show the two disjuncts
    \(\provrel{}{\foglsone}{(\lforall{\ivarcseq}\gentranslation[G_1]{\fml})\limply(\gentranslationp[G_1]{\fmlfreesetto{\pvar}{\fml}}\limply\gentranslation[G_1]{\fml})}\)
    and \(\provrel{}{\foglsone}{(\lforall{\ivarcseq}\lnot\gentranslation[G_1]{\fml})\limply(\gentranslationp[G_1]{\fmlfreesetto{\pvar}{\fml}}\limply\gentranslation[G_1]{\fml})}\).
    The former derives from \irref{existsaxiom} with propositional reasoning.
    The latter reduces by \Cref{prop:gonesubsitutiontwo} to \(\provrel{}{\foglsone}{\gentranslationp[G_1]{\fmlfreesetto{\pvar}{\lfalse}}\limply\gentranslation[G_1]{\fml}}\).
    This is provable as \(\gentranslationp[G_1]{\fmlfreesetto{\pvar}{\lfalse}}\synequiv\gentranslation[G_1]{\fml}\) by definition of \(G_1\).

    For instances of \irref{muind} by induction hypothesis assuming that \(\provrel{}{\foglsone}{\gentranslationp[G_1]{\fmlfreesetto{\pvar}{\fmlb}}\limply\gentranslation[G_1]{\fmlb}}\)
    and that \(\pvar\) appears only positively in \(\fml\), it needs to be shown that \(\provrel{}{\foglsone}{\gentranslationp[G_1]{\flfp{\pvar}{\fml}}\limply\gentranslation[G_1]{\fmlb}}\).
    This follows as \(\provrel{}{\foglsone}{\gentranslation[G_1]{\fml}\limply\gentranslationp[G_1]{\fmlfreesetto{\pvar}{\fmlb}}}\) holds by \Cref{prop:gonesubsitution}.

    To verify \iref{prop:logicequivalencethereandback} of \Cref{def:logics-equivalent} start by proving \(\provrel{}{\folmus}{\fml\lbisubjunct \gentranslation[G_1F]{\fml}}\).
    Since \(\gentranslation[G_1]{\fml}\synequiv\gentranslation[G_1F]{\fml}\) it suffices to show  \(\provrel{}{\folmusone}{\fml\lbisubjunct \gentranslation[G_1]{\fml}}\) by induction on \(\fml\).
    More generally, by induction on formulas \(\fml\) prove that \(\provrel{}{\folmusone}{\gentranslationp[G_1]{\fmlreplacepvarby[\fml]{\pvar}{\fmlb}}\lbisubjunct \fmlreplacepvarby[\fml]{\pvar}{\gentranslation[G_1]{\fmlb}}}\), where \(\pvar\) can be a list of \pvarname{s} and \(\fmlb\) a list of \folmus formulas.
    For the implication \( \fmlreplacepvarby[\flfp{\pvarb}{\fml}]{\pvar}{\gentranslation[G_1]{\fmlb}}\limply\gentranslationp[G_1]{\fmlreplacepvarby[\flfp{\pvarb}{\fml}]{\pvar}{\fmlb}}\) use \(\provrel{}{\folmus}{\gentranslation[G_1]{\fml}}\limplied\gentranslationp[G_1]{\fmlfreesetto{\pvar}{\fml}}\), which is proved as above.

    Finally, it remains to prove \(\provrel{}{\foglsone}{\fmlb\lbisubjunct \gentranslation[FG_1]{\fmlb}}\).
    Note that for iteration-free \fogls formulas \(\fmlb^-\) the translation \(\gentranslation[F]{\fmlb^-}\) is both a \folmus and a \fogls formula so that \(\provrel{}{\foglsone}{\fmlb^-\lbisubjunct \gentranslationp[F]{\fmlb^-}}\) follows by a straightforward by induction.
    Hence, it suffices to find for every \fogls formula \(\fmlb\) an iteration-free formula \(\fmlb^-\) such that \(\gentranslationp[F]{\fmlb^-}\synequiv\gentranslationp[FG_1]{\fmlb}\) and \(\provrel{}{\foglsone}{\fmlb\lbisubjunct \fmlb^-}\).
    The formula \(\fmlb^-\) can be obtained from \(\fmlb^-\) by replacing any iteration game \(\garepeat{\game}\) by the single round game \(\gtest{\top}\gachoice\game\).
    The proof that \(\provrel{}{\foglsone}{\fmlb\lbisubjunct \fmlb^-}\) is just like the proof that \(\provrel{}{\folmusone}{\fml\lbisubjunct \gentranslation[G_1]{\fml}}\).
\end{proofE}

Via propositional reduction, this section transferred the equiexpressiveness equivalence of the propositional \(\mu\)-calculus and expressive game logics (\pglss or \prlgl) to the first-order setting.
Note that these results apply to every instance of \fogls and \folmus so that there is a translation for every \gamesignature~\(\gsig\).
Recall that the \generalparamname parameters of the semantics of \fogls and \folmus are the \(\gsig\)-\fonstructure{s}~\(\fonstr\) and the \localparamname parameters are the states \(\fostates\).
Thus equiexpressiveness ensures that the \fogls-definable sets and the \folmus-definable sets over any \(\gsig\)-\fonstructure~\(\fonstr\) coincide.

%########################################################

\section{Differential Equations As Fixpoints}\label{sec:differential-equations-as-fixpoints}

\DifferentialGameLogic  (\dGL) is a logic for reasoning about adversarial hybrid games and is applied to the verification of adversarial cyber-physical systems \cite{DBLP:journals/tocl/Platzer15}.
On a high-level, \dGL is \firstordergamelogic interpreted over the real numbers \(\reals\) with nondeterministic assignment, extending the definition of games with \emph{atomic} differential equation actions described by an ordinary differential equation \(\gode{\ivarseq}{\odefof}\).
Player \Playerone can choose to move to any state, which is reachable along the solution curve of the ODE.
While these actions make it possible to express properties of hybrid games, they add new challenges to the deductive system.

The novel insight here is that these continuous processes can be viewed equivalently as discrete fixpoints via a new recursive understanding of ODE reachability.
This understanding differs from numerical approaches to treat ODEs via discrete approximations in important ways.
Those methods compute approximations in small time-steps (forwards or backwards) to keep the (accumulated) error small.
This complicates the correct use of algorithms like Euler's method, which require careful topological arguments to handle approximation errors and consequently lead to complex axiomatizations \cite{DBLP:conf/lics/Platzer12b}.
The approach presented here is \emph{global} and describes ODE reachability via recursive satisfaction of simple second-order Taylor bounds.
This simplifies the discrete characterization of ODEs as no consideration of error accumulation is needed.
Making use of the equivalence from \Cref{sec:equiexpressivenessandequivalence} and the fixpoints of the \firstordermucalculus, the fixpoint understanding of ODEs makes it possible to axiomatize them natively in \dGL as a game.
This section shows that \dGL is an instance of \firstordergamelogic in \Cref{sec:dgl-as-fogl}, describes the fixpoint understanding of continuous evolution in \Cref{sec:fp-axiomatization} and presents an axiomatization in \Cref{sec:dgl-axioms}.

\subsection{\DifferentialGameLogic as a \FirstorderGameLogic}\label{sec:dgl-as-fogl}

The syntax and semantics of \differentialgamelogic are first explained intuitively before \Cref{sec:ode-modalities} shows how \differentialgamelogic and \differentialmucalculus are formally defined as instances of first-order game logic and \(\mu\)-calculus, respectively.

To model continuous evolutions in first-order logic interpreted over the real numbers, the syntax will allow actions of the form \(\stdode\), where \(\odevarsx\) is an \(\ell\)-sequence of variables and \(\odef\) is an \(\ell\)-sequence of polynomials.
Semantically this modality describes a game in which one player can evolve the values of \(\odevarsx\) by following a solution of the differential equation \(\stdodewo\), as long as the region defined by the \emph{evolution domain constraint} \(\evdfml\) is not left.
Thus, \(\lnecessary{\stdode}\fml\) intuitively means that \(\fml\) holds for the entire evolution along the differential equation \(\stdodewo\) within \(\evdfml\).
For example, the simple formula \((x=1\land t=0)\limply \lnecessary{x'=x,t'=1}x\geq 1+t\) says that the solution of the initial value problem \(\gamma'=\gamma\) with \(\gamma_0=1\) satisfies \(\gamma_t\geq 1+t\) (or equivalently \(e^t\geq 1+t\)) for \(t\geq 0\).

\subsubsection{Continuous Reachability}

A continuous function \(\gamma: [a,b]\to\reals^\ell\) is a \emph{solution} of the differential equation \(x'=F(x)\) for the vector field $F:\reals^{\ell}\to \reals^\ell$ if $\gamma$ is differentiable on $(a,b)$ and satisfies ${\gamma}'(t)= F(\gamma(t))$ for all $t\in (a,b)$.
For readability the notation $\gamma_s =\gamma(s)$ is used synonymously.
A point $y$ is \emph{$t$-reachable} from $x$ along $F$ in \(C\subseteq\reals^\ell\), written \(\contreachin[t]{F}{x}{y}{C}\), if and only if there is a solution $\gamma:[0,t] \to C$ of $x'=F(x)$ such that $\gamma_0 = x$ and $\gamma_t =y$.
The \emph{reachability relation} is the set of all triples \((x,y,t)\) such that \(y\) is reachable from \(x\) along \(F\) in \(C\) in time \(t\):
\[\reachrelofin{F}{C} = \{(x,y,t)\in\reals^\ell\times\reals^\ell\times[0,\infty) : \contreachin[t]{F}{x}{y}{C}\}.\]

\subsubsection{Polynomial Differential Equations}

In order to include differential equations as \gamesymbol{s} in \fogls, a syntactic notion of vector fields is presented here.
The first-order signature of rings consists of constant symbols \(0,1\), binary function symbols \(+,\cdot\) and unary function symbol \(-\).
For a sequence of variables \(\odevarsx = \ivar_1,\ldots,\ivar_\ell\), a \emph{syntactic vector field} \(\odef\) is a \emph{sequence} \(\odef_1,\ldots \odef_\ell\) of terms such that the only free variables in \(\odef_i\) are from \(\odevarsx\).
The notation \(\odefof\) emphasizes that \(\odevarsx=\freevars{\odef}\).
Because \(\odefof\) mentions only variables of \(\odevarsx\) the function \(F(x)=\ltden[\streplaceby{\fostate}{\odevarsx}{x}]{\odef}\) does not depend on the choice of state \(\fostate\). In the following write \(\odefinterp\) for this function \(F:\reals^\ell\to\reals^\ell\) interpreting the syntactic vector field \(\odefof\).

\subsubsection{Differential Equation Modalities in \FirstorderGameLogic and \FirstorderMuCalculus}\label{sec:ode-modalities}

Formally \emph{\differentialgamelogic} (\dGL) is viewed here as an interpreted version of \firstordergamelogic.
The \gamesignature \(\dglsig\) of \dGL is the first-order signature of rings with additional \klary{\ell}{\ell} \gamesymbol{s} \(\dsymb\) for every syntactic vector field \(\odefof\) and every first-order formula \(\evdfml\) in the signature of rings with only variables from \(\odevarsx  = \ivar_1,\ldots,\ivar_\ell\) free.
\Differentialgamelogic is interpreted over the \fonstructure \(\rfnstr\).
The domain is the field of real numbers and the first-order constant and function symbols \(0,1,+,\cdot,-\) are interpreted as usual.
For readability write \(\evdfmlinterp\) for the interpretation \(\{\fostate(\ivarseq) : \fostate\in\glfden[\rfnstr]{\evdfml}\}\subseteq\reals^\ell\), where \(\ivarseq\) is the sequence of free variables of \(\evdfml\).

The semantics of the action \(\dsymb\) for the syntactic vector field \(\odefof\) and the evolution domain constraint \(\evdfml\) is defined as
\[\fonstrint[\rfnstr]{\left(\dsymb\right)}(\fonstrel) = \{\{\fonstrelb\} : \contreachin{\odefinterp}{\fonstrel}{\fonstrelb}{\evdfmlinterp}\}.\]
Use the notation \(\stdode\) in place \({\gamesymbat[{\dsymb[{\odefof}]}]{\ivarseq}{\ivarseq}}\). Then
\[
    \glgden{\stdode}(\fostatessubset)
    =
    \{\fostate : \mexists{\fostateb}\; \contreachin{\odefinterp}{\strestrvar{\fostate}{\odevarsx}}{\strestrvar{\fostateb}{\odevarsx}}{\evdfmlinterp}\text{ and } \fostateb\in\fostatessubset\}.
\]
The formulation here is slightly different from the original formulation of \dGL \cite{DBLP:journals/tocl/Platzer15}, which takes \(\stdode\) games as primitives.
The present formulation fits neatly into the \firstordergamelogic framework and does not in fact add any generality, since the general modalities \(\lpossible{\gamesymbat[\dsymb]{\ivarseq}{\termseq}}\) can be expressed with deterministic assignment and modalities of the form \(\stdode\).
The restriction that the syntactic vector field \(\odefof\) and the evolution domain constraint \(\evdfml\) mention only variables from \(\odevarsx\) is irrelevant, as any additional variable \(\ivarb\) can be made available by adding \(\ivarb'=0\) to the differential equation.
Write \(\stdodewo\) for the action \(\stdodewith{\ltrue}\) with the trivial evolution domain constraint \(\ltrue\).
The \emph{\differentialmucalculus} (\dLmu) is the interpreted version of the \firstordermucalculus in the signature \(\dglsig\) interpreted over \(\rfnstr\) and shares the interpretation of the continuous evolution symbol \(\dsymb\) with \dGL.

The general equiexpressiveness result of \Cref{sec:equiexpressiveness} applies to the case of \dGL and \dLmu:

\begin{proposition}\label{prop:dgl-and-dlmu-equiexpressive}
    \Differentialgamelogic and the \differentialmucalculus are equiexpressive.
\end{proposition}

\begin{proof}
    The logics \dGL and \dLmu are instances of \firstordergamelogic and \firstordermucalculus, respectively for the \gamesignature \(\dglsig\) and restricted to the first-order neighborhood structure \(\rfnstr\).
    Since the structures are the \generalparamname parameters of the semantics, equiexpressiveness of \dGL and \dLmu is an instance of \Cref{prop:glmufoequiexp}.
\end{proof}
    
\subsubsection{Non-autonomous Differential Equations}

The definition of continuous game actions (as the original definition of \differentialgamelogic \cite{DBLP:journals/tocl/Platzer15}) uses \emph{autonomous} differential equations.
However by adding a \emph{time} variable \(\timevar\) to measure the duration of the evolution, non-autonomous differential equations can be handled easily by adding \(\timevar'{=}1\) as in \(\gode{\odevarsx}{\odefof},\gode{\timevar}{1}\).
Since this will often be useful, the notation \(\stdodet\) is viewed as an abbreviation of the composition \(\gassign{\timevar}{0}\gcom\gode{\odevarsx}{\odefof},\gode{\timevar}{1}\evdsep\evdfml\), which resets the special time variable \(\timevar\) to measure the duration of evolution.

\subsection{Fixpoint Description of ODEs}\label{sec:fp-axiomatization}

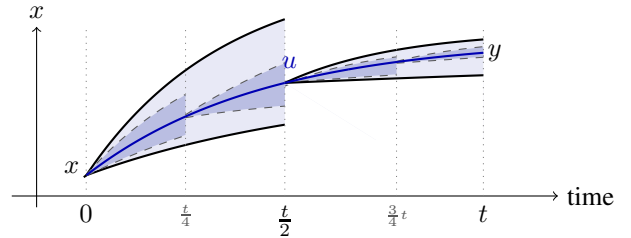
\begin{figure}[!htbp]
    % !TeX root = fig-dynamics-nabla0.tex
\providecommand{\ivr}{Q}%
\providecommand{\genDE}[1]{F(#1)}%
\providecommand{\I}{\dLint[state=x]}%
\providecommand{\It}{\dLint[state=y]}%
\providecommand{\Iz}{\dLint[state=u]}%
\newcommand{\arccol}{semblue}
\newcommand{\shadecol}{vblue!10}
\newcommand{\shadecolb}{vblue!30}
\newcommand{\reccol}{black!50}
\newcommand{\smallcoll}{vgray}
\providecommand{\rname}{t}%
\begin{tikzpicture}[xscale=2.2,scale=1.5]
  \tikzstyle{mode switch}=[gray,thin,dotted]
  \tikzstyle{smallbound}=[\smallcoll,thick,dashed]
  \useasboundingbox (-0.2,-0.5) rectangle (2.6,1.5);
  \def\conelen{-0.2}%{-0.2}% started from time -1
  \def\r{1.8}%
  
  %% %%%%%%%%%%%%%%%%%
  %% coordinate system 
  \begin{scope}
    \draw[->] (-0.1,0) -- (2.1,0) node[right] {time} coordinate(T axis);
    \draw[->] (0,-0.1) -- (0,1.5) node[above] {$x$} coordinate(x axis);
  \end{scope}

  % LHS Fill
  \begin{scope}
    \clip (-0.1,0.6) arc[start angle=160, end angle=400,x radius=0.95cm, y radius=0.35cm] -- +(0,2) -- +(-2,0) -- cycle;
    \fill[fill=\shadecol,domain=-1:\conelen,smooth,xshift=1.2cm,yshift=-0.1cm] plot (\x,{1.5-0.3*exp(-1.8*\x) - 0.3*exp(-1.4*-1) + 0.3*exp(-1.8*-1)}) -| (-0.2,0);
  \end{scope}
  
  % Below LHS cutout
  \begin{scope}
    \clip (-0.1,0.6) arc[start angle=160, end angle=400,x radius=0.95cm, y radius=0.4cm]
    .. controls +(-0.4,1.2) and +(0.5,-0.3) .. (-0.1,0.6);
    \fill[fill=white!10,domain=-1:-0.2,smooth,xshift=1.2cm,yshift=-0.1cm] plot (\x,{1.5-0.3*exp(-1*\x) - 0.3*exp(-1.4*-1) + 0.3*exp(-1*-1)}) -| (0,0);
  \end{scope}
  
  % RHS Fill
  \begin{scope}
    \clip (1,0.5) rectangle (1.8,1.6);
    \fill[fill=\shadecol,domain=\conelen:0.6,smooth,xshift=1.2cm,yshift=-0.1cm] plot (\x,{1.5-0.3*exp(-2.2*\x) - 0.3*exp(-1.4*-0.2) + 0.3*exp(-2.2*-0.2)}) -| (0.6,0);
  \end{scope}

  %% %%%%%%%%%%%%%%%%%
  %% half to 3/4 bounds
  % Upper bound
  \draw[smallbound,thick,domain=\conelen:0.25,smooth,xshift=1.2cm,yshift=-0.24cm] plot (\x,{1.5-0.2*exp(-2.8*\x) - 0.3*exp(-1.4*-0.2) + 0.3*exp(-2.2*-0.2)});
  % Lower bound
  \draw[smallbound,thick,domain=\conelen:0.25,smooth,xshift=1.2cm,yshift=0.28cm] plot (\x,{1.5-0.5*exp(-0.3*\x) - 0.5*exp(-1.4*-0.2) + 0.4*exp(-0.3*-0.2)});
  \begin{scope}
    % \clip (\conelen, -1) rectangle (0.25, 2); % Clipping to the domain
    \fill[\shadecolb] 
      plot[domain=\conelen:0.25,smooth,xshift=1.2cm,yshift=-0.24cm] 
      (\x,{1.5-0.2*exp(-2.8*\x) - 0.3*exp(-1.4*-0.2) + 0.3*exp(-2.2*-0.2)})
      -- plot[domain=0.25:\conelen,smooth,xshift=1.2cm,yshift=0.28cm] 
      (\x,{1.5-0.5*exp(-0.3*\x) - 0.5*exp(-1.4*-0.2) + 0.4*exp(-0.3*-0.2)}) 
      -- cycle;
  \end{scope}

  %% %%%%%%%%%%%%%%%%%
  %% 3/4 to end bounds
  % Upper bound
  \draw[smallbound,domain=\conelen:0.15,smooth,xshift=1.65cm,yshift=-0.17cm] plot (\x,{1.5-0.12*exp(-2.8*\x) - 0.3*exp(-1.4*-0.2) + 0.3*exp(-2.2*-0.2)});
  % Lower bound
  \draw[smallbound,domain=\conelen:0.15,smooth,xshift=1.65cm,yshift=0.45cm] plot (\x,{1.5-0.5*exp(-0.3*\x) - 0.5*exp(-1.4*-0.2) + 0.4*exp(-0.3*-0.2)});
  \begin{scope}
    \fill[\shadecolb] 
      plot[domain=\conelen:0.15,smooth,xshift=1.65cm,yshift=-0.17cm] 
      (\x,{1.5-0.12*exp(-2.8*\x) - 0.3*exp(-1.4*-0.2) + 0.3*exp(-2.2*-0.2)})
      -- plot[domain=0.15:\conelen,smooth,xshift=1.65cm,yshift=0.45cm] 
      (\x,{1.5-0.5*exp(-0.3*\x) - 0.5*exp(-1.4*-0.2) + 0.4*exp(-0.3*-0.2)}) 
      -- cycle;
  \end{scope}

  % Below RHS cutout
  \begin{scope}
    \clip (1,0.5) rectangle (1.8,1.6);
    \fill[fill=white,domain=\conelen:0.6,smooth,xshift=1.2cm,yshift=-0.1cm] plot (\x,{1.5-0.3*exp(-0.3*\x) - 0.3*exp(-1.4*-0.2) + 0.3*exp(-0.3*-0.2)}) -| (0.6,0);
  \end{scope}
  
    %% %%%%%%%%%%%%%%%%%
  %% half to end bounds
  % Upper bound
  \draw[thick,domain=\conelen:0.6,smooth,xshift=1.2cm,yshift=-0.1cm] plot (\x,{1.5-0.3*exp(-2.2*\x) - 0.3*exp(-1.4*-0.2) + 0.3*exp(-2.2*-0.2)});
  % Lower bound
  \draw[,thick,domain=\conelen:0.6,smooth,xshift=1.2cm,yshift=-0.1cm] plot (\x,{1.5-0.3*exp(-0.3*\x) - 0.3*exp(-1.4*-0.2) + 0.3*exp(-0.3*-0.2)});

  %% %%%%%%%%%%%%%%%%%
  %% begin to half bounds
  % Upper bound
  \draw[thick,domain=-1:-0.2,smooth,xshift=1.2cm,yshift=-0.1cm] plot (\x,{1.5-0.3*exp(-1.8*\x) - 0.3*exp(-1.4*-1) + 0.3*exp(-1.8*-1)});
  % Lower bound
  \draw[thick,domain=-1:-0.2,smooth,xshift=1.2cm,yshift=-0.1cm] plot (\x,{1.5-0.3*exp(-1*\x) - 0.3*exp(-1.4*-1) + 0.3*exp(-1*-1)});
  
  %% %%%%%%%%%%%%%%%%%
  %% begin to quater bounds
  %  Upper bound
  \draw[smallbound,domain=-1:-0.6,smooth,xshift=1.2cm,yshift=0.08cm] plot (\x,{1.5-0.23*exp(-1.8*\x) - 0.3*exp(-1.4*-1) + 0.2*exp(-1.8*-1)});
  % Lower bound
  \draw[smallbound,domain=-1:-0.6,smooth,xshift=1.2cm,yshift=-0.1cm] plot (\x,{1.5-0.4*exp(-1*\x) - 0.3*exp(-1.4*-1) + 0.4*exp(-1*-1)}); 
  
  %% %%%%%%%%%%%%%%%%%
  %% begin to mid color
  % Upper bound
  \begin{scope}
    \fill[\shadecolb]
    plot[domain=-1:-0.6,smooth,xshift=1.2cm,yshift=0.08cm]
    (\x,{1.5-0.23*exp(-1.8*\x) - 0.3*exp(-1.4*-1) + 0.2*exp(-1.8*-1)})
    -- plot[domain=-0.6:-1,smooth,xshift=1.2cm,yshift=-0.1cm]
    (\x,{1.5-0.4*exp(-1*\x) - 0.3*exp(-1.4*-1) + 0.4*exp(-1*-1)}) 
    -- cycle;
  \end{scope}
  
  %% %%%%%%%%%%%%%%%%%
  %% quater to mid bounds
  % Upper bound
  \draw[smallbound,domain=0:0.4,smooth,xshift=0.6cm,yshift=0.52cm] plot (\x,{1.5-1.4*exp(-1*\x) - 0.3*exp(-1.4*0.45) + 0.4*exp(-1*0.45)}); 
  % Lower bound
  \draw[smallbound,domain=0:0.4,smooth,xshift=0.6cm,yshift=-0.53cm] plot (\x,{1.5-0.2*exp(-1.8*\x) - 0.3*exp(-1.4*0.45) + 0.2*exp(-1.8*0.45)});
  
  %% %%%%%%%%%%%%%%%%%
  %% quater to mid color
  % Upper bound
  \begin{scope}
    % \clip (0, -1) rectangle (0.4, 2); % Clipping to the domain
    \fill[\shadecolb]
    plot[domain=0:0.4,smooth,xshift=0.6cm,yshift=0.52cm]
    (\x,{1.5-1.4*exp(-1*\x) - 0.3*exp(-1.4*0.45) + 0.4*exp(-1*0.45)})
    -- plot[domain=0.4:0,smooth,xshift=0.6cm,yshift=-0.53cm]
    (\x,{1.5-0.2*exp(-1.8*\x) - 0.3*exp(-1.4*0.45) + 0.2*exp(-1.8*0.45)}) 
    -- cycle;
  \end{scope}

  %% %%%%%%%%%%%%%%%%%
  %% dotted vertical help lines
  \draw[mode switch] (0.2,0) node[below,black] {$0$} -- ++(0,1.5);
  \draw[mode switch] (\r,0) node[below,black] {$\rname$} -- ++(0,1.5);
  \draw[mode switch] (0.5*\r+0.1,0) node[below,black] {$\tfrac{\rname}{2}$} -- ++(0,1.5);
  \draw[mode switch] (0.25*\r+0.15,0) node[below,black] {\color{\smallcoll}\tiny$\tfrac{\rname}{4}$} -- ++(0,1.5);
  \draw[mode switch] (0.75*\r+0.1,0) node[below,black] {\color{\smallcoll}\tiny$\tfrac{3}{4}\rname$} -- ++(0,1.5);
  % \path (0.2,0) %-- node[below=8pt] {$\pevolve{\D{x}=F(x)}\phantom{{}\&\ivr}$} (2,0);

  %% %%%%%%%%%%%%%%%%%
  %% begin-to-mid arc
  \draw[\arccol,thick,domain=-1:-0.2,smooth,xshift=1.2cm,yshift=-0.1cm] plot (\x,{1.5-0.3*exp(-1.4*\x)})  node[above=2pt, xshift=2pt] {~$\iget[state]{\Iz}$};
  \begin{scope}[xshift=0.2cm,yshift=-0.1cm]
    % \draw[ultra thick,domain=-1:-0.2,smooth,xshift=1cm] plot (\x,{1.5-0.3*exp(-1.4*\x)});
    \draw[thick,domain=-1:-1.01,smooth,xshift=1cm] plot (\x,{1.5-0.3*exp(-1.4*\x)})  node[above left=-2pt] {$\iget[state]{\I}$};
  \end{scope}

  %% %%%%%%%%%%%%%%%%%
  %% mid-to-end arc 
  \begin{scope}[xshift=0.2cm,yshift=-0.1cm]
    \draw[\arccol,thick,domain=-0.2:0.055,smooth,xshift=1cm] plot (\x,{1.5-0.3*exp(-1.4*\x)});
    \draw[\arccol,thick,domain=0.055:0.6,smooth,xshift=1cm] plot (\x,{1.5-0.3*exp(-1.4*\x)}) node[right=-2pt,black] {$\iget[state]{\It}$};
  \end{scope}
  
\end{tikzpicture}%
    \Description{Illustration of recursive splitting.}
    \caption{Recursive splitting of differential equation evolution: The shaded areas illustrate the local growth bound from \Cref{prop:first-order-bound}. If these bounds are satisfied recursively, \(y\) is reachable from \(x\) by \Cref{prop:ode-as-fixpoint-theorem}.}
    \label{fig:dynamics-nabla}
\end{figure}

Reasoning about differential equations is difficult, as it involves syntactically reasoning about the non-computable reachability relation along the flow of a vector field.
There are incomplete approaches using invariance reasoning \cite{DBLP:conf/lics/PlatzerT18,DBLP:journals/jacm/PlatzerT20} or using classical numerical approximation methods \cite{DBLP:conf/lics/Platzer12b}.
In contrast to these, this section presents a \emph{complete} and \emph{global} formulation of differential equations as \emph{fixpoints}.\footnote{%
    This fixpoint description is different from the contracting fixpoint description for \emph{local} existence of solutions to differential equations in the Picard-Lindel\"of theorem using the Banach fixpoint theorem.
    It describes the reachability relation as a global fixpoint in a lattice \cite{DBLP:journals/pjm/Tarski55}.}
This unifies the notion of discrete adversarial gameplay and continuous gameplay through fixpoints, by identifying a continuous dynamical system with a discrete dynamical system.

\subsubsection{Reachability Relation as Fixpoint}

The continuous reachability relation can be viewed as a fixpoint on the semantic level.
Behind this is the observation that a point \(y\) is reachable from point \(x\) along the flow of a differential equation if and only if there is a point \(u\) in the middle which is reachable from \(x\) and from which \(y\) is reachable.
The insight is that this recursive definition can be weakened \emph{locally} to require \(u\) to be within suitable bounds rather than being reachable from \(x\) and \(y\) being reachable from \(u\).
The purpose of these bounds is to ensure that if they are satisfiable recursively, then the points are reachable along the flow.
See \Cref{fig:dynamics-nabla} for an illustration.
The bounds are natural second-order growth restrictions, that follow from continuous reachability.
The next lemma recalls these second-order Taylor bounds on the flow of a continuous function:

\begin{lemmaE}\label{prop:first-order-bound}
    If $F:\reals^\ell\to\reals^\ell$ is a continuously differentiable function
    and $\gamma:[0,t]\to \reals^\ell$ is a solution of $x'=F(x)$, then
    \[\realnorm{\gamma_t-\gamma_0-t F(\gamma_0)}\leq \tfrac{t^2}{2} \supnorm[{[0,t]}]{\gamma''},\]
    where \(\supnorm[\compactset]{f}=\sup_{x\in \compactset}\realnorm{f(x)}\) is the \emph{uniform norm} for \(f:\reals^k\to\reals^\ell\) on \(\compactset\subseteq \reals^k\).
\end{lemmaE}

\begin{proofE}
    Since \(\gamma\) is a solution of $x'=F(x)$ and $F$ is continuously differentiable, the second derivative  \(\gamma_s''=D F(\gamma_s)F(\gamma_s)\) exists.
    Integration by parts gives
    \(
    \int_0^t (t-s)\gamma_s''\ud s = - t\gamma'_0+\int_0^t \gamma_s'\ud s = \gamma_t-\gamma_0- t\gamma_0'.
    \)
    Now estimate
    \begin{align*}
         & \realnorm{\gamma_t-\gamma_0- t\gamma_0'}
        \leq       \int_0^t \realnorm{(t-s) \gamma_s''}\ud s
        \leq  \supnorm[{[0,t]}]{\gamma''} \int_0^t (t-s)\ud s
        =  \tfrac{t^2}{2}\supnorm[{[0,t]}]{\gamma''}.
        \qedhere
    \end{align*}
\end{proofE}

Consider an evolution domain constraint set \(\evdset\subseteq \reals^\ell\), and let \(\derbounds=(\derbound,\secondderbound)\in\reals^2\) be constant \emph{derivative bounds}.
Then \(\evdsetbound{\derbounds} = \{x\in\evdset : \realnorm{F(x)}\leq\derbound\mand\realnorm{D F(x)\cdot F(x)}\leq \secondderbound\}\) denotes the \emph{bounded evolution domain constraint}.
\Cref{prop:first-order-bound} shows that every reachable triple \((x,y,t)\in\reachrelofin{F}{\evdsetbound{\derbounds}}\) satisfies the \emph{growth condition} \((x,y,t)\in\growthcond{F}{\derbounds}\) where
\[\growthcond{F}{\derbounds} = \{(x,y,t)\in\evdsetbound{\derbounds}\times \evdsetbound{\derbounds}\times [0,\infty) :  2\realnorm{y-x-t F(x)}\leq t^2 \secondderbound\}.\]
Hence if \((x,y,t)\notin \growthcond{F}{\derbounds}\) then \emph{no} solution with derivatives bounded by \(\derbounds\) witnessing \(\contreachin{F}{x}{y}{\evdset}\) can exist.
The exact bounded continuous reachability relation can be obtained by taking the fixpoint of the recursive splitting process as follows.

\begin{theoremE}[Fixpoint Characterization][normal]\label{prop:ode-as-fixpoint-theorem}
    Let $F:\reals^\ell\to\reals^\ell$ be continuously differentiable and~\(\evdset\) be a closed set.
    The \(\derbounds\)-bounded continuous reachability relation is a greatest fixpoint
    \(\reachrelofin{F}{\evdsetbound{\derbounds}} = \mgfp{Z}{\splitop{F}{\derbounds}(Z)}\) of the \emph{splitting operation}
    \[\splitop{F}{\derbounds}(Z) = \left\{(x,y,t)\in \growthcond{F}{\derbounds} : \mexists{u} (x,u ,\tfrac{t}{2}), (u, y,\tfrac{t}{2}) \in Z\right\}.\]
\end{theoremE}

\begin{proofE}
    For readability write \(\reachrelsymbol\) for \(\reachrelofin{F}{\evdsetbound{\derbounds}}\) and \(\splitsymbol{}{}(Z)\) for \(\splitop{F}{\derbounds}(Z)\).
    It will be shown that~\(\reachrelsymbol\) is the greatest post-fixpoint of \(\splitsymbol{}{}\).

    \emph{\(\reachrelsymbol\) is a post-fixpoint} (\(\reachrelsymbol \subseteq \splitsymbol{}{}(\reachrelsymbol{}{})\)).
    Suppose that $\gamma$ is a solution of $x'=F(x)$ witnessing \((x,{y},t)\in\reachrelsymbol\).
    Then \(x,y\in\evdset\), \(\supnorm[{[0,t]}]{\gamma'}= \supnorm[{\gamma([0,t])}]{F}\leq \derbound\) and \(\supnorm[{[0,t]}]{\gamma''}= \supnorm[{\gamma([0,t])}]{(D F)F}\leq \secondderbound\) because \(\gamma\) takes only values in \(\evdsetbound{\derbounds}\).
    Hence $(x,{y},t)\in \growthcond{F}{\derbounds}$ by \Cref{prop:first-order-bound}.
    Let $u=\gamma_{t/2}$ and note that the restrictions of $\gamma$
    to the intervals $[0,\tfrac{t}{2}]$ and
    $[\tfrac{t}{2},t]$ witness that
    \((x,u,\tfrac{t}{2}),(u,{y},\tfrac{t}{2})\in \reachrelsymbol\).
    Hence, $(x,y,t)\in \splitsymbol{}{}(\reachrelsymbol)$.

    \emph{\(\reachrelsymbol\) contains every post-fixpoint}.
    Suppose \(A\) is some post-fixpoint of \(\splitsymbol{}{}\) such that \(A\subseteq \splitsymbol{}{}(A)\) and \((x,y,t)\in A\).
    It needs to be shown that \((x,y,t)\in\reachrelsymbol\).
    Let \(\splitilengthinv[m]=t2^{-m}\).
    Let $x^0_0 = x$, $x^0_1 = y$ and by recursion on \(m\) define a sequence \(x^m_0,\ldots,x^m_{2^m}\in \reals^\ell\) such that \((x^m_i,x^m_{i+1},\splitilengthinv)\in A\) for all \(i<2^m\) as follows.
    This is possible because $(x^{m-1}_i,x^{m-1}_{i+1},\splitilengthinv[m-1])\in A=\splitsymbol{}{}(A)$ for every \(i<2^m\).
    So by definition of \(\splitsymbol\) for each $i< 2^{m-1}$ pick some $u$ such that $(x^{m-1}_{i}, u, \splitilengthinv), (u,x^{m-1}_{i+1}, \splitilengthinv) \in A$.
    Set $x^{m}_{2i} = x^{m-1}_i$ and $x^{m}_{2i+1}=u$.
    Finally, let \(x^{m}_{2^{m}}=y\).
    Since \(A\subseteq \splitsymbol{}{}(A)\subseteq \growthcond{F}{\derbounds}\) this ensures that \(x^m_i\in\evdsetbound{\derbounds}\) and \(\realnorm{x^m_{i}-x^m_{i-1} -\splitilengthinv F(x^m_{i-1})} \leq \splitilengthdoubinv\tfrac{\secondderbound}{2}\)
    for all $m$ and all $1\leq i \leq 2^m$.

    Define the sequence of piecewise constant functions \(\gamma^m : [0,t] \to \reals^\ell\) by
    \(\gamma^m_s = x^m_{\scalesplitilength}\).
    Consider some \(s\in[0,t]\) and let \(i=\floor{\splitilength s}\) and \(j=\floor{\splitilength[m-1] s}\).
    Then either \(i=2j\) or \(i=2j+1\).
    In case \(i=2j\) note \(\gamma^m_s=x^m_i=x^{m-1}_j=\gamma^{m-1}_s\) by definition of the sequence.
    In case \(i=2j+1\) note that \(\gamma^m_s=x^m_{2j+1}\) and \(\gamma^{m-1}_s=x^{m-1}_j=x^m_{2j}\).
    Then
    \begin{align*}
        \realnorm{\gamma^m_s-\gamma^{m-1}_s}\leq \realnorm{x^m_{2j+1}-x^m_{2j}}
         & \leq \realnorm{x^m_{2j+1}-x^m_{2j}-\splitilengthinv F(x^m_{2j})}  + \realnorm{\splitilengthinv F(x^m_{2j})}
        \leq \splitilengthdoubinv\tfrac{\secondderbound}{2}+ \splitilengthinv \derbound.
    \end{align*}
    The sequence of distances is summable
    \begin{align*}
        \textstyle
        \sum_{m=1}^\infty\supnorm[{[0,t]}]{\gamma^m-\gamma^{m-1}}\leq \sum_{m=1}^\infty (\splitilengthdoubinv\tfrac{\secondderbound}{2}+ \splitilengthinv \derbound) <\infty
    \end{align*}
    as \(\splitilengthinv = t2^{-m}\) is a geometric series.
    Hence, the sequence $\gamma^m$ is uniformly Cauchy on \([0,t]\).
    By completeness, the sequence \(\gamma^m\xrightarrow{m\to\infty}\gamma\) converges uniformly to some \(\gamma:[0,t]\to \evdsetbound{\derbounds}\).
    Note that \(\gamma_s\in\evdsetbound{\derbounds}\) for all \(s\in [0,t]\), since \(\evdsetbound{\derbounds}\) is closed and \(\gamma^m_s\in\evdsetbound{\derbounds}\) and all \(s\in [0,t]\) and all \(m\).

    Since $\gamma^m_r$ is piecewise constant applying the triangle inequality yields for \(n=\scalesplitilength\)
    \begin{align*}
        \realnorm{\gamma^m_s - \gamma^m_0 - \int_0^s F(\gamma^m_r) \ud r}
         & \leq
        \sum_{i=1}^{n}
        \realnorm{x^m_{i} - x^m_{i-1} - \splitilengthinv F(x_{i-1}^m)}
        +
        \realnorm{\gamma^m_s-x^m_{n}-(s-n\splitilengthinv) F(x_{n}^m)}                                   \\
         & \leq  n\splitilengthdoubinv\tfrac{\secondderbound}{2}+\realnorm{s-n\splitilengthinv}\derbound
        \leq  s\splitilengthinv\tfrac{\secondderbound}{2}+\splitilengthinv\derbound=\splitilengthinv(\tfrac{s\secondderbound}{2}+\derbound)\xrightarrow{m\to\infty}0.
    \end{align*}
    Since the functions \(\gamma^m\) are all bounded and uniformly convergent, the sequence \(F(\gamma^m)\) also converges uniformly (\Cref{prop:bound-converge-uniformly}). Hence,
    \[
        \gamma_s -\gamma_0 =
        \lim_{m\to \infty}\int_0^s F(\gamma_r^m)\ud r
        \stackrel{}{=} \int_0^s F(\gamma_r)\ud r.
    \]
    Since \(\gamma_s\in \evdsetbound{\derbounds}\) the integrand \(\realnorm{F(\gamma_s)}\leq\derbound\) is bounded and because \(\realnorm{\gamma_s-\gamma_r}=\realnorm{\int_r^s F(\gamma_r)\ud r}\leq \realnorm{r-s}\derbound\) the curve \(\gamma\) is Lipschitz continuous.
    Consequently, \(\gamma'_s = F(\gamma_s)\) by the fundamental theorem of calculus.
    Moreover, $\gamma_t = \gamma^m_t=y$ and $\gamma_0= \gamma^m_0=x$ so that $\gamma$ is a solution of $x'=F(x)$ witnessing $(x,y,t)\in \reachrelsymbol$.
\end{proofE}

The last theorem describes the reachability relation on a domain \(\evdsetbound{\derbounds}\) with bounds on the derivative as a fixpoint.
The derivative bounds are not critical and on a compact domain suitable bounds can be fixed:

\begin{corollary}\label{prop:bound-existence}
    Let $F:\reals^\ell\to\reals^\ell$ be continuously differentiable.
    Then \(\reachrelofin{F}{\evdset} = \bigcup_{\derbounds}\reachrelofin{F}{\evdsetbound{\derbounds}} \).
    If \(\evdset\) is bounded, there exist bounds \(\derbounds\) so that
    \(\reachrelofin{F}{\evdset} = \reachrelofin{F}{\evdsetbound{\derbounds}}\).
\end{corollary}

\begin{proof}
    Suppose \((x,y,t)\in \reachrelofin{F}{\evdset}\) is witnessed by \(\gamma:[0,t]\to \evdset\).
    Then \((x,y,t)\in \reachrelofin{F}{\evdsetbound{\derbounds}}\) for \(\derbounds=(\supnorm[{[0,t]}]{\gamma'},\supnorm[{[0,t]}]{\gamma''})\).
    Conversely, \(\reachrelofin{F}{\evdsetbound{\derbounds}} \subseteq \reachrelofin{F}{\evdset}\) for all \(\derbounds\).
    In case \(\closure{\evdset}\) is compact, \(\derbounds=(\supnorm[\closure{\evdset}]{\odef},\supnorm[\closure{\evdset}]{D\odef\cdot \odef})\) are suitable derivative bounds.
\end{proof}

\subsubsection{Quantitative Analysis of the Approximation}\label{sec:quantitative-approximation}
By the Knaster-Tarski theorem the greatest fixpoint of a monotone function \(F:\pow{X}\to\pow{X}\) is constructible via fixpoint iterations.
The \(n\)-th iteration of the greatest fixpoint of \(F\) is defined recursively \(\mgfp[n]{Z}{F(Z)}=F(\mgfp[n-1]{Z}{F(Z)})\) starting with \(\mgfp[0]{Z}{F(Z)}=X\).
These fixpoint iterations monotonically approximate the ODE reachability relation.
The following lemma gives quantitative bounds relating the quality of the ODE-approximation of the fixpoint of the splitting operation to the number of fixpoint iterations.

\begin{proposition}\label{prop:quantitative-diameter}
    Let \(L>0\), \(\evdset\subseteq\reals^\ell\) and \(F:\evdset\to\reals^\ell\) be \(L\)-Lipschitz. For any \(t\geq 0\) and \(x\in \evdset\):
    \[
        \mathrm{diam}\{y: (x,y,t)\in\mgfp[n]{Z}{\splitop{F}{\derbounds}(Z)}\}\leq \tfrac{\secondderbound t}{L 2^{n-1}}(e^{Lt}-1)
        \leq \tfrac{\secondderbound t^2}{2^{n-1}}e^{Lt}.
    \]
\end{proposition}

\noindent
The \emph{diameter} of a set \(X\subseteq\reals^\ell\) is \(\mathrm{diam}(X)=\sup_{x,y\in X}\realnorm{x-y}\).

\begin{proof}
    Consider some \((x,y,t),(x,\tilde{y},t)\in\mgfp[n]{Z}{\splitop{F}{\derbounds}(Z)}\).
    As in \Cref{prop:ode-as-fixpoint-theorem} define \(x^m_i\in\evdset\) and \(\tilde{x}^m_i\in\evdset\) for \(m< n\) and \(i\leq 2^m\) such that \(x^m_0=\tilde{x}^m_0=x\) and \(x^m_{2^m}=y\) and \(\tilde{x}^m_{2^m}=\tilde{y}\).
    For \(m=n-1\) the superscript \(m\) is dropped.
    By construction \((x_i,x_{i+1},t2^{-m}) ,(\tilde{x}_i,\tilde{x}_{i+1},t2^{-m}) \in  \growthcond{F}{\derbounds}\) so that
    \begin{align*}
        |x_{i}-x_{i-1} -t2^{-m} F(x_{i-1})|  \leq \tfrac{\beta}{2}
        \quad\text{and}\quad
        |\tilde{x}_{i}-\tilde{x}_{i-1} -t2^{-m} F(\tilde{x}_{i-1})|  \leq \tfrac{\beta}{2}
        \qquad \text{where } \beta=(t2^{-m})^2\secondderbound.
    \end{align*}
    for all \(0<i\leq 2^m\).
    Let \(a_i=\realnorm{x_i-\tilde{x_i}}\) so that \(a_0=0\) and
    \begin{align*}
        a_{i+1}=\realnorm{x_{i+1}-\tilde{x}_{i+1}}
         & \leq \realnorm{x_{i}-\tilde{x}_i + t2^{-m}(F(x_i)-F(\tilde{x}_i))}+\beta
        \leq \realnorm{x_{i}-\tilde{x}_i} + t2^{-m}\realnorm{F(x_i)-F(\tilde{x}_i)}+\beta
        \\
         & \leq \realnorm{x_{i}-\tilde{x}_i} + t2^{-m}L\realnorm{x_i-\tilde{x}_i}+\beta= (1+t2^{-m}L)a_i+\beta=\alpha a_i+\beta
    \end{align*}
    where \(\alpha=(1+t2^{-m}L)\).
    Solving this yields \(a_i=\sum_{k=0}^{i-1}\beta\alpha^k=\beta\tfrac{(\alpha^i-1)}{\alpha-1}=\tfrac{t\secondderbound}{2^mL}(\alpha^i-1)\).
    Then
    \[\realnorm{y-\tilde{y}}=a_{2^m}=\tfrac{t\secondderbound}{2^mL}(e^{Lt}-1)\]
    follows with \(\alpha^{2^m}\leq e^{Lt}\).
    The lemma follows by taking the supremum and since \(e^{Lt}-1\leq Le^{Lt}\).
\end{proof}

\begin{remark}[Solution Uniqueness]\label{rem:loclip-unique}
    Since the diameter vanishes as \(n\to\infty\) \Cref{prop:quantitative-diameter} and \Cref{prop:ode-as-fixpoint-theorem} also prove the uniqueness of solutions to locally Lipschitz ODEs.
\end{remark}

The idea behind the fixpoint characterization of continuous reachability is a form of dynamic programming:
A longer-duration reachability question is turned into two independent simpler (shorter-duration) problems of the same form.
Crucially an intermediate point is chosen nondeterministically.
This form of \emph{nondeterministic dynamic definition}, allows for a fast rate of convergence.
Interestingly, the required number of fixpoint iterations is \emph{linear} in time and \emph{logarithmic} in the desired accuracy.

\begin{corollary}\label{prop:quantitative-approximation-works}
    Let \(L>0\), \(\evdset\subseteq\reals^\ell\) be a set and let \(F:\evdset\to\reals^\ell\) be an \(L\)-Lipschitz function and
    \[n\geq \log_2(\varepsilon^{-1})+2\log_2(t)+2Lt+\log_2(\secondderbound)+1.\]
    Then \(\realnorm{y-z}<\varepsilon\) whenever \((x,y,t)\in\reachrelofin{F}{\evdsetbound{\derbounds}}\) and \((x,z,t)\in\mgfp[n]{Z}{\splitop{F}{\derbounds}(Z)}\).
\end{corollary}

\subsection{Axiomatization of \DifferentialGameLogic and \mutex-Calculus}\label{sec:dgl-axioms}

\subsubsection{The ODE Fixpoint Axiom}\label{sec:nabla-axiom}

The syntactic characterization of ODE reachability by fixpoints developed here uses \emph{syntactic Lie derivatives}.
These are powerful general tools for differential equations and have been used to completely axiomatize \emph{differential equation invariance} properties \cite{DBLP:journals/jacm/PlatzerT20}.
Intuitively, the Lie derivative \(\synlie{\odef}{g}\) is a useful syntactic device to describe the derivative \(\tfrac{\ud}{\ud t}g(\gamma)\) of a solution \(\gamma\) of a differential equation \(x'=\odef(x)\).
Recall the syntactic Lie derivative \cite{DBLP:journals/jacm/PlatzerT20} of a term \(\termb\) with respect to a syntactic vector field \(\odefof\) is defined to be the \emph{term}
\(\textstyle\synlie{\odef}{\termb}=\sum_{i=1}^\ell \tfrac{\partial \termb}{\partial \ivar_i} \odefi{i}\)
where the partial derivative \(\tfrac{\partial \termb}{\partial \ivar_i}\) of term \(\termb\) with respect to variable \(\ivar_i\) can be defined syntactically, when considering \(\termb\) as a polynomial in \(\ivar_i\).
Note that \(\freevars{\synlie{\odef}{\termb}}\subseteq\odevarsx\cup\freevars{\termb}\).
The \emph{syntactic} second-order Lie derivative is
\[\synliedouble{\odef}=(\synlie{\odef}{\odefi{1}},\ldots,\synlie{\odef}{\odefi{\ell}}).\]
Note that \(\freevars{\synliedouble{\odef}}\subseteq\odevarsx\) and \(\synliedouble{\odef}=\synlie{\odef}{\synlie{\odef}{\odevarsx}}\) when \(\synlie{\odef}{\cdot}\) is interpreted pointwise.
Semantically, then \(\ltden{\synliedouble{\odef}}=\realnorm{D F(x)\cdot F(x)}\) where \(F=\odefinterp\).

\begin{table}[htbp]
    \centering
    \begin{tabular}{l p{6.2cm}p{2.2cm}p{2.2cm}}
        \toprule
        \textbf{Formula} & \textbf{Formal Definition}                                                                                                                                                                                                  & \textbf{Semantics} & \textbf{Free Variables} \\
        \midrule
        \(\syntacticbounds{\evdfml}\)
                         & \(\evdfml\land\realnorm{\odef}\leq \derbound\land\realnorm{\synliedouble{\odef}}\leq \secondderbound\)
                         & \(\evdsetbound[{\evdfmlinterp}]{\derbounds} \)
                         & \(\odevarsx,\derbounds\)                                                                                                                                                                                                                                                   \\
        \addlinespace
        \(\growthbounds{\evdfml}\)
                         & \(2\realnorm{\odevarsb-\odevarsx-\timecalcvar\odef}\leq \timecalcvar^2 \secondderbound \land \syntacticbounds{\evdfml}\land\fmlreplacevarby[\syntacticbounds{\evdfml}]{\odevarsx}{\odevarsb}\)
                         & \(\growthcond[{\evdfmlinterp}]{F}{\derbounds}\)
                         & \(\odevarsx,\odevarsb,\timecalcvar,\derbounds\)                                                                                                                                                                                                                            \\
        \addlinespace
        \(\synsplitmu{\evdfml}\)
                         & \(\growthbounds{\evdfml}\land\lexists{\midpointseq}\lpossible{\gassign{\timecalcvar}{\tfrac{\timecalcvar}{2}}}(\lpossible{\gassign{\odevarsx}{\midpointseq}}\pvar\land\lpossible{\gassign{\odevarsb}{\midpointseq}}\pvar)\)
                         & \(\splitop[{\evdfmlinterp}]{F}{\derbounds}(X)\)
                         & \(\odevarsx,\odevarsb,\timecalcvar,\derbounds,\pvar\)                                                                                                                                                                                                                      \\
        \addlinespace
        \(\synsplitgame\)
                         & \(\gtest{\growthbounds{\evdfml}};\gassign{\timecalcvar}{\tfrac{\timecalcvar}{2}}\gcom{\gndassign{\midpointseq}};(\gassign{\odevarsx}{\midpointseq}\gdchoice\gassign{\odevarsb}{\midpointseq})\)
                         & -
                         & \(\odevarsx,\odevarsb,\timecalcvar,\derbounds\)                                                                                                                                                                                                                            \\
        \bottomrule
    \end{tabular}
    \caption{Abbreviations for formulas and games used in the recursive splitting for \(\stdode\).}\label{tab:splitting-formulas}
\end{table}

To formulate the fixpoint characterization of differential equation reachability syntactically in \differentialgamelogic the abbreviations of formulas from \Cref{tab:splitting-formulas} are useful.
(The notation \(\realnorm{\ivarseq}\leq y\) is an abbreviation for the formula \(\ivar_1^2+\ldots+\ivar_\ell^2\leq y^2\).)
The formula \(\syntacticbounds{\evdfml}\) describes the restricted evolution domain constraint \(\evdsetbound[{\evdfmlinterp}]{\derbounds}\) using the syntactic second-order Lie derivative.
The growth bounds \(\growthcond[{\evdfmlinterp}]{F}{\derbounds}\) of the splitting operation and the splitting game are captured by the formula \(\growthbounds{\evdfml}\).
The variable \(\timecalcvar\) is a fresh variable used to keep track of time in the recursive splitting.
For convenience the value of \(\timecalcvar\) will always be assumed to be nonnegative.
In particular, any quantifier \(\lforall{\timecalcvar}\) is implicitly understood as \(\lforall{\timecalcvar{\geq} 0}\) and \(\gndassign{\timecalcvar}\) is viewed as \(\gndassign{\timecalcvar};\gtest{\timecalcvar{\geq} 0}\).
Finally, the formulas \(\synsplitmu{\evdfml}\) and \(\synsplitgame\) characterize the splitting operation and game in \dLmu and \dGL, respectively:
\begin{align*}
    \lmden[]{\intp}{\synsplitmu{\evdfml}}
     & =
    \{\fostate : (\fostate(\odevarsx),\fostate(\odevarsb),\fostate(\timecalcvar))\in\splitop[{\evdfmlinterp}]{\odefinterp}{\fostate(\derbounds)}(\intp(\pvar))\}
    \\
    \glgden[]{\synsplitgame}(Z)
     & =\{\fostate:
    (\fostate(\odevarsx),\fostate(\odevarsb),\fostate(\timecalcvar))\in\splitop[{\evdfmlinterp}]{\odefinterp}{\fostate(\derbounds)}(Z)\}.
\end{align*}
Naturally the fixpoint of the recursive splitting is thus characterizable syntactically by a \emph{fixpoint formula} in \dLmu. Moreover it can be characterized equally well by a \emph{game} in \dGL 
\begin{align*}
    \nsynreachrelex[\mu]{\odef}{\evdfml}
    \synequiv
    \fgfp{\pvar}{\synsplitmunolabel{\evdfml}}
    \quad
    \text{ and }
    \quad
    \nsynreachrelex[\smallGame]{\odef}{\evdfml}
    \synequiv
    \lpossible{\gdrepeat{\synsplitgamenolabel}}\top.
\end{align*}
Here \(\gdrepeat{\synsplitgamenolabel}\) denotes demonic repetition defined as \(\gdualp{\gdrepeatp{\gdual{\synsplitgamenolabel}}}\), which behaves like \(\garepeat{}\), except that \Playertwo controls the iteration.
Intuitively, when playing \(\gdrepeat{\synsplitgamenolabel}\)  \Playerone must ensure that the growth bounds are maintained and must produce an intermediate witness to reachability.
\Playertwo always picks one of the two (time) interval halves and challenges \Playerone to continue producing a witness there.
In this game both players genuinely have to make choices and as such it describes a proper game with real adversarial interaction.
Thus, the fixpoint and game characterization are two related but distinct characterizations of continuous reachability differently via recursive and adversarial dynamics, respectively.
Whenever it is clear from context or irrelevant whether the fixpoint formula or game version is intended, the superscripts in \(\synreachrelsymbol^\mu\) or \(\synreachrelsymbol^\smallGame\) are dropped.
Note that \(\freevars{\nsynreachrelex{\odef}{\evdfml}}=\{\odevarsx,\odevarsb,\timecalcvar,\derbounds\}\).
Since the fixpoint characterization of \Cref{prop:ode-as-fixpoint-theorem} only applies to formulas with closed domain, the axiomatization needs to take this into account.
A formula \(\evdfml\) is \emph{closed} or \emph{open} if \(\evdfmlinterp\) is closed or open, respectively.
For closed formulas \(\evdfml\) the formula \(\nsynreachrelex{\odef}{\evdfml}\) characterizes differential equation reachability by \Cref{prop:ode-as-fixpoint-theorem}:

\begin{theorem}[Continuous Reachability via Fixpoint Definition]\label{def:reachability-formula}
    If \(\evdfmlinterp\) is closed
    \[\glfden[]{\nsynreachrelex{\odef}{\evdfml}} =
    \{\fostate: 
    \contreachin[\timecalcvar]{\odefinterp}{\fostate(\odevarsx)}{\fostate(\odevarsb)}{\evdsetbound[\evdfmlinterp]{\fostate(\derbounds)}}\}.\]
\end{theorem}

\noindent 
Thus, intuitively \(\nsynreachrelex{\odef}{\evdfml}\) can be understood as saying that \(\odevarsb\) is reachable along \(x'=\odefinterp\) from \(\odevarsx\) in domain \(\evdfml\) while respecting the bounds \(\derbound\).
\Cref{prop:bound-existence} can thus be formulated syntactically as the following sound axiom
\[
    \cinferenceRule[nabla-both|{$\langle\nabla\rangle$}]{ODE fixpoint box axiom}
    {
    \linferenceRule[viuqe]
    {\lpossible{\stdodetwith{\evdfml}}\fml}
    {\lexists{\odevarsb,\timecalcvar,\derbounds}(\nsynreachrelex{\odef}{\evdfml}\land
        \fmlreplacevarby[{\fmlreplacevarby[\fml]{\odevarsx}{\odevarsb}}]{\timevar}{\timecalcvar})}
    \quad
    }{\evdfml \text{ closed, } \odevarsb,\timecalcvar,\derbounds\notin\freevars{\fml}}.
\]

Since \(\nsynreachrelex{\odef}{\evdfml}\) can be a formula of \dLmu or of \dGL equally, axiom \irref{nabla-both} describes ODE reachability in terms of games and fixpoints.
The following axioms are slightly more general versions of \irref{nabla-both} and explicitly show how differential equation modalities correspond to fixpoints and games.

    {\renewcommand{\linferenceRuleNameSeparation}{\hspace{0.7em}}
        \begin{calculuscollection}%
            \begin{calculus}%
                \cinferenceRule[nabla|{$\langle\nabla\rangle_\mu$}]{ODE fixpoint box axiom}
                {
                \linferenceRule[viuqe]
                {\lpossible{\stdodetwith{\syntacticbounds{\evdfml}}}\fml}
                {\lexists{\odevarsb,\timecalcvar}(\fgfp{\pvar}{\synsplitmunolabel{\evdfml}}\land
                    \fmlreplacevarby[{\fmlreplacevarby[\fml]{\odevarsx}{\odevarsb}}]{\timevar}{\timecalcvar})}
                \quad
                }{$\evdfml$ closed, $\odevarsb,\timecalcvar\notin\freevars{\fml}$}
                \cinferenceRule[nabla-game|{$\langle\nabla\rangle_\smallGame$}]{ODE fixpoint box axiom}
                {
                \linferenceRule[viuqe]
                {\lpossible{\stdodetwith{\syntacticbounds{\evdfml}}}\fml}
                {
                    \lpossible{\gndassign{\odevarsb,\timecalcvar};({\gdrepeat{\synsplitgamenolabel}}
                        \gdchoice \gtest{\top});\gassign{\odevarsx,\timevar}{\odevarsb,\timecalcvar}}   \fml
                    \quad
                }}{$\evdfml$ closed, $\odevarsb,\timecalcvar\notin\freevars{\fml}$}
            \end{calculus}
        \end{calculuscollection}}

Axiom \irref{nabla-both} can then be derived from either of the more general axioms \irref{nabla} and \irref{nabla-game}.\footnote{
    The derivation syntactically mirrors \Cref{prop:bound-existence} and requires the addition of the localization axiom \(
    {\lpossible{\stdodewo}\fml}
    \lbisubjunct
    {\lexists{M}\lpossible{\stdodewith{\term \leq M}}\fml}
    \), where \(M\notin\{\odevarsx\}\cup\freevars{\fml}\).
    This axiom says that a state is reachable if it is reachable within some bounded domain.
    This reflects compactness of the domain \([0,t]\) of the solution.}
Soundness is an immediate consequence of \Cref{prop:ode-as-fixpoint-theorem}.

\begin{theorem}[ODE Fixpoint Axiom]
    The differential equation fixpoint axioms \irref{nabla} and  \irref{nabla-game} are sound for \(\dLmu\) and \(\dGL\), respectively.
\end{theorem}

This axiomatization of ODE reachability is a promising approximation method for the deductive verification of properties for ODEs and for bounded model-checking.
Unlike numerical methods like Euler's method, the solution is approximated \emph{uniformly} by recursive splitting.
Hence, behavior that appears after longer evolution times of the ODE is detected much earlier than it would be through forward approximation methods.

\subsubsection{Evolution Domain Constraints}

The \irref{nabla-both} axiom completely treats the behavior of a differential equation discretely on a \emph{closed} evolution domain constraint.
The assumption of closedness is necessary to ensure convergence of the iteration \emph{inside the domain}.
This limitation can be overcome easily by treating the evolution domain constraint separately.
To this end the axiom
\[
    \cinferenceRule[evolutiondomain|${[\&]}$]{evolution domain axiom}
    {
        \linferenceRule[viuqe]
        {\lnecessary{\stdodet}\fml}
        {\lforall{T}(\lnecessary{\stdodetwo}(\timevar\leq T\limply\evdfml)\limply\lnecessary{\stdodetwo}(\timevar\leq T \limply\fml))}
    }{}
\]
reduces a formula with \emph{any} evolution domain constraint to one without a constraint.
The axiom \irref{evolutiondomain} captures the special role of the time variable \(\timevar\) in a differential equation.
Intuitively, \(\fml\) remaining true along every evolution restricted to \(\fmlb\) is equivalent to \(\fml\) remaining true for any time~\(T\) such that \(\fmlb\) remains true for time \(T\).
The evolution domain constraint elimination axiom is similar to an axiom used in the proof of relative completeness for \dL \cite{DBLP:conf/lics/Platzer12b}.
Unlike the previous axiom for evolution domain constraints, axiom \irref{evolutiondomain} does not require nested modalities.
Moreover, axiom \irref{evolutiondomain} aligns well with \irref{nabla-both}, since in both time plays a more important explicit role than in previous axioms.
In fact, \irref{evolutiondomain} completely decouples the property \(\fml\) from the domain invariant \(\fmlb\) and links both via the time variable~\(T\).
This justifies ordinary logical reasoning \emph{inside} the evolution domain constraint.
For example, it immediately derives the intuitive formula \(\lnecessary{\stdodetbdd}\fml\limply\lnecessary{\stdodet}(\timevar\leq T\limply\fml)\), which was previously not easily derivable.
This formula is useful, as the more restrictive evolution domain makes the antecedent simpler to prove using differential invariance techniques, while the succedent is more flexible due to its weaker constraint.
As another example, the domain monotonicity axiom \irref{evdmon} is immediately derivable from \irref{evolutiondomain} and \irref{mon}:
\[
    \cinferenceRule[evdmon|{[$\&$M]}]{evolution domain monotonicity}
    {
    \linferenceRule[impl]
    {\lforall{\odevarsx}(\fmlc\limply\evdfml)}
    {(\lnecessary{\stdode}\fml\limply \lnecessary{\stdodetwith{\fmlc}}\fml)}
    }{}.
\]
Moreover, evolution domain constraints provably do not add any expressive power.
This is easily proved for the \differentialmucalculus and can easily be transferred to \differentialgamelogic thanks to the general equiexpressiveness result of \Cref{prop:glmufoequiexp}.
This demonstrates the power of the point of view of games as fixpoints.

\begin{proposition}[Evolution domain elimination]\label{prop:remove-evolution-domain}
    The \differentialmucalculus and \differentialgamelogic are equiexpressive with their evolution domain constraint-free fragments.
\end{proposition}

\begin{proof}
    In the \differentialmucalculus any modality with evolution domain constraint can be replaced locally using the equivalent rewriting of \irref{evolutiondomain}.
    By \Cref{prop:dgl-and-dlmu-equiexpressive}, \dGL and \dLmu are equiexpressive and by the same argument the evolution domain constraint-free fragment of \dGL is equiexpressive with the evolution domain constraint-free fragment of \dLmu.
    The equiexpressiveness of \dGL and its evolution domain constraint-free fragment follows from the case for \dLmu by the transitivity of equiexpressiveness.
\end{proof}

By \Cref{prop:foequiv}, the equiexpressiveness of \Cref{prop:remove-evolution-domain} immediately extends to an equivalence once the proof calculi including \irref{evolutiondomain} is defined in \Cref{sec:proof-calculi-dl}.

\subsubsection{Solution Existence and Uniqueness Axioms}

Two more standard axioms of \dL which are important  complements to \irref{nabla-both} are recalled here.
While \irref{nabla-both} completely characterizes how a differential equation evolves, it does not internalize any existence or uniqueness properties of the continuous evolution.
In order to prove solution existence, the bounded existence axiom \irref{bounded-existence} is used \cite{DBLP:journals/fac/TanP21}.
\[
    \cinferenceRule[bounded-existence|BEx]{ODE bounded existence axiom}
    {
    {\lpossible{\stdodetwith{\realnorm{\odevarsx}\leq v}}(\realnorm{\odevarsx} = v\lor\timevar\geq T)}
    }{}
\]
Intuitively, \irref{bounded-existence} asserts that a differential equation \(x'=F(x)\) either has a solution for time \(T\) on a domain bounded by \(v\) or it reaches the boundary of the domain.
In other words, a differential equation where the solution does not exist before time \(T\) must blow-up in space before then.
Thus, axiom \irref{bounded-existence} makes it possible to prove liveness properties by proving boundedness.
Axiom \irref{bounded-existence} is a slightly strengthened version of the original bound existence axiom, which read \(\lpossible{\stdode}(\lnot\evdfml\lor\timevar\geq T)\) for \(\evdfml\) bounded.
The new version of \irref{bounded-existence} proves existence \emph{within} a  domain constraint.

The second additional axiom for differential equations is the uniqueness axiom \irref{uniqueness}, which has been used to prove the completeness of \dL for differential invariants \cite{DBLP:journals/jacm/PlatzerT20}.
\[
    \cinferenceRule[uniqueness|Uniq]{ODE uniqueness axiom}
    {
        \linferenceRule[viuqe]
        {\lpossible{\stdodewith{\fml_1\land\fml_2}}\fmlc}
        {\lpossible{\stdodewith{\fml_1}}\fmlb\land\lpossible{\stdodewith{\fml_2}}\fmlc}
    }{}
\]
The idea behind \irref{uniqueness} is that a point is reachable by evolution with two different domains if it is reachable by evolution in the intersection of these domains.
This reflects the uniqueness of the solution to a locally Lipschitz differential equation.
As in \Cref{rem:loclip-unique} uniqueness is provable from the fixpoint approximation and this naturally leads to the alternative, more concise formulation of the uniqueness axiom
\[
    \cinferenceRule[approximation-uniqueness|U]{ODE uniqueness axiom}
    {
        (\nsynreachrelex{\odef}{\evdfml}\land \fmlreplacevarby[{\nsynreachrelex{\odef}{\evdfml}}]{\odevarsb}{\odevarsc})\limply \odevarsb=\odevarsc
    }{}
\]
in the language of fixpoints or games.
Axiom \irref{approximation-uniqueness} directly asserts that if \(\odevarsb\) and \(\odevarsc\) are both reachable from \(\odevarsx\) in time \(\timecalcvar\) then \(\odevarsx=\odevarsb\).
Note that \irref{approximation-uniqueness} and \irref{uniqueness} are interderivable.
Axiom \irref{uniqueness} will be used in \Cref{sec:rational-completeness} in the relative completeness of \rationaldifferentialgamelogic.

\subsubsection{Proof Calculus of \dLmuheadline and \dGLheadline}\label{sec:proof-calculi-dl}

The proof theory of \differentialgamelogic builds on the proof theory for real-closed fields.
In fact the basic proof calculus is designed to be an extension of some complete axiomatization of the theory of real-closed fields \cite{tarski_decisionalgebra51}.
Formally, the proof calculus for the \differentialmucalculus (\dLmu) consists of the axioms and proof rules for the \firstordermucalculus from \Cref{sec:mu-calculus} together with a {complete} axiomatization of the theory of real-closed fields and the differential equations axioms \irref{nabla-both}, \irref{nabla}, \irref{evolutiondomain}, \irref{bounded-existence} and \irref{uniqueness}.
The notation \(\provdlmu \fml\) indicates that \(\fml\) is provable in this calculus.
Analogously, the calculus for \differentialgamelogic (\dGL) consists of the axioms and proof rules for \firstordergamelogic from \Cref{sec:gl-calculus} together with the axioms \irref{nabla-both}, \irref{nabla-game}, \irref{evolutiondomain}, \irref{bounded-existence} and \irref{uniqueness} and a {complete} axiomatization of the theory of real-closed fields.

The modality-free formulas of \dLmu and \dGL are exactly the formulas of first-order real arithmetic.
Since both logics are based on a complete axiomatization of the theory of \(\reals\), any modality-free or real-arithmetic formula of is provable in \dLmu and \dGL.

In the differential \(\mu\)-calculus it is a straightforward consequence of the axiom \irref{evolutiondomain} that any evolution domain can be provably removed locally  (see \Cref{prop:remove-evolution-domain}).
By \irref{nabla-both} any evolution domain-free continuous action can also be provably removed equivalently locally.
Hence, \dLmu is equiexpressive and equivalent with the \firstordermucalculus \(\folmus[\dglsig]\) over \(\rfnstr\).
By the equivalence results from \Cref{sec:equiexpressivenessandequivalence}, this equiexpressiveness and equivalence carries over to~\dGL:

%########################################################

\begin{theoremE}[ODE Elimination]\label{prop:dlmu-dgl-equiexpressive}
    The following logics are equiexpressive and equivalent:
    \begin{enumerate}
        \item \differentialgamelogic \dGL
        \item the \differentialmucalculus \dLmu
        \item \firstordergamelogic \(\fogls[\dglsig]\) over \(\rfnstr\)
        \item the \firstordermucalculus \(\folmus[\dglsig]\) over \(\rfnstr\)
    \end{enumerate}
\end{theoremE}

Hence axioms \irref{nabla-both} and \irref{evolutiondomain} \emph{completely} capture the additional expressiveness of differential equations.
When viewed as fixpoints, differential equations are already available in the purely discrete \(\mu\)-calculus interpreted over \(\rfnstr\).
And when viewed as games, the assignment modality suffices to express the continuous behavior fully discretely in \differentialgamelogic.
Moreover, by the additional results in \Cref{sec:lfp-and-folmu}, differential equations can already be expressed in least fixpoint logic interpreted over the real numbers.

\section{Completion for Robust Properties}\label{sec:completion-for-robust-properties}
This section demonstrates the power of the fixpoint characterization of differential equations by showing how safety and reachability questions for differential equations can be reduced to \emph{provable} questions about \emph{finite} fixpoint iterations.
The nondeterministic dynamic approach of globally restricting the reachable set provides an efficient and complete approach.
As a consequence, the fixpoint perspective provides a complete axiomatization for open properties, which is very different from the invariance-based \dL calculus, which is also complete for open properties \cite{DBLP:journals/jacm/PlatzerQ25}.

\subsection{Approximation Axiom}

To fully utilize the fixpoint characterization of differential equations from \Cref{sec:fp-axiomatization}, a formal axiom capturing the quantitative approximation bounds from \Cref{sec:quantitative-approximation} will be introduced.
For convenience, \Cref{tab:reachability_abbreviations} defines some abbreviation formulas.
For example, the formula \(\synlipbd{\evdfml}\) says that \(\odef\) is \(L\)-Lipschitz on the domain \(\evdfml\).
Recall also that exponentiation is definable in \dLmu and \dGL \cite{DBLP:conf/lics/Platzer12b}.\footnote{Exponentiation is defined via \(e^\ivar=\ivarb\lbisubjunct\lpossible{\gassign{\ivarc}{1};\gassign{\timevar}{0};(\gode{\ivarc}{\ivarc},\gode{\timevar}{1}\gachoice\gode{\ivarc}{{-}\ivarc},\gode{\timevar}{{-}1})}(\timevar=\ivar\land \ivarc=\ivarb)\) for fresh \(z\) and \(\timevar\).}
Axiom \irref{fp-approximation-bound} exactly captures the quantitative approximation of \Cref{prop:quantitative-diameter}, which also proves that axiom \irref{fp-approximation-bound} is sound.
\[
    \cinferenceRule[fp-approximation-bound|App]{ODE approximation axiom}
    {
        (\synlipbd[\odef]{\evdfml}\land\nsynapproxreachrel[n]{\odef}{\evdfml}\land \fmlreplacevarby[{\nsynapproxreachrel[n]{\odef}{\evdfml}}]{\odevarsb}{\odevarsc})
        \limply
        \realnorm{\odevarsb-\odevarsc}< 2\secondderbound Te^{LT-n}
        \quad
    }{}
\]
\begin{table}[tbp]
    \centering
    \begin{tabular}{l p{7.2cm}p{2.5cm}}
        \toprule
        \textbf{Abbreviation}                                                                                                & \textbf{Formal Definition}                                                                                                                                                                               & \textbf{Free Variables}  \\
        \midrule
        $\synlipbd{\evdfml}$
                                                                                                                             & $\lforall{\odevarsx,\odevarsb} ((\evdfml\land\fmlreplacevarby[\evdfml]{\odevarsx}{\odevarsb})\limply\realnorm{\odef-\fmlreplacevarby[\odef]{\odevarsx}{\odevarsb}}\leq L\realnorm{\odevarsx-\odevarsb})$
                                                                                                                             &
        $L$                                                                                                                                                                                                                                                                                                                                                        \\
        \addlinespace
        $\nsynapproxreachrel[n]{\odef}{\evdfml}$
                                                                                                                             & $\fgfpnum{n}{\pvar}{\synsplitmunolabel{\evdfml}}$
        \text{ in \dLmu and } $\lnecessary{\garepeatn{({\synsplitgamenolabel})}{n}}(\growthbounds{\evdfml})$ \text{ in \dGL} & $\odevarsx,\odevarsb,\timecalcvar,\derbounds$                                                                                                                                                                                       \\
        \addlinespace
        \(\synsolex{\evdfml}\)
                                                                                                                             &
        \(
        \lnecessary{\stdodetwith{\realnorm{\odevarsx} {\leq} \derbounds_1{+}1\land t{\leq }\timecalcvar}}\syntacticbounds{\evdfml}
        \)
                                                                                                                             & $\odevarsx,\timecalcvar,\derbounds$
        \\
        \addlinespace
        $\synapprox[n]{\odef}{\evdfml}$
                                                                                                                             & $\lexists{\timecalcvar{\leq}T}(\synsolex{\derbounds}\land\nsynapproxreachrel[n]{\odef}{\evdfml}{})$                                                                                                      & $\odevarsx,T,\derbounds$ \\
        \addlinespace
        \(\fmlepsminus{\fml}{\varepsilon}\)
                                                                                                                             &
        \(\lforall{\odevarsc} (\realnorm{\odevarsx-\odevarsc}<\varepsilon \limply \fmlreplacevarby[\fml]{\odevarsx}{\odevarsc})\) for \(\odevarsc\) fresh
                                                                                                                             & $\odevarsx,\varepsilon,\freevars{\fml}$                                                                                                                                                                                             \\
        \bottomrule
    \end{tabular}
    \caption{Abbreviations for auxiliary formulas}\label{tab:reachability_abbreviations}
\end{table}
The formula \(\nsynapproxreachrel[\ivarnat]{\odef}{\evdfml}\) describes the \(\ivarnat\)-th approximation to the reachability relation.
In the case of \dLmu, the finite fixpoint unfolding \(\nsynapproxreachrel[\ivarnat]{\odef}{\evdfml}\) is defined as \(\fgfpnum{\ivarnat}{\pvar}{\synsplitmunolabel{\evdfml}}\), where the \(\ivarnat\)-fold iteration of a fixpoint is defined as \(\fgfpnum{\ivarnat}{\pvar}{\fmlc}\synequiv\fgfp{\pvar}{(\ivarnat=0\lor \lpossible{\gassign{\ivarnat}{\ivarnat-1}}\fmlc)}\).
In \dGL, the finite approximation \(\nsynapproxreachrel[\ivarnat]{\odef}{\evdfml}\) is defined by\(\lnecessary{\garepeatn{({\synsplitgamenolabel})}{n}}(\growthbounds{\evdfml})\), where the \(\ivarnat\)-fold iteration game is \(\garepeatn{\game}{\ivarnat}\synequiv\garepeat{(\gtest{\ivarnat\neq 0};\game;\gassign{\ivarnat}{\ivarnat-1})}\).
Axiom \irref{fp-approximation-bound} can be viewed as a quantitative refinement of the uniqueness axiom in the form \irref{approximation-uniqueness}.

For the rest of this article, the calculi for \dLmu and \dGL are assumed to include axiom \irref{fp-approximation-bound}.
From these axioms a useful approximation axiom is derivable.
If a solution exists for time \(\timecalcvar\) and \(\odevarsb\) satisfies the approximate fixpoint reachability relation $\nsynapproxreachrel[n]{\odef}{\evdfml}$, then \(\odevarsb\) is almost reached.
\Cref{tab:reachability_abbreviations} defines \(\synapprox[n]{\odef}{\evdfml}\) to abbreviate the assumptions and axiom \irref{approx} captures that a point close to \(\odevarsb\) is reached.
The assumption of solution existence is ensured by \(\synapprox[n]{\odef}{\evdfml}\) in conjunction with \irref{bounded-existence}.

\begin{theorem}\label{prop:approx-soundness-quant}
    The differential equation approximation axiom \irref{approx} is derivable in \(\dLmu\):
    \[
        \dinferenceRule[approx|{$\lpossible{\approx}$}]{ODE approximation axiom}
        {
        \linferenceRule[impl]
        {(\synlipbd{\evdfml}\land\synapprox[n]{\odef}{\evdfml})}
        {\lpossible{\stdodetwith{\evdfml\land t{\leq} T}}\realnorm{\odevarsx-\odevarsb}< 2\secondderbound Te^{LT-n}}\quad
        }{\text{\(\evdfml\) bounded} }.
    \]
\end{theorem}

\begin{proof}
    Note that \(\provdlmu\syntacticboundsnoevd\limply\realnorm{\odevarsx}\leq \derbound\land\timevar{\leq}\timecalcvar\). Thus by \irref{bounded-existence} and \irref{nabla} deduce \(\provdlmu \synsolex{\evdfml}\limply\lpossible{\stdodetwith{\realnorm{\odevarsx}{\leq\derbound{+}1}}}\timevar\geq\timecalcvar\).
    Using \Cref{prop:bound-existence} and axioms \irref{evolutiondomain} and \irref{bounded-existence} show \(\provdlmu \synsolex{\evdfml}\limply\lpossible{\stdodetwith{\syntacticboundsnoevd\land\timevar\leq\timecalcvar}}\timevar\geq\timecalcvar\).
    Now using axioms \irref{nabla} and \irref{fp-approximation-bound} derive
    \(\provdlmu (\synlipbd{\evdfml}\land\synsolex{\evdfml}\land \nsynapproxreachrel[n]{\odef}{\evdfml}{})\limply\lpossible{\stdodetwith{\syntacticboundsnoevd\land\timevar\leq\timecalcvar}}\realnorm{\odevarsx-\odevarsb}< 2\secondderbound Te^{L\timecalcvar-n}.\)
    By definition of \(\synsolex{\evdfml}\) and using \irref{evdmon} the evolution domain is not left \(\provdlmu \synsolex{\evdfml}\limply\lnecessary{\stdodetwith{\syntacticboundsnoevd\land t{\leq }\timecalcvar}}\evdfml\) so that
    \(\provdlmu (\synlipbd{\evdfml}\land 0{\leq}\timecalcvar{\leq}T\land\synsolex{\evdfml}\land \nsynapproxreachrel[n]{\odef}{\evdfml}{})\limply\lpossible{\stdodetwith{\syntacticbounds{\evdfml}\land\timevar\leq T}}\realnorm{\odevarsx-\odevarsb}< 2\secondderbound Te^{LT-n}\)
    by axiom \irref{evdmon}.
    Axiom \irref{approx} is derived by introducing the existential quantifier of \(\synapprox[n]{\odef}{\evdfml}\).
\end{proof}

Note that \irref{fp-approximation-bound} and \irref{approx} can be understood in two different ways.
Either (as above) it can be viewed as a single formula or it can be viewed equivalently as an axiom schema with an instance for every value \(\ivarnat\in\naturals\).
From this perspective, the formula \(\nsynapproxreachrel[n]{\odef}{\evdfml}\) can be viewed as a fixpoint-free formula in \dLmu by inductively defining the fixpoint approximations
\[\fgfpnum{0}{\pvar}{\fmlc} \synequiv\top\qquad\fgfpnum{\ivarnat+1}{\pvar}{\fmlc}\synequiv
    \fmlreplacepvarby[\fmlc]{\pvar}{\fgfpnum{\ivarnat}{\pvar}{\fmlc}}.\]
While this is provably equivalent to the above definition of finite fixpoint iteration formulas, it is only possible for constants \(\ivarnat\).
Importantly, this shows that \(\nsynapproxreachrel[n]{\odef}{\evdfml}\) is essentially a formula of real arithmetic, if \(\evdfml\) is a formula of real arithmetic and \(\ivarnat\) is a constant.

\subsection{Completeness for Robust Properties}

As suggested by \Cref{prop:quantitative-diameter}, the fixpoint characterization is of a numerical, discrete nature.
It makes it possible to uniformly prove many kinds of properties of differential equations, which are robust in an appropriate sense, so that they can be proved either through \emph{finite} unrolling using \irref{fp} or the approximation axiom \irref{approx}.
The following theorems show that with the addition of \irref{approx} \differentialgamelogic and the \differentialmucalculus are complete for robust properties.

\subsubsection{Completeness for Safety}

Of particular interest in verification are bounded safety properties.
It is first shown that a safety assertion \(\fmlc\limply\lnecessary{\stdode}\fmlb\) is provable, whenever a  slight strengthening of it is valid.
As in \Cref{tab:reachability_abbreviations} the formula \(\fmlepsminus{\fml}{\varepsilon}\) defines the set of points that satisfy \(\fml\) and have distance \({\geq} \varepsilon\) to the boundary of \(\glfden[]{\fml}\).
A formula \(\evdfml\) is \emph{bounded} or \emph{compact} if \(\evdfmlinterp\) is bounded or compact, respectively.

\begin{theorem}[Provable Safety]\label{prop:provable-safety}
    Let \(\fml,\evdfml,\fmlc\) be real-arithmetic formulas where \(\evdfml\) is compact and let \(T,\varepsilon\in \rationals\) be positive rational constants. Then
    \[
        \text{if }\valdlmu\fmlc\limply\lnecessary{\stdodetbdd}\fmlepsminus{\fml}{\varepsilon}\; \text{ then } \;
        \provdlmu\fmlc\limply\lnecessary{\stdodetbdd}\fml.
    \]
\end{theorem}

\begin{proof}
    Note that by boundedness of \(\evdfmlinterp\) the vector field \(F=\odefinterp\) is \(L\)-Lipschitz on \(\evdfmlinterp\) for some~\(L\).
    Suppose the formula \(\fmlc\limply\lnecessary{\stdodetbdd}\fmlepsminus{\fml}{\varepsilon}\) is valid.
    By boundedness of \(\evdfmlinterp\) and the argument of \Cref{prop:bound-existence} there are bounds \(\conderbounds\in\rationals^2\) such that \(\evdfmlinterp=\evdsetbound[\evdfmlinterp]{\conderbounds}\).
    Thus, also \(\provdlmu\fmlb\limply\fmlreplacevarby[\syntacticbounds{\evdfml}]{\derbounds}{\conderbounds}\) by completeness of real arithmetic.

    By \Cref{prop:quantitative-approximation-works} the formula \(\fmlc\limply\lforall{\odevarsb}\lforall{\timecalcvar{\leq} T}(\fmlreplacevarby[{\nsynapproxreachrel[n]{\odef}{\evdfml}}]{\derbounds}{\conderbounds}\limply\fmlreplacevarby[{\fmlreplacevarby[\fml]{\odevarsx}{\odevarsb}}]{\timevar}{\timecalcvar})\) is valid for some sufficiently large \(n\) and by completeness for real arithmetic this formula is also provable in the \dLmu calculus.
    Thus \(\provdlmu\fmlc\limply\lforall{\odevarsb}\lforall{\timecalcvar{\leq} T}(\fmlreplacevarby[\nsynreachrelex{\odef}{\evdfml}]{\derbounds}{\conderbounds} \limply\fmlreplacevarby[{\fmlreplacevarby[\fml]{\odevarsx}{\odevarsb}}]{\timevar}{\timecalcvar})\)
    since an \(n\)-fold application of \irref{fp} derives \(\provdlmu \nsynapproxreachrel[n]{\odef}{\evdfml}\limply\nsynreachrelex{\odef}{\evdfml}\).
    With axiom \irref{nabla} derive \(\provdlmu\fmlc\limply\lnecessary{\stdodetwith{\fmlreplacevarby[\syntacticbounds{\evdfml}]{\derbounds}{\conderbounds}\land t{\leq}T}}\fml\).
    Finally, the required formula is provable with \irref{evdmon} since \(\provdlmu\fmlb\limply\fmlreplacevarby[\syntacticbounds{\evdfml}]{\derbounds}{\conderbounds}\).
\end{proof}

\noindent 
Under the right topological assumption, \Cref{prop:provable-safety} yields completeness for safety.

\begin{corollary}[Completeness for Safety]\label{prop:safety-complete}
    Let \(\fml,\evdfml,\fmlc\) be real-arithmetic formulas where \(\fmlc,\evdfml\) are compact and \(\fml\) is open.
    Let \(T\in \rationals\) be a rational constant. Then
    \[
        \valdlmu\fmlc\limply\lnecessary{\stdodetbdd}\fml \;\text{ if and only if }\;
        \provdlmu\fmlc\limply\lnecessary{\stdodetbdd}\fml.
    \]
\end{corollary}

\begin{proof}
    The backward implications follow from the soundness of the \dGL calculus.
    Since \(\glfden[]{\fmlc}\) is compact, by \Cref{prop:ode-continuity-properties}\eqref{comp2} there is some \(\varepsilon>0\) so that
    \(\valdlmu\fmlc\limply\lnecessary{\stdodetbdd}\fmlepsminus{\fml}{\varepsilon}.\)
    The desired provability follows by \Cref{prop:provable-safety}.
\end{proof}

\begin{remark}[Invariant Generation]\label{rem:invariant-generation}
    The finite fixpoint approximation provides a useful new technique for \emph{generating} differential invariants.
    If the formula \(\fmlc\limply\lnecessary{\stdodetbdd}\fml\) is valid and \(\fml,\evdfml,\fmlc\) are as in \Cref{prop:safety-complete}, then by \Cref{prop:quantitative-approximation-works} there is~\(n\) such that \(\fmld\synequiv\lforall{\odevarsb}\lforall{\timecalcvar{\leq} t}(\fmlreplacevarby[{\nsynapproxreachrel[n]{\odef}{\evdfml}}]{\derbounds}{\conderbounds}\limply\fmlreplacevarby[{\fmlreplacevarby[\fml]{\odevarsx}{\odevarsb}}]{\timevar}{\timecalcvar})\) is a semialgebraic differential invariant and \(\provdlmu\fmlc\limply\fmld\) and \(\provdlmu\fmld\limply\fml\).
    By \cite{DBLP:journals/jacm/PlatzerT20} this technique can be used to prove \(\provdlmu\fmlc\limply\lnecessary{\stdodetbdd}\fml\).
    Thus, the fixpoint formulation can be used to construct differential invariants that witness safety.
    Moreover, this approach has strong existence guarantees.
\end{remark}
    
    \subsubsection{Completeness for Reachability}

Similarly to \Cref{prop:provable-safety}, reachability assertions are also provable, if a minimal strengthening is valid.
However the proof is very different and relies on axiom \irref{approx} to ensure existence.

\begin{theorem}[Provable Reachability]\label{prop:provable-completeness}
    Let \(\fml,\evdfml,\fmlc\) be real-arithmetic formulas where \(\evdfml\) is bounded and \(T,\varepsilon\in \rationals\) are positive constants. Then:
    \[
        \text{If }\valdlmu\fmlc\limply\lpossible{\stdodetwith{\fmlepsminus{\evdfml}{\varepsilon}\land \timevar{\leq} T}}\fmlepsminus{\fml}{\varepsilon}\; \text{ then } \;
        \provdlmu\fmlc\limply\lpossible{\stdodetbdd}\fml.
    \]
\end{theorem}

\begin{proof}
    Note that \(F=\odefinterp\) is continuously differentiable on the bounded set~\({\evdfmlinterp}\).
    Thus, there is some \(L\in\rationals\) so that \(F\) is \(L\)-Lipschitz on \({\evdfmlinterp}\) and by completeness for first-order real arithmetic \(\provdlmu\synlipbd{\evdfml}\).
    As in the proof of \Cref{prop:provable-safety} pick \(\conderbounds\in\rationals^2\) such that \(\provdlmu\fmlb\limply\fmlreplacevarby[\syntacticbounds{\evdfml}]{\derbounds}{\conderbounds}\).
    By the assumption together with \Cref{prop:quantitative-approximation-works} pick some sufficiently large \(n\) such that \(2\conderbounds_2T^2e^{LT-n}<\varepsilon\) and so that
    \begin{equation}\label{eq:provable-completeness-real}
        \fmlc\limply\lexists{s{\leq} T}(\lforall{\odevarsb}\lforall{\timecalcvar{\leq} s}(
        \fmlreplacevarby[{
                \nsynapproxreachrel[n]{\odef}{\realnorm{\odevarsx}{\leq}\derbound{+}1}}]{\derbounds}{\conderbounds}
        \limply
        \fmlreplacevarby[\evdfml]{\odevarsx}{\odevarsb})
        \land\lexists{\odevarsb}(
        \fmlreplacevarby[{\fmlreplacevarby[{\nsynreachapproxrelnobound[n]{\odef}}]{\derbounds}{\conderbounds}}]{\timecalcvar}{s}
        \land\fmlreplacevarby[{\fmlreplacevarby[\fmlepsminus{\fml}{\varepsilon}]{\odevarsx}{\odevarsb}}]{\timevar}{s}))
    \end{equation}
    is a valid formula.
    By completeness for real arithmetic \eqref{eq:provable-completeness-real} is provable in the \dLmu calculus.
    By \(n\)-fold application of \irref{fp} derive \(\provdlmu\nsynapproxreachrel[n]{\odef}{\realnorm{\odevarsx}{\leq}\derbound{+}1}\limply\nsynreachrelex{\odef}{\realnorm{\odevarsx}{\leq}\derbound{+}1}\).
    Then \(\provdlmu(\lforall{\odevarsb}\lforall{\timecalcvar{\leq} s}(
        \fmlreplacevarby[{
                \nsynapproxreachrel[n]{\odef}{\realnorm{\odevarsx}{\leq}\derbound{+}1}}]{\derbounds}{\conderbounds}
        \limply
        \fmlreplacevarby[\evdfml]{\odevarsx}{\odevarsb}))
    \limply\synsolex{\evdfml}\) is derivable by \irref{nabla} and~\irref{mon} using \(\provdlmu\evdfml\limply\syntacticbounds{\evdfml}\).
    Combined with provability of \eqref{eq:provable-completeness-real} and provable Lipschitz continuity, derive
    \(\provdlmu\fmlc\limply(\synlipbd{\evdfml}\land\lexists{\odevarsb}\lexists{\timecalcvar{\leq} T}(\synsolex{\evdfml}\land
        \fmlreplacevarby[{\nsynreachapproxrelnobound[n]{\odef}}]{\derbounds}{\conderbounds}
        \land\fmlreplacevarby[{\fmlreplacevarby[\fmlepsminus{\fml}{\varepsilon}]{\odevarsx}{\odevarsb}}]{\timevar}{\timecalcvar}
        )).\)
    Now axiom \irref{approx} yields provability of
    \(\provdlmu\fmlc\limply\lexists{\odevarsb}(\lpossible{\stdodetbdd}\realnorm{\odevarsx-\odevarsb}{<} \varepsilon \land\fmlreplacevarby[{\fmlreplacevarby[\fmlepsminus{\fml}{\varepsilon}]{\odevarsx}{\odevarsb}}]{\timevar}{\timecalcvar})\).
    By definition of \(\fmlepsminus{\fml}{\varepsilon}\) and \irref{evdmon}, it follows that \(\provdlmu\fmlc\limply\lpossible{\stdodetbdd}\fml\).
\end{proof}

\noindent As for safety, reachability is complete assuming suitable topological conditions.

\begin{corollary}[Completeness for Reachability]\label{prop:diacompnice}
    Let \(\fml,\evdfml,\fmlc\) be real-arithmetic formulas where~\(\fmlc\) is compact and \(\fml,\evdfml\) are open.
    Then
    \[
        \valdlmu\fmlc\limply\lpossible{\stdodet}\fml \;\text{ if and only if }\;
        \provdlmu\fmlc\limply\lpossible{\stdodet}\fml.
    \]
\end{corollary}

\begin{proof}
    The backward direction is immediate by soundness of \dGL.
    So suppose \(\valdlmu\fmlc\limply\lpossible{\stdodet}\fml\).
    By continuity properties (\Cref{prop:ode-continuity-properties}\eqref{comp1}) of differential equation flows, there are positive constants \(T,\varepsilon\in\rationals\)  such that \(\valdlmu\fmlc\limply\lpossible{\stdodetwith{\fmlepsminus{\evdfml}{\varepsilon}\land \timevar{\leq} T}}\fmlepsminus{\fml}{\varepsilon}\).
    The desired provability follows by \Cref{prop:provable-completeness}.
\end{proof}

\subsubsection{Completeness for Unsafety and Unreachability}

Dually to the completeness results for safety and reachability, their negations can also be proved whenever valid.
The proofs are very similar.

\begin{theorem}[Unsafety]\label{prop:provable-unsafety}
    Let \(\fml,\evdfml,\fmlc\) be real-arithmetic formulas where~\(\fml\) is closed and \(\evdfml\) is bounded and open. Then
    \[
        \not\valdlmu\fmlc\limply\lnecessary{\stdodet}\fml \;\text{ if and only if }\;
        \provdlmu\lnot\lforall{\odevarsx}(\fmlc\limply\lnecessary{\stdodet}\fml).
    \]
\end{theorem}

\begin{proof}
    The backward direction holds by soundness for \dLmu.
    As in the proof of \Cref{prop:provable-safety} pick \(L,\conderbounds\in\rationals^2\) such that \(\provdlmu\synlipbd{\evdfml}\) and \(\provdlmu\fmlb\limply\fmlreplacevarby[\syntacticbounds{\evdfml}]{\derbounds}{\conderbounds}\).
    For the forward direction let \(F=\odefinterp\) and pick a state \(\fostate\in\glfden[]{\fmlc}\) and a solution \(\gamma:[0,T]\to\evdfmlinterp\) of \(x'=F(x)\) such that \(\gamma_0=\fostate(\odevarsx)\) and \(\gamma_T\in\glfden[]{\lnot\fml}\).
    By \Cref{prop:ode-continuity-properties}\eqref{comp1} there is a positive constant \(\varepsilon\in\rationals\) so that \(\gamma_T\in\glfden[]{\fmlepsminus{(\lnot\fml)}{\varepsilon}}\) and \(\gamma_t\in \glfden[]{\fmlepsminus{\evdfml}{\varepsilon}}\) for \(t\in[0,T]\).
    By \Cref{prop:quantitative-approximation-works} for sufficiently large \(n\) satisfying \(2\conderbounds_2T^2e^{LT-n}<\varepsilon\) then \(\fostate\valdlmu\fmld\), where
    \[\fmld\synequiv\lforall{\odevarsb}\lforall{\timecalcvar{\leq} T}(
        \fmlreplacevarby[{
                \nsynapproxreachrel[n]{\odef}{\realnorm{\odevarsx}{\leq}\derbound{+}1}}]{\derbounds}{\conderbounds}
        \limply
        \fmlreplacevarby[\evdfml]{\odevarsx}{\odevarsb})
        \land\lexists{\odevarsb}(
        \fmlreplacevarby[{\fmlreplacevarby[{\nsynreachapproxrelnobound[n]{\odef}}]{\derbounds}{\conderbounds}}]{\timecalcvar}{T}
        \land\fmlreplacevarby[{\fmlreplacevarby[\fmlepsminus{(\lnot\fml)}{\varepsilon}]{\odevarsx}{\odevarsb}}]{\timevar}{T}).\]
    By completeness for real arithmetic \(\provdlmu\lexists{\odevarsx}(\fmlc\land\fmld)\).
    Deduce \(\provdlmu\lexists{\odevarsx}(\fmlc\land\lpossible{\stdodet}\lnot\fml)\) just as in the proof of \Cref{prop:provable-completeness}.
\end{proof}

\begin{theorem}[Unreachability]\label{prop:provable-unreachability}
    Let \(\fml,\evdfml,\fmlc\) be real-arithmetic formulas where~\(\fml\) is closed, \(\evdfml\) is compact and \(T\in\rationals\). Then
    \[
        \not\valdlmu\fmlc\limply\lpossible{\stdodetbdd}\fml 
        \;\text{ if and only if }\;
        \provdlmu\lnot\lforall{\odevarsx}(\fmlc\limply\lpossible{\stdodetbdd}\fml).
    \]
\end{theorem}

\begin{proof}
    The backward direction holds by soundness of \dLmu.
    For the forward implication assume \(\lexists{\odevarsx}(\fmlc\land\lnecessary{\stdodetbdd}\lnot\fml)\) is valid.
    As in the proof of \Cref{prop:provable-safety} there are bounds \(\conderbounds\in\rationals^2\) such that \(\provdlmu\fmlb\limply\fmlreplacevarby[\syntacticbounds{\evdfml}]{\derbounds}{\conderbounds}\).
    By \Cref{prop:ode-continuity-properties}\eqref{comp2} the formula \(\lexists{\odevarsx}(\fmlc\land\lnecessary{\stdodetwith{\fmlreplacevarby[\syntacticbounds{\evdfml}]{\derbounds}{\conderbounds}\land\timevar{\leq} T }}\fmlepsminus{(\lnot\fml)}{\varepsilon})\) is also valid. 
    By \Cref{prop:quantitative-approximation-works}
    \(
    \valdlmu\lexists{\odevarsx}(\fmlc\land\lforall{\odevarsb},\lforall{\timecalcvar{\leq} T}(\fmlreplacevarby[{\nsynapproxreachrel[n]{\odef}{\evdfml}}]{\derbounds}{\conderbounds}\limply\fmlreplacevarby[(\lnot\fml)]{\odevarsx}{\odevarsb}))
    \)
    holds for some sufficiently large \(n\).
    As a formula of real arithmetic this is provable and by \(n\)-fold application of \irref{fp} the formula
    \(\provdlmu\lexists{\odevarsx}(\fmlc\land\lforall{\odevarsb},\lforall{\timecalcvar{\leq} T}(\fmlreplacevarby[\nsynreachrelex{\odef}{\evdfml}]{\derbounds}{\conderbounds}\limply\fmlreplacevarby[(\lnot\fml)]{\odevarsx}{\odevarsb}))\)
    is derivable.
    Thus, \(\provdlmu\lexists{\ivarseq}(\fmlc\land\lnecessary{\stdodetwith{\fmlreplacevarby[\syntacticbounds{\evdfml}]{\derbounds}{\conderbounds}\land\timevar{\leq}T}}\lnot\fml)\) is derivable with \irref{nabla}.
    The claim follows with \irref{evdmon} since \(\provdlmu\fmlb\limply\fmlreplacevarby[\syntacticbounds{\evdfml}]{\derbounds}{\conderbounds}\).
\end{proof}

The results of this section illustrate the advantage of the \dLmu perspective for theoretical arguments. 
The lack of composite games simplifies the reasoning.
Nonetheless, the results can be transferred to \dGL exactly by the equivalence of \dLmu and \dGL.

\subsection{Decidability of Robust Formulas}

As has been observed \cite{DBLP:conf/csl/Franzle99}, the complexity (and incompleteness) of many hybrid systems properties can be traced back to idealized and precise models, which require correctness without any tolerance for error.
However, this precision is unrealistic, as many physical hybrid systems exhibit at least some noise.
The subtle differences between the topological conditions required to ensure provability of a safety/reachability property (in \Cref{prop:safety-complete,prop:diacompnice}) and those required to prove the negations (in \Cref{prop:provable-unsafety,prop:provable-unreachability}) show precisely which boundary phenomena can be problematic.
This hints at an infinitesimal notion of robustness, which excludes such behavior and thereby ensures that the class of robust properties can be handled completely.

Consider for example the simple property \({\ivar}{\geq} 1\limply\lnecessary{\godeev{\ivar}{{-}1}{\timevar{\leq} 1}}\ivar{\geq} 0\), which says that the value of \(\ivar\) stays nonnegative along \(\gode{\ivar}{-1}\) for one time unit as long as \(\ivar{\geq} 1\) initially.
However, the minimal strengthening \({\ivar}{{\geq}} 1\limply\lnecessary{\godeev{\ivar}{{-}1}{\timevar{\leq} 1}}\ivar{>} 0\) is not true, since \(\ivar\) does not need to remain strictly positive.
Dually, the reachability property \({\ivar}{\geq} 0\limply\lpossible{\godeev{\ivar}{1}{\timevar{\leq} 1}}\ivar{\geq} 1\) is valid, while \({\ivar}{\geq} 0\limply\lpossible{\godeev{\ivar}{1}{\timevar{\leq} 1}}\ivar{>} 1\) is not.
These examples illustrate the kinds of fragile properties that are difficult for symbolic methods.
Robustness is exactly the condition that these phenomena are excluded. For example, \({\ivar}{\geq} 0\limply\lpossible{\godeev{\ivar}{1}{\timevar{\leq} 1}}\ivar{\geq} 0.99\) is a robust version of the reachability property \({\ivar}{\geq} 0\limply\lpossible{\godeev{\ivar}{1}{\timevar{\leq} 1}}\ivar{\geq} 1\).

To formally define the class of robust properties, define the topological closure \(\fmlcl{\fml}\) of a \dLmu or \dGL formula \(\fml\) syntactically by \(\fmlcl{\fml}\synequiv\lforall{\varepsilon{>}0}\lexists{\odevarsb} (\realnorm{\odevarsb-\odevarsx}{<}\varepsilon\land\fmlreplacevarby[\fml]{\odevarsb}{\odevarsx})\).
The interior of a formula is \(\fmlint{\fml}\synequiv\lnot\fmlcl{\lnot\fml}\).
If \(\fml\) is a formula of real-arithmetic, then so are \(\fmlcl{\fml}\) and \(\fmlint{\fml}\).

\begin{definition}[Robustness]\label{def:robustness}
    Let \(T\in\rationals\) and let \(\fmlc,\evdfml,\fml\) be real arithmetic formulas where \(\fmlc\) is compact, \(\fml\) is open and \(\evdfml\) is bounded. Then
    \begin{enumerate}
        \item the safety property \(\fmlc\limply\lnecessary{\stdodetbdd}\fml\) is \emph{robust} if \(\evdfml\) is closed and
              \[\valdgl\fmlc\limply\lnecessary{\stdodetintbdd}{\fmlcl{\fml}} \quad\Longrightarrow \quad
                  \valdgl\fmlc\limply\lnecessary{\stdodetbdd}{{\fml}}\]
        \item the reachability property \(\fmlc\limply\lpossible{\stdodet}\fml\) is \emph{robust} if \(\evdfml\) is open and
              \[\valdgl\fmlc\limply\lpossible{\stdodetclbdd}{\fmlcl{\fml}} \quad\Longrightarrow \quad \valdgl\fmlc\limply\lpossible{\stdodetbdd}{{\fml}}\]
    \end{enumerate}
\end{definition}

Robustness means that boundary points can safely be discarded.
For example, if a property stays away from the boundary of the safe region, then it is robustly true.
Under the assumption of robustness, axiom \irref{nabla-both} can prove or disprove all safety and reachability properties.
Some restriction on the class of properties is necessary since many safety and reachability properties of solutions to polynomial differential equations are not even semidecidable \cite{DBLP:journals/aam/GracaCB08}.

\begin{corollary}[Completeness and Decidability of Robust Properties]\label{prop:robust-completeness}
    \Differentialgamelogic is complete for robust safety and reachability properties and their negations.
    Moreover, the validity of robust properties is decidable by proof search.
\end{corollary}

\begin{proof}
    Suppose \(\fmlc\limply\lnecessary{\stdodetbdd}\fml\) is a robust box property.
    If it is valid, then it is provable by \Cref{prop:safety-complete}.
    If the robust box property is not valid, then \(\fmlc\limply\lnecessary{\stdodetintbdd}{\fmlcl{\fml}}\) is also not valid.
    By \Cref{prop:provable-unsafety} the negation \(\lexists{\odevarsx}(\fmlc\not\limply\lnecessary{\stdodetintbdd}{\fmlcl{\fml}})\) is provable.
    By \irref{mon} and \irref{evdmon} it follows that \(\provdgl\lexists{\odevarsx}(\fmlc\not\limply\lnecessary{\stdodetbdd}{\fml})\).

    Next let  \(\fmlc\limply\lpossible{\stdodetbdd}\fml\) be a robust diamond property.
    If it is valid, it is provable by \Cref{prop:diacompnice}.
    Suppose then that the diamond property is not valid.
    Then \(\fmlc\limply\lpossible{\stdodetclbdd}{\fmlcl{\fml}}\) is not valid by robustness.
    The negation \(\lexists{\odevarsx}(\fmlc\not\limply\lpossible{\stdodetclbdd}{\fmlcl{\fml}})\) is provable by \Cref{prop:provable-unreachability}.
    By \irref{mon} and \irref{evdmon}
    \(\lexists{\odevarsx}(\fmlc\not\limply\lpossible{\stdodetbdd}{\fml})\) is also provable.
\end{proof}

\Cref{prop:robust-completeness} shows that for \emph{physical} (robust) models of hybrid systems, the  verification part of the continuous dynamics can be handled completely.
Axiom \irref{approx} can be used in the hybrid systems theorem prover KeYmaera~X \cite{DBLP:conf/cade/FultonMQVP15} to enable complete and automatic proofs for the continuous parts of a hybrid systems model.

%########################################################

\section{Relative Completeness for Rational \dGLheadline}\label{sec:rational-completeness}

\Differentialgamelogic is complete relative to any differentially expressive fragment \cite{DBLP:journals/tocl/Platzer15}, which is any expressively complete fragment that can provably represent differential equations syntactically.
However, the difficulty of finding differentially expressive fragments means that no relative completeness result with respect to a small natural logic for \dGL was known.
In this section, it is shown that the restriction of \differentialgamelogic to \emph{rational} gameplay is complete relative to single player \differentialgamelogic.
This extends the relative completeness of \differentialdynamiclogic \cite{DBLP:conf/lics/Platzer12b} to the \emph{adversarial} context of \differentialgamelogic.

\emph{\Differentialdynamiclogic} (\dL) \cite{DBLP:conf/lics/Platzer12b} is the first dynamic logic for verification of hybrid systems and is used for interactive theorem proving in cyber-physical systems verification.
Formally \dL is the fragment of \differentialgamelogic without the duality operator \(\gdual{}\) and thus single player without adversarial dynamics.
By \Cref{prop:glmufoequiexp} \dL corresponds to some fragment of the \differentialmucalculus.
To specify this fragment precisely, let the set of \emph{\angelicformulas} formulas of \dLmu be defined by the following grammar
\begin{align*}
    \fml \grassign \pvar \| \atfml \| \lnot \atfml \| \fml_1\lor\fml_2 \| \lpossible{\gndassign{\ivar}}\fml \| \lpossible{\stdode}\fml  \| \flfp{\pvar}{\fml}.
\end{align*}%
A formula \(\fml\) of \dLmu is \emph{\systemicformula} if it is a Boolean combination of \angelicformulas formulas.
In terms of games, \systemicformula{} formulas correspond to games in which only the player controlling the loop may make choices in the loop body.
Let \systemicdlmus be the fragment of \dLmu consisting of all \systemicformula \dLmu formulas.
The fragments \systemicdlmus and \dL are equiexpressive and equivalent by the results of \Cref{sec:equivalence}.

\Differentialgamelogic is significantly more expressive than \differentialdynamiclogic \cite{DBLP:journals/tocl/Platzer15}.
This increase may be expected to come from the addition of nested alternating loop games.
However, it is shown that the expressive power originates from the interaction of the alternating loops with the \emph{unrestricted} choice of continuous evolutions along the ODE and the \emph{unrestricted} quantification, which may have \emph{uncountably} many outcomes.
If the semantics of ODEs is modified minimally, this expressivity gap vanishes.
By restricting the players to stop the evolution only at \emph{rational} times, the expressive power of \dGL reduces to that of \differentialdynamiclogic.
This restriction is harmless for applications, since most properties of interest are robust enough, so that they are not only observed at irrational times.

\subsection{Rational \DifferentialGameLogic}

Recall that the natural numbers can be defined in \dL and hence also in \dGL \cite{DBLP:conf/lics/Platzer12b}.
Hence, the rational numbers $\rationals$ are also definable in \dL.
To define the fragment of \dGL restricted to rational evolution, define the \emph{rationally restricted} continuous evolution \(\stdoderat\) and non-deterministic assignment \(\gndassignrat{\ivar}\) as syntactic abbreviations as follows:
\begin{align*}
    \stdoderat             ~\synequiv~ \stdodet;\gtestp{\timevar\in\rationals}
                          \qquad\text{and} \qquad
    \gndassignrat{\ivarb}  ~\synequiv~ \gndassign{\ivarb};\gtestp{\ivarb\in\rationals}.
\end{align*}
The tests force the players to choose rational values for \(\timevar\) and \(\ivarb\), respectively.
This ensures that every atomic game is only countably branching.

\emph{\Rationaldifferentialgamelogic} (\dGLrat) is the fragment of \differentialgamelogic such that \emph{within the scope of a repetition game \(\garepeat{}\)} all evolutions and nondeterministic assignments are rational.
Similarly, \dLmurat is the fragment of \dLmu such that \emph{within the scope of a fixpoint operator} all evolutions and nondeterministic assignments are rational.

Nondeterministic assignment (and, thus, quantification and deterministic assignment) can be written equivalently in terms of a differential equation:
\[\provdl\lpossible{\gndassign{\ivar}}\fml\lbisubjunct\lpossible{\gode{\ivar}{1}\gachoice\gode{\ivar}{-1}}\fml.\]
The same holds for rationally restricted non-deterministic assignment, which can be written in terms of rational evolution games.
For this reason, in the following it is assumed, without loss of generality, that continuous evolution is the \emph{only} \gamesymbol in \dLmu and \dLmurat.

\subsection{Sequence Representations}

The proof of the equiexpressiveness result for \dGLrat uses sequence representations in \systemicdlmus that are introduced next.
The definition of the natural numbers in \systemicdlmus \cite{DBLP:conf/lics/Platzer12b} extends to an encoding of (countable) sequences of real numbers into a single real number and there are formulas characterizing such sequences \cite[Corollary A.2]{DBLP:journals/tocl/Platzer17}.
In fact, every real number can be treated as the representation of a sequence of numbers.
An extension by definitions makes it possible to write \(\finseqlen[\ivar]=n\) and \(\finseqlen[\ivar]=\infty\) to say the sequence is of length \(n\in\naturals\) or infinite, respectively.
Similarly, \(\elof{\ivar}{\ivarnat}\) describes the \(\ivarnat\)-th element of this sequence for \(\ivarnat\in\naturals\).
Finally, countable sets can be represented by \(\varincode{\ivar}{\ivarb}\), which says that the value of \(\ivar\) appears in the sequence represented by \(\ivarb\).%
\footnote{Formally, let \(\realat{\ivar}{\ivarnat}{\ivarb}\) be the coding \systemicdlmus-formula \cite[Corollary A.2]{DBLP:journals/tocl/Platzer17}.
(Here standard parenthesis notation is used to simplify notation for free variables and substitution.)
This representation has the natural property that \(\provdlmu \lforallnat{\ivarnat}\existsunique{\ivarb} \realat{\ivar}{\ivarnat}{\ivarb}\).
The notation \(\elof{\ivar}{i}\) for the \(i\)-th element in a context \(\fml(\elof{\ivar}{i})\) is defined by \(\lexists{\ivarb}(\realat{\ivar}{i}{\ivarb}\land\fml(\ivarb))\) for some fresh variable \(\ivarb\).
By convention, sequences begin with \(\elof{\ivar}{1}\) and the length of the sequence is stored in the zeroth position so that the predicate \(\finseqlen[\finseq]=n\) is defined as \(\elof{\ivar}{0}=n\) and \(\finseqlen[\finseq]=\infty\) is defined as \(\elof{\ivar}{0}<0\).
Finally, \(\varincode{\ivar}{\ivarb}\) abbreviates the formula \(\lexistsnat{\ivarnat}((n<\finseqlen[\ivarb]\lor\finseqlen[\ivarb]=\infty)\land\ivar=\elof{\ivarb}{\ivarnat}\)).
}
In order to construct sequences, the following axiom is added to the proof calculus of \dLmurat:
\[
    \cinferenceRule[unique-separation|${R}{}$]{axiom of countable \selection}
    {
        (\lforallnat{\ivarnat}\existsunique{\ivar}\fml)\limply{\lexists{\ivarb}\lforallnat{\ivarnat} \lpossible{\gassign{\ivar}{\elof{\ivarb}{\ivarnat}}}\fml}\quad
    }{\text{$y$ not in $\fml$}}.
\]
Axiom \irref{unique-separation} ensures that any \emph{definable} sequence, such as the sequence of the unique \(\ivar_1,\ivar_2,\ldots\) such that \(\fmlreplacevarby[{\fmlreplacevarby[\fml]{\ivarnat}{i}}]{\ivar}{\ivar_i}\) has a representation \(\ivarb\).
For example, \irref{unique-separation} proves that the concatenation \(\ivar{}^\frown\ivarb\) of two finite sequences represented by \(\ivar,\ivarb\) has a representation.

For a finite sequence \(\ivarbseq\) of variables \(\ivarb_1,\ldots,\ivarb_k\) the formula  \(\varisseq{\ivar}{\ivarbseq}\) says that \(\ivar\) represents the sequence of values of the variables \(\ivarbseq\).
The action \(\seqasvar{\ivarbseq}{\ivar}\) assigns the values of the sequence represented by \(\ivar\) to the variables in~\(\ivarbseq\) and action \(\seqassigntovar{\ivar}{\ivarbseq}\) assigns the representation of the values in \(\ivarbseq\) to \(\ivar\).%
\footnote{Formally, \(\varisseq{\ivar}{\ivarbseq}\) is defined by \(\elof{\ivar}{1}{=}\ivarb_1\land\ldots\land\elof{\ivar}{k}{=}\ivarb_k\) and \(\seqasvar{\ivarbseq}{\ivar}\) represents \(\gassign{\ivarb_1}{\elof{\ivar}{1}} \ldots{\gassign{\ivarb_k}{\elof{\ivar}{k}}}\).
    The assignment \(\seqassigntovar{\ivar}{\ivarbseq}\) is defined by \(\gndassign{\ivar};\gtestp{\varisseq{\ivar}{\ivarbseq}}\) and behaves like an assignment, since there provably is always some \(\ivar\) with \(\varisseq{\ivar}{\ivarbseq}\) by \irref{unique-separation}.
}

\subsection{Rational Fixpoint and Game Representation}

Every \dLmurat formula is provably reducible to an equivalent \dL formula and consequently \rationaldifferentialgamelogic is complete relative to \dL.
This is shown through an equivalent reduction from \dLmurat to \systemicdlmus, which are equivalent to \dGLrat and \dL respectively.

\begin{theoremE}[Rational \dL Representability]\label{prop:expressiveness-rational}
    Any \dLmurat formula \(\fml\) without free \pvarname{s} is provably equivalent to a \systemicdlmus formula \(\fmlb\), that is
    \(\provdlmu \fml \lbisubjunct \fmlb\).
\end{theoremE}

The rest of this subsection is devoted to the proof of \Cref{prop:expressiveness-rational}.
The idea of the reduction is to turn the \emph{fixpoint constructs} of \dLmurat into \emph{first-order quantifiers} in a compositional and fully local translation.
To turn a fixpoint into a variable, the necessary information about the fixpoint is encoded in a single \ivarname.
The crucial insight is that the entire (potentially uncountable) fixpoint does not need to be represented.
Because \dLmurat is only countably branching, it suffices to represent the \emph{reachable} states of the fixpoint.
A coding to represent countable sequences of numbers in \dL \cite{DBLP:journals/tocl/Platzer17} is used. 
With an additional \selection axiom these representations of infinite sequences can be reasoned about deductively.

For the following, fix an arbitrary finite list \(\ivarsall=\ivarsallel{1},\ldots,\ivarsallel{m}\) of \ivarname{s} and a degree bound \(N\in \naturals\).
Let \dLmuratminus be the fragment of \dLmurat consisting of formulas which mention only \ivarname{s} from \(\ivarsall\) and have only syntactic vector fields of polynomial degree at most~\(N\).
Since the choice of \(\ivarsall\) and \(N\) is arbitrary, it suffices to prove \Cref{prop:expressiveness-rational} for \dLmuratminus formulas.

A \emph{program pair} \((\stdodewo,\timecalcvar)\) consists of a continuous program and a rational number \(\timecalcvar\in\rationals\).
Because \(\stdodewo\) is essentially a finite sequence of rational numbers (the coefficients of \(\odef\)), any program pair can be assigned a single rational number, its \emph{program code} \(\progcodequote{\stdodewo,\timecalcvar}\).\footnote{List all possible \(\ivarsall\)-monomials of degree up to \(N\) as \(p_2,\ldots, p_k\).
    The program code \(\ivarc=\progcodequote{\stdodewo,\timecalcvar}\) is the encoding of the length \(k\) sequence of rational values where \(\elof{\ivarc}{1}{=}\timecalcvar\) and \(\elof{\ivarc}{i}\) is the coefficient of \(p_i\) in \(\odef\).}
Moreover, there is a formula \(\Phi_\fmlc\) that \emph{evaluates} a program code \(\progcodequote{\stdodewo,\timecalcvar}\) with respect to a postcondition~\(\fmlc\).%
\footnote{%
    Simply use \(\Phi_\fmlc \synequiv \lnecessary{\gode{\ivarsall(\timevar)}{\sum_{i=2}^k\elof{\ivarc}{i}p_i}}(t=\elof{\ivarc}{0}\limply \fmlc)\).
}
Formally, the formula \(\Phi_\fmlc\) says that \(\fmlc\) holds in the state reached along the program represented by the value of the fresh variable \(\ivarcode\), i.e.,
\[\provdlmu {\lpossible{\gassign{\ivarcode}{\progcodequote{\stdodewo,\timecalcvar}}}\Phi_\fmlc}\lbisubjunct\lnecessary{\stdodetwo}(\timevar=\timecalcvar\limply\fmlc).
\]
This encoding of program behaviors can be extended from atomic differential equation modalities to sequences of differential equation modalities.
A \emph{program sequence} \(\progseq\) is a finite sequence of program pairs.
The code \(\progcodequote{\progseq}\) of a program sequence~\(\progseq\) is the (rational) number encoding the sequence of codes of the elements of \(\progseq\).
Again evaluation of program sequence codes is characterizable syntactically.
For any formula \(\fml\), there is a formula \(\Psi_\fml\) which says that \(\fml\) holds in the state reached by the program sequence represented by the value of the variable \(\progcode\).
For a term \(term\), write \(\applyprogcodetoprog[\term]{\fml}\) for the formula \(\lpossible{\gassign{\progcode}{\term}}\Psi_\fml\).%
\footnote{The overloading of the notation \(\applyprogcodetoprog{\fml}\) is not ambiguous, as \(\progcode\) is a variable, not a program.
    Let \(\Psi_\fml
    \synequiv{\lforall{\ivarc}{}  ( R
        \limply
        \lpossible{\seqasvar{\ivarsall}{\ivarc}}\fml
        )}\),
    where \(R\synequiv
    \lexists{\finseq}({\finseqlen[\finseq]{=}\finseqlen[\progcode]+1}\land \varisseq{\elof{\finseq}{1}}{\ivarsall}\land \elof{\finseq}{\finseqlen[\finseq]}{=}\ivarc\land
    \lforallnat{i}({1\leq i\leq\finseqlen[\progcode]}\limply\lpossible{\seqasvar{\ivarsall}{\elof{\finseq}{i}};\gassign{\ivarcode}{\elof{\progcode}{i}}}\Phi_{\varisseq{\ivarsall}{\elof{\finseq}{i+1}}}))\)
    and \(\ivarc\) is fresh.}
The key property of the evaluation formula is that
\begin{equation*}\label{eq:reachability-encoding-property}
    \provdlmu \applyprogcodetoprog[\progcodequote{\emptyprogseq}]{\fml}\lbisubjunct\fml \qquad \provdlmu\applyprogcodetoprog[\progcodequote{(\stdodewo,\timecalcvar)}^\frown\progcode]{\fml}\lbisubjunct \lnecessary{\stdodetwo}(t=\timecalcvar\limply\applyprogcodetoprog[\progcode]{\fml}).\tag{T}
\end{equation*}
Thus, \(\applyprogcodetoprog[\progcodequote{\progseq}]{\fml}\) asserts that \(\fml\) holds in the state reached by following the ODEs in \(\progseq\) for their respective times sequentially.
While \(\applyprogcodetoprog[\progcode]{\fml}\) is always defined, it is only meaningful if the value of \(\progcode\) is the code of a \emph{finite} program sequence.
Importantly, the property of being a program sequence is definable in \systemicdlmus by a formula \(\finpspred{\ivar}\).%
\footnote{Define \(\finpspred{\ivar}\) by \(\finseqlen[\ivar]{<}\infty \land \lforall{i{\leq} \finseqlen[\ivar]} (\finseqlen[\ivar]=k\land \elof{\ivar}{i}{\in}\rationals)\), where \(k-1\) is the number of \(\ivarsall\)-monomials of degree at most \(N\).}
\odo{Better notation?} 

The representations of program sequences and the evaluation formulas make it possible to define a modality that describes reachability by \emph{any} \dLmuratminus program.
Let \(\forallprogcode{\fml}\) abbreviate the formula \(\lforallfinps{\progcode}\applyprogcodetoprog[\progcode]{\fml}\) and dually \(\forsomeprogcode{\fml}\synequiv\lnot\forallprogcode{\lnot\fml}\).
Semantically, the definable modality \(\forallprogcode{}\) is interpreted by rational reachability along any \dLmurat program.
Since \dLmurat reachability is an equivalence relation, this modality naturally and provably satisfies the axioms of modal logic S5 \cite{DBLP:books/cu/Chellas80}:

\begin{proposition}[Modality \(\forallprogcode{}\)]\label{prop:quantifier-properties-of-star}
    The \(\forallprogcode{}\) modality is provably S5 modal logic, i.e., \dLmu derives

    \begin{center}{\renewcommand{\linferenceRuleNameSeparation}{\hspace{1em}}
            \begin{calculuscollection}%
                \begin{calculus}%
                    \dinferenceRule[axiomT|Ax T]{axiom T}
                    {
                        \linferenceRule[impl]
                        {\forallprogcode{\fml}}
                        {\fml}
                    }{}
                    \dinferenceRule[axiomK|Ax K]{axiom K}
                    {
                        \linferenceRule[impl]
                        {\forallprogcodep{\fml\limply\fmlb}}
                        {(\forallprogcode{\fml}\limply \forallprogcode{\fmlb})}
                    }{}
                \end{calculus}
                \qquad
                \begin{calculus}
                    \dinferenceRule[axiom4|Ax 4]{axiom 4}
                    {
                        \linferenceRule[impl]
                        {\forallprogcode{\fml}}
                        {\forallprogcode{\forallprogcode{\fml}}}
                    }{}
                    \dinferenceRule[axiomB|Ax B]{axiom B}
                    {
                        \linferenceRule[impl]
                        {{\fml}}
                        {\forallprogcode{\forsomeprogcode{\fml}}}
                    }{}
                \end{calculus}
                \qquad
                \begin{calculus}
                    \dinferenceRule[generalization|G]{box generalization}
                    {
                        \linferenceRule[sequent]
                        {\fml}
                        {\forallprogcode{\fml}}
                    }{}
                \end{calculus}
            \end{calculuscollection}}
    \end{center}
\end{proposition}

\begin{proof}
    Most cases are straightforward and only axioms \irref{axiom4} and \irref{axiomB} are interesting.
    Axiom \irref{axiom4} reflects the transitivity of rational reachability and can be proved by combining the (representations of) sequences of programs and intermediate states using \irref{unique-separation}.
    Axiom \irref{axiomB} reflects the symmetry of the relation.
    To derive \irref{axiomB}, note that any sequence \(\progseq \equiv ((\gode{\odevarsx}{\odefofsth{\odef_1}},\timecalcvar_1), \ldots, (\gode{\odevarsx}{\odefofsth{\odef_m}},\timecalcvar_m))\) has an inverse sequence \(\progseq^- \equiv ((\gode{\odevarsx}{\odefofsth{{-}\odef_m}},\timecalcvar_m), \ldots, (\gode{\odevarsx}{\odefofsth{{-}\odef_1}},\timecalcvar_1))\).
    (Given the code of any program sequence \(\progseq\), a code for \(\progseq^-\) exists by \irref{unique-separation}. )
    This sequence provably reverses~\(\progseq\), since \(\provdlmu\fml\lbisubjunct \lnecessary{\gode{\odevarsx}{\odefofsth{\odef_i}}}(\timevar=\timecalcvar_i\land\lnecessary{\gode{\odevarsx}{\odefofsth{{-}\odef_i}}}(\timevar=\timecalcvar_i\land\fml))\) is derivable by \irref{nabla}.
    Then \irref{axiomB} can be proved, since it is possible to backtrack from any state reached from an initial state satisfying \(\fml\).
\end{proof}

\noindent
Axioms \irref{axiomMontone}, \irref{axiomAnd} and \irref{axiom5} are standard derivable axioms in S5 and are therefore derivable in~\dLmu.

\begin{center}{\renewcommand{\linferenceRuleNameSeparation}{\hspace{1em}}
        \begin{calculuscollection}%
            \begin{calculus}%
                \dinferenceRule[axiomMontone|$\square$M]{Monotonicity}
                {
                    \linferenceRule[sequent]
                    {\fml\limply\fmlb}
                    {\forallprogcode{\fml}\limply \forallprogcode{\fmlb}}
                }{}
            \end{calculus}
            \qquad\qquad
            \begin{calculus}
                \dinferenceRule[axiomAnd|$\square\land$]{axiom square and}
                {
                    \linferenceRule[equiv]
                    {\forallprogcodep{\fml\land\fmlb}}
                    {\forallprogcode{\fml}\land\forallprogcode{\fmlb}}
                }{}
            \end{calculus}
            \qquad\qquad
            \begin{calculus}
                \dinferenceRule[axiom5|Ax 5]{axiom 5}
                {
                    \linferenceRule[impl]
                    {\forsomeprogcode{\fmlb}}
                    { \forallprogcode{\forsomeprogcode{\fmlb}}}
                }{}
            \end{calculus}
        \end{calculuscollection}}
\end{center}

\noindent
By the definition of the \(\forallprogcode{}\) modality in terms of first-order quantifiers, it immediately inherits the following useful technical properties from the universal quantifier:

\begin{lemma}[Properties of \(\forallprogcode{\fmlb}\)]\label{prop:properties-of-star}
    For any \dLmuratminus formula \(\fml\), the modality \(\forallprogcode{}\) satisfies
    \begin{enumerate}
        \item \(\provdlmu\lforall{\ivarsall}\fml\limply\forallprogcode{\lforall{\ivarsall}\fml}\).\label{it:star-absorption}
        \item \(\provdlmu\forallprogcode{\lforall{\ivarb}\fml}\lbisubjunct\lforall{\ivarb}\forallprogcode{\fml}\) if \(\ivarb\notin\ivarsall\).\label{it:star-commutes}
        \item \(\provdlmu(\fmlb\limply\forallprogcode{\fml})\limply\forallprogcodep{\fmlb\limply\fml}\) where \(\fmlb\) does not mention variables from \(\ivarsall\).\label{it:prenex}
    \end{enumerate}
\end{lemma}

The purpose of the modality \(\forallprogcode{}\) is to say that a formula is true in any reachable state.
Thus \(\forallprogcode{}\fmlb\) can be used to ensure that \(\fmlb\) remains true in all possibly reachable states and \(\fmlb\) can be assumed independently of the context.

\begin{proposition}[Preservation]\label{prop:star-preservation-property}
    For any \dLmuratminus formulas  \(\fml,\fmlb\) the modality \(\forallprogcode{}\) satisfies
    \[\provdlmu(\forallprogcode{\fmlb}\land \fml)\limply\fmlreplacepvarby[\fml]{\pvar}{\forallprogcode{\fmlb}\land\pvar}.\]
\end{proposition}

\begin{proof}
    This is proved by induction on the \dLmuratminus formula \(\fml\).
    Most cases are straightforward and the only interesting ones are for differential equation modalities and fixpoints.

    \begin{caselist}
        \case{\(\lnecessary{\stdoderatwo}\fml\)}
        Note that \(\provdlmu  (\timecalcvar\in\rationals\land\finpspred{\progcode}\land\forallprogcode{\fmlb})\limply\lnecessary{\stdodewo}(\timevar=\timecalcvar\limply\applyprogcodetoprog[\progcode]{\fmlb})\) is obtained from \eqref{eq:reachability-encoding-property} by instantiating \(\forallprogcode{\fmlb}\).
        Consequently, \(\provdlmu \forallprogcode{\fmlb} \limply \lnecessary{\stdoderatwo}\forallprogcode{\fmlb}\) follows by universally quantifying over \(\progcode\) and \(\timecalcvar\).
        With \irref{nabla}, two differential equation modalities can be combined to show
        \(\provdlmu (\forallprogcode{\fmlb}\land \lnecessary{\stdoderatwo}\fmlc) \limply \lnecessary{\stdoderatwo}(\forallprogcode{\fmlb}\land\fmlc)\).
        The claim follows from the induction hypothesis with \irref{mon}.

        \case{\(\flfp{\pvarb}{\fml}\)}
        Let \(\fmlc\equiv\flfp{\pvarb}{\fmlreplacepvarby[\fml]{\pvar}{\forallprogcode{\fmlb}\land\pvar}}\) and \(\fmld \equiv\forallprogcode{\fmlb}\limply\fmlc\).
        Then \(\provdlmu \fmlreplacepvarby[{\fmlreplacepvarby[\fml]{\pvar}{\forallprogcode{\fmlb}\land\pvar}}]{\pvarb}{\fmlc}\limply\fmlc\) is an instance of \irref{fp}.
        By the induction hypothesis it follows that \(\provdlmu(\forallprogcode{\fmlb}\land\fmlreplacepvarby[\fml]{\pvarb}{\fmld})\limply\fmlc\) and thus \(\provdlmu\fmlreplacepvarby[\fml]{\pvarb}{\fmld}\limply\fmld\).
        Rule \irref{muind} shows that \(\provdlmu\flfp{\pvarb}{\fml}\limply\fmld\) and consequently \(\provdlmu(\forallprogcode{\fmlb}\land\flfp{\pvarb}{\fml})\limply \fmlc\).

        \case{\(\fgfp{\pvarb}{\fml}\)}
        Let \(\fmlc\equiv\forallprogcode{\fmlb}\land\fgfp{\pvarb}{\fml}\).
        By the induction hypothesis
        \[\provdlmu(\forallprogcode{\fmlb}\land \fmlreplacepvarby[{\fml}]{\pvarb}{\fgfp{\pvarb}{\fml}})\limply \fmlreplacepvarby[\fmlreplacepvarby[\fml]{\pvar}{\forallprogcode{\fmlb}\land\pvar}]{\pvarb}{\forallprogcode{\fmlb}\land\fgfp{\pvarb}{\fml}}.\]
        From this \(\provdlmu\fmlc\limply \fmlreplacepvarby[\fmlreplacepvarby[\fml]{\pvar}{\forallprogcode{\fmlb}\land\pvar}]{\pvarb}{{\fmlc}}\) derives by \irref{fp}.
        Finally, \(\provdlmu\fmlc\limply \fgfp{\pvarb}{\fmlreplacepvarby[\fml]{\pvar}{\forallprogcode{\fmlb}\land\pvar}}\) follows with \irref{muind}, as required.
    \end{caselist}
\end{proof}

For any \ivarname write \(\predofreal{\ivar}\) for the formula \(\lpossible{\seqassigntovar{\ivarb}{\ivarsall}}\varincode{\ivarb}{\ivar}\), which holds in those states that are in the set of states represented by \(\ivar\).
With this notation an \ivarname \(\ivar\) can be viewed as a representation of the countable set \(\predofreal{\ivar} \cong \{\elof{\ivar}{1},\elof{\ivar}{2}, \ldots, \elof{\ivar}{\finseqlen[\ivar]}\}\).
The key step in the reduction is to prove that there is a variable that encodes the denotation of any \folmus formula and this encoding is valid in every reachable state. 

\begin{proposition}[Representation]\label{prop:encodability}
    The formula \(
    \provdlmu\lexists{{{\ivar}}}\forallprogcodep{{\predofreal{\ivar}}\lbisubjunct \fml}\) is provable for every \(\dLmuratminus\) formula \(\fml\).
\end{proposition}

\begin{proof}
    Let \(\defassumption\synequiv \lforall{\ivarb}(\varincode{\ivarb}{\ivar} \lbisubjunct \fmlb)\), where \(\fmlb\synequiv \forsomeprogcodep{\varisseq{\ivarb}{\ivarsall}\land\fml}\) and \(\ivarb\) is a fresh variable.
    First, note that \(\lforallfinps{\progcode}\existsatmost[1]{\ivarb}\applyprogcodetoprogdia[\progcode]{(\varisseq{\ivarb}{\ivarsall}\land \fml)}\) can be proved using \irref{nabla} since \(\progcode\) is fully deterministic (\irref{approximation-uniqueness}).\footnote{This uses the fact that there are provably only countably many program sequences. In other words, there is a formula \(\finseqcodefml\) such that \(\provdlmu\lforallfinps{\progcode}\existsunique{\ivarnat{\in}\naturals} \finseqcodefml\).
        Such a formula can, for example, be constructed via a prime number encoding.}
    Since moreover \(\finps\) consists of rationals, \(\lexists{\ivar}\defassumption\) is provable using \irref{unique-separation}.
    By \irref{axiomMontone}, it suffice to show the two implications \(\provdlmu \defassumption\limply\forallprogcode{\defassumption}\) and \(\provdlmu\defassumption\limply(\predofreal{{\ivar}}\lbisubjunct \fml)\).

    Begin with showing the implication \(\provdlmu \defassumption\limply\forallprogcode{\defassumption}\).
    By \irref{axiom5} note that \(\provdlmu\fmlb\limply \forallprogcode{\fmlb}\) and by the definition of \(\defassumption\) derive \(\provdlmu\defassumption\limply\lforall{\ivarb}(\varincode{\ivarb}{\ivar}\limply \forallprogcode{\fmlb})\).
    It follows that \(\provdlmu\defassumption\limply\lforall{\ivarb}\forallprogcode{(\varincode{\ivarb}{\ivar}\limply \fmlb)}\) by \Cref{prop:properties-of-star}\eqref{it:prenex}.
    Now since \(\provdlmu\forsomeprogcode{\fmlb}\limply\fmlb\) by \irref{axiom4}, it follows from the definition of \(\defassumption\) that \(\provdlmu\defassumption\limply\lforall{\ivarb}(\varincode{\ivarb}{\ivar}\limplied \forsomeprogcode{\fmlb})\).
    Again by \Cref{prop:properties-of-star}\eqref{it:prenex} then \(\provdlmu\defassumption\limply\lforall{\ivarb}\forallprogcode{(\varincode{\ivarb}{\ivar}\limplied \fmlb)}\).
    Combining these with \irref{axiomAnd} yields \(\provdlmu\defassumption\limply\lforall{\ivarb}\forallprogcode{(\varincode{\ivarb}{\ivar}\lbisubjunct \fmlb)}\).
    By \Cref{prop:properties-of-star}\eqref{it:star-commutes} then \(\provdlmu \defassumption\limply\forallprogcode{\defassumption}\).

    For the implication \(\provdlmu\defassumption\limply(\predofreal{{\ivar}}\lbisubjunct \fml)\), note that \(\provdlmu \fmlb\limply  \forsomeprogcode{\lpossible{\seqasvar{\ivarsall}{\ivarb}}\fml}\) by \irref{axiomMontone} (using \(\provdlmu(\varisseq{\ivarb}{\ivarsall}\land\fml)\limply \lpossible{\seqasvar{\ivarsall}{\ivarb}}\fml\)) and thus \(\provdlmu \fmlb\limply  \lpossible{\seqasvar{\ivarsall}{\ivarb}}\fml\) by \Cref{prop:quantifier-properties-of-star}\eqref{it:star-absorption}.
    Then \(\provdlmu\varisseq{\ivarb}{\ivarsall}\limply(\fmlb\limply\fml)\).
    Conversely, \(\provdlmu\varisseq{\ivarb}{\ivarsall}\limply(\fmlb\limplied \fml)\) since \(\provdlmu (\varisseq{\ivarb}{\ivarsall}\land\fml)\limply \fmlb\) is an instance of reflexivity \irref{axiomT}.
    Thus, \(\provdlmu\varisseq{\ivarb}{\ivarsall}\limply(\fmlb\lbisubjunct \fml)\) and
    by definition of \(\defassumption\) this yields \(\provdlmu (\varisseq{\ivarb}{\ivarsall}\land\defassumption)\limply(\varincode{\ivarb}{\ivar}\lbisubjunct\fml)\).
    Hence, also \(\provdlmu (\varisseq{\ivarb}{\ivarsall}\land\defassumption)\limply(\predofreal{{\ivar}}\lbisubjunct\fml)\).
    Since \(\ivarb\) is not free in \(\defassumption, \predofreal{\ivar}\) and \(\fml\) it follows that \(\provdlmu\defassumption\limply(\predofreal{{\ivar}}\lbisubjunct \fml)\).
\end{proof}

With the help of the formulas \(\applyprogcodetoprog[\progcode]{\fml}\) and \(\forallprogcode{\fmlb}\), the translation from \dLmurat to \systemicdlmus can be defined and \Cref{prop:expressiveness-rational} proved.

\begin{proofE}[Proof of \Cref{prop:expressiveness-rational}]
    By axiom \irref{evolutiondomain} any \dLmuratminus formula is provably equivalent to one without an evolution domain constraint and it suffices to prove \Cref{prop:expressiveness-rational} for formulas without evolution domain constraint.
    Moreover, assume without loss of generality that the set \(\pvars\) of \pvarname{s} of interest is finite and no fixpoint variable is bound more than once.
    For every \pvarname \(\pvar\in\pvars\) let \(\ivarforpvar{\pvar}\) be a new \ivarname not in \(\ivarsall\).
    These \ivarname{s}~\(\ivarforpvar{\pvar}\) will encode the \emph{countable} set of states that satisfy the fixpoint formula associated to \(\pvar\).

    By induction on the negation normal form of a \dLmurat formula \(\fml\), define the equivalent \systemicdlmus formula \(\reccode{\fml}\) as follows
    \begin{aligntable}[3]
        \renewcommand{\eqsym}{\synequiv}
        \begin{align*}
            \nextit{\reccode{\atfml}}{\atfml}
            \nextit{\reccode{\fml\lor\fmlb}}{\reccode{\fml}\lor\reccode{\fmlb}}
            \nextit{\reccode{\lpossible{\stdoderatwo}\fml}}{\lpossible{\stdoderatwo}\reccode{\fml}}
            \nextit{\reccode{\pvar}}{\pvar\qquad\qquad}
            \nextit{\reccode{\fml\land\fmlb}}{\reccode{\fml}\land\reccode{\fmlb}}
            \nextit{\reccode{\lnecessary{\stdoderatwo}\fml}}{\lnecessary{\stdoderatwo}\reccode{\fml}}
            \nextit{\reccode{\flfp{\pvar}{\fml}}}{\lforall{\ivarforpvar{\pvar}}{(\forallprogcodep{\fmlreplacepvarby[\reccode{\fml}]{\pvar}{\predofivar{{\pvar}}}\limply\predofivar{{\pvar}}}\limply\predofivar{{\pvar}})}\span\span}
            \lastit{\reccode{\fgfp{\pvar}{\fml}}}{\lexists{\ivarforpvar{\pvar}}{(\forallprogcodep{\predofivar{{\pvar}}\limply\fmlreplacepvarby[\reccode{\fml}]{\pvar}{\predofivar{{\pvar}}}}\land\predofivar{{\pvar}})}}
        \end{align*}%
    \end{aligntable}%
    where \(\atfml\) is an atomic formula or its negation.

    Now by induction on \(\dLmurat\) formula \(\fml\), it is shown that \(\provdlmu\reccode{\fml}\lbisubjunct\fml\).
    All cases except the fixpoint cases are straightforward.
    Begin with the implication \(\provdlmu\reccode{\flfp{\pvar}{\fml}}\limply\flfp{\pvar}{\fml}\).
    For convenience write \(\delta \synequiv\predofivar{{\pvar}}\lbisubjunct \flfp{\pvar}{\fml}\) so that
    \begin{equation}\label{eq:delta-rewriting}
        \provdlmu\forallprogcode{\delta}\limply (\predofivar{\ivar}\lbisubjunct\flfp{\pvar}{\fml}).
    \end{equation}
    By \Cref{prop:encodability}, note \(\provdlmu\lexists{\ivarforpvar{{\pvar}}}\forallprogcode{\delta}\).
    Note that \(\provdlmu(\forallprogcode{\delta}\land\fmlreplacepvarby[\fml]{\pvar}{\predofivar{\ivar}})\limply\fmlreplacepvarby[\fml]{\pvar}{\predofivar{\ivar}\land\forallprogcode{\delta}}\) by \Cref{prop:star-preservation-property}.
    By \eqref{eq:delta-rewriting} and monotonicity derive
    \(\provdlmu(\forallprogcode{\delta}\land\fmlreplacepvarby[\fml]{\pvar}{\predofivar{\ivar}})\limply\fmlreplacepvarby[\fml]{\pvar}{\flfp{\pvar}{\fml}}\).
    Now by \irref{fp} it follows that \(\provdlmu(\forallprogcode{\delta}\land\fmlreplacepvarby[\fml]{\pvar}{\predofivar{\ivar}})\limply\flfp{\pvar}{\fml}\) so that \(\provdlmu(\forallprogcode{\delta}\land\fmlreplacepvarby[\fml]{\pvar}{\predofivar{\ivar}})\limply\predofivar{\ivar}\) by \eqref{eq:delta-rewriting}.
    By the induction hypothesis conclude \(\provdlmu\forallprogcode{\delta}\limply\fmlb\) where \(\fmlb \synequiv \fmlreplacepvarby[\reccode{\fml}]{\pvar}{\predofivar{{\pvar}}}\limply\predofivar{{\pvar}}\).
    Using \irref{generalization} and \irref{axiomK} derive \(\provdlmu\forallprogcode{\forallprogcode{\delta}}\limply\forallprogcode{\fmlb}\).
    Then \(\provdlmu\forallprogcode{\delta}\limply\forallprogcode{\fmlb}\) by \irref{axiom4}.
    Letting \(\eta\synequiv\forallprogcode{\fmlb}\limply\predofivar{{\pvar}}\) this yields \(\provdlmu(\forallprogcode{\delta}\land\eta)\limply\predofivar{{\pvar}}\).
    Then \(\provdlmu(\forallprogcode{\delta}\land\eta)\limply\flfp{\pvar}{\fml}\) again by \eqref{eq:delta-rewriting}.
    Finally, because \(\provdlmu\lexists{\ivarforpvar{{\pvar}}}\forallprogcode{\delta}\), the formula \(\lforall{\ivarforpvar{\pvar}}\eta\limply\flfp{\pvar}{\fml}\) is derivable.
    Equivalently \(\provdlmu\reccode{\flfp{\pvar}{\fml}}\limply\flfp{\pvar}{\fml}\).

    For the reverse implication \(\provdlmu\flfp{\pvar}{\fml}\limply\reccode{\flfp{\pvar}{\fml}}\) note that
    \(\provdlmu (\forallprogcode{\fmlb}\land \fmlreplacepvarby[\reccode{\fml}]{\pvar}{\predofivar{{\pvar}}})\limply\predofivar{{\pvar}}\)
    is derivable with \irref{axiomT}.
    By the induction hypothesis this is equivalent to
    \(\provdlmu(\forallprogcode{\fmlb}\land \fmlreplacepvarby[{\fml}]{\pvar}{\predofivar{{\pvar}}})\limply\predofivar{{\pvar}}\).
    Since \(\provdlmu(\eta\land\forallprogcode{\fmlb})\limply\ivarforpvar{\pvar}\) monotonicity derives
    \(\provdlmu(\forallprogcode{\fmlb}\land\fmlreplacepvarby[{\fml}]{\pvar}{\eta\land\forallprogcode{\fmlb}})\limply\predofivar{{\pvar}}\).
    Then, \(\provdlmu(\forallprogcode{\fmlb}\land \fmlreplacepvarby[{\fml}]{\pvar}{\eta})\limply\predofivar{{\pvar}}\),
    follows by \Cref{prop:properties-of-star}.
    Equivalently, \(\provdlmu\fmlreplacepvarby[{\fml}]{\pvar}{\eta}\limply\eta\) so that \(\provdlmu\fmlreplacepvarby[{\fml}]{\pvar}{\lforall{\ivarforpvar{\pvar}}\eta}\limply\lforall{\ivarforpvar{\pvar}}\eta\).
    Using \irref{muind} conclude \(\provdlmu\flfp{\pvar}{\fml}\limply\lforall{\ivarforpvar{\pvar}}\eta\) or equivalently \(\provdlmu\flfp{\pvar}{\fml}\limply\reccode{\flfp{\pvar}{\fml}}\).

    The case for the greatest fixpoint is entirely dual to the case for the least fixpoint.
\end{proofE}

The proof of the equivalence of the reduction relies on the fixpoint axiomatization of differential equations \irref{nabla} to handle ODEs completely.
This is subtle, as axiom \irref{nabla} introduces \emph{unrestricted} nondeterministic choice even for rational-time ODEs and thus leaves the rational fragment \dLmurat.
However as axiom \irref{nabla} is only used locally, the inductive argument is possible.

The same proof can reduce other fragments of \dGL to \systemicdlmus.
It is only important that the restrictions defining the fragment can be enforced syntactically, that the resulting restricted reachability relation \(\forallprogcode{}\) can be characterized syntactically and provably satisfies the preservation property (\Cref{prop:star-preservation-property}).
Then any game in such a fragment can be reduced to a single player game by the argument of \Cref{prop:expressiveness-rational}.
Instead of restricting the game to force the players to make rational-valued choices, other restrictions on the strategies such as computability, definability or continuity of the strategies also ensure equiexpressiveness.

\subsection{\DifferentialGameLogic and \DifferentialDynamicLogic}

Restricting to rationally played games, the adversarial dynamics of \dGL add no expressive or deductive power to \differentialdynamiclogic.

\begin{theoremE}[Rational Equiexpressiveness][normal]
    \Differentialdynamiclogic \dL and \rationaldifferentialgamelogic \dGLrat are equiexpressive.
\end{theoremE}

\begin{proofE}
    Equiexpressiveness for formulas follows from \Cref{prop:expressiveness-rational}, since \systemicdlmus and \dL, as well as \dLmurat and \dGLrat are equiexpressive.
\end{proofE}

An important consequence is the alternation hierarchy collapse for \dGLrat.
This property of \dGLrat is of practical interest, as the main algorithmic challenges of the propositional \(\mu\)-calculus stem from the strict alternation of fixpoints \cite{DBLP:conf/lics/EmersonL86}.
\begin{corollary}[\dGLrat Alternation Hierarchy]\label{prop:adglratcollapse}
    Any \dGLrat formula is equivalent to a formula without nested repetitions.
\end{corollary}

Note that \dGLrat is not a conservative extension over \dL, as the \dGL calculus can prove more true \dL formulas than \dL.
Consequently \dGLrat is not equivalent to \dL. 
However, similar to \Cref{prop:equivalence-transfers-completeness} \rationaldifferentialgamelogic is \emph{complete relative} to \differentialdynamiclogic.
That is, any formula that is valid in \differentialdynamiclogic is provable in \rationaldifferentialgamelogic from \dL tautologies.
Write \(\provdglplus\fml\) if there is a proof for \(\fml\) in the \dGL proof calculus from \dL tautologies.

\begin{theoremE}[\dGLrat Relative Completeness]\label{thm:relativecomplete}
    \Rationaldifferentialgamelogic is complete relative to \differentialdynamiclogic: For any \dGLrat formula~\(\fml\) if \(\valdgl\fml\) then \(\provdglplus\fml.\)
\end{theoremE}

By \cite{DBLP:conf/lics/Platzer12b} it follows that \dGLrat is complete relative to the purely discrete and the purely continuous fragment of \differentialdynamiclogic.
This further reduces the required axioms to prove all valid \dGLrat formulas.
In addition, validity of \dGLrat is decidable relative to validity in either fragment.
This is in contrast to full \dGL, which cannot be expressively reduced to \dL to prove relative completeness.
Unlike previous relative completeness results for \dGL this relative completeness result reduces the two-player adversarial complexity of \dGLrat completely to the single player case.

Note that the proof of \Cref{thm:relativecomplete} does not require axiom \irref{unique-separation}, since it is a \dL tautology.
In fact, \Cref{prop:expressiveness-rational} is not needed in its full generality and it suffices to show that for any \dLmurat formula \(\fml\) there is a \systemicdlmus formula \(\fmlb\) such that \(\dL\provdlmu \fml \lbisubjunct \fmlb\).
For this purpose the proof of \Cref{prop:expressiveness-rational} could be simplified. For example, the syntactic coding statements are \dL tautologies and have a trivial proof in the \(\provdglplus\) calculus.

%########################################################

\section{Related Work}\label{sec:related-work}

First-order fixpoint logics \cite{DBLP:journals/bsl/DawarG02} have been studied in finite model theory and descriptive complexity.
First-order least fixpoint logic was shown to have the same expressiveness as first-order logic with inflationary fixpoints \cite{DBLP:journals/apal/Kreutzer04}.
Versions of the first-order \(\mu\)-calculus have been introduced with less general modalities \cite{DBLP:journals/igpl/AfshariEL24,DBLP:journals/iandc/CalvaneseGMP18}.
Inductive definitions play an important role in recursion theory and descriptive set theory \cite{Moschovakis74}.
First-order modal logics have been investigated in philosophy \cite{Fitting1998-MELFML}.

The relationship between the propositional \(\mu\)-calculus and game logic has been investigated \cite{DBLP:conf/focs/Parikh83} and they were shown to have different expressiveness \cite{DBLP:journals/mst/BerwangerGL07}.
This gap was closed via sabotage games \cite{DBLP:conf/lics/WafaP24}, which make the two equiexpressive and proof-theoretically equivalent, although with a non-elementary increase in description complexity.
While completeness of the propositional modal \(\mu\)-calculus has been proven \cite{DBLP:conf/lics/Walukiewicz95}, completeness for propositional game logic \cite{DBLP:conf/lics/EnqvistHKMV19} is open \cite{kloibhofer2023note}.

Applications of games and the modal \(\mu\)-calculus for hybrid systems have been considered in the literature \cite{DBLP:journals/corr/abs-0911-4833,DBLP:journals/tocl/Platzer15,DBLP:conf/hybrid/Davoren97}.
Relative completeness for \dL was shown, with an axiomatization of ODEs using Euler-approximation steps \cite{DBLP:conf/lics/Platzer12b} and relative completeness of \dGL to differentially expressive logics has been shown \cite{DBLP:journals/tocl/Platzer15}.
Completeness of \dL for differential equation invariance was established using syntactic Lie derivatives \cite{DBLP:conf/lics/PlatzerT18}.
This was applied to establish completeness of \dL invariance reasoning for open properties of compact initial value problems \cite{DBLP:journals/jacm/PlatzerQ25}.

Coinductive definitions on metric spaces have been studied to relate quantitative approximation to greatest fixpoint iterations as in \Cref{sec:quantitative-approximation} via metric coinduction principles \cite{DBLP:journals/corr/abs-0908-2793} and via quantitative fixpoint theories \cite{DBLP:conf/lics/MardarePP21} building on abstract iteration theories \cite{DBLP:conf/lics/SimpsonP00}.

Notions of robust properties for hybrid systems have been considered for hybrid automata with noise \cite{DBLP:conf/csl/Franzle99} and via \(\delta\)-decision procedures \cite{DBLP:conf/tacas/KongGCC15}.
A quasi-decision procedure for arbitrary properties of hybrid systems that may not terminate for nonrobust systems has been developed \cite{DBLP:journals/fmsd/Ratschan14}.
All these notions of robustness capture the same idea of safety or reachability being independent of boundary behaviour.
They differ in the exact definition depending on their application context.
However, in these approaches robustness serves as a quantitative algorithmic device to guarantee convergence and usually requires a constant positive margin of safety (\(\varepsilon>0\)) at the boundary.
In contrast, \Cref{def:robustness} is a purely topological condition that integrates seamlessly into the logical framework and does not require fixing safety margins.

%########################################################

\section{Conclusion}
This article unifies the descriptions of program properties via fixpoints and games by presenting and proving equiexpressive and equivalent general first-order game logic and first-order model $\mu$-calculus.
By instantiating this result, \differentialgamelogic and \differentialmucalculus are proved equiexpressive and equivalent.
Through a fixpoint axiomatization of differential equations, continuous behavior is reduced to a discrete fixpoint.
This enables strong completeness results for robust properties and a relative completeness result for rational differential game logic to simple program properties.
The equiexpressiveness result for rational \differentialgamelogic shows that rational adversarial behavior can be handled completely within \dL.

The underlying ideas of semantics- and provability-preserving roundtrip translations that enable syntactic proof transfer significantly simplify subtle arguments and are of more general interest.
The use of fixpoints in a form of nondeterministic dynamic programming for ODE reachability is a new perspective on continuous reachability.
The coinductive description promises to be useful to integrate differential equation into other frameworks of coinduction beyond \folmus and \fogls.

While \dGL is a good language for specification and theorem proving for hybrid systems, the fixpoint version \dLmu is more amenable to the development of (bounded) model checking techniques.
An implementation of the translation could profitably combine the practical utility of a theorem prover like KeYmaera X \cite{DBLP:conf/cade/FultonMQVP15} and model checking tools for the modal \(\mu\)-calculus.

A practical application of the fixpoint description of reachability is invariant generation, which represents a major bottleneck in hybrid systems verification.
Computing reachability via symbolic fixpoint iterations enables fast and provable approximation of the weakest precondition.
Moreover, convergence to a (differential) invariant is guaranteed under explicit conditions (\Cref{rem:invariant-generation}).
This approximation can be used by invariant generation tools such as Pegasus \cite{DBLP:journals/fmsd/SogokonMTCP21}, which generate invariant candidates for verification.
Since invariants are foundational to formal verification more broadly, the coinductive characterization offers many potential applications.

\begin{acks}
    We thank the anonymous reviewers for their valuable comments, which helped improve the paper.
\end{acks}

\balance
\bibliography{Diffix.bib}

\iflongversion
    \newpage
    \nobalance

    \appendix

    %########################################################

    \section{Free Variables, Bound Variables and Substitution}\label{sec:free-and-bound-variables}

    \subsection{Free Variables in \FirstorderGameLogic}

    For a game of \firstordergamelogic define the set of \emph{necessarily bound variables}
    \begin{aligntable}[2]
        \begin{align*}
            \nextit{\mustboundvars{\gamesymbat{\ivarseq}{\termseq}}}{\ivarseq}
            \nextit{\mustboundvars{\gdual{\game}}}{\mustboundvars{\game}}
            \nextit{\mustboundvars{\gtest{\fml}}}{\emptyset}
            \nextit{\mustboundvars{\game\gachoice\gameb}}{\mustboundvars{\game}\cap\mustboundvars{\gameb}}
            \nextit{\mustboundvars{\garepeat{\game}}}{\emptyset}
            \lastit{\mustboundvars{\game\gcom\gameb}}{\mustboundvars{\game}\cup\mustboundvars{\gameb}}%
        \end{align*}%
    \end{aligntable}%
    The set of free variables of a \fogls formula and of a \fogls game are defined by simultaneous induction.
    For formulas:
    \begin{aligntable}[2]
        \begin{align*}
            \nextit{\freevars{\fml\land\fmlb}}{\freevars{\fml}\cup\freevars{\fmlb}}
            \nextit{\freevars{\lnot\fml}}{\freevars{\fml}\qquad}
            \lastit{\freevars{\lpossible{\game}\fml}}{\freevars{\game}\cup(\freevars{\fml}\setminus\mustboundvars{\game})\span\span}
        \end{align*}%
    \end{aligntable}%
    and for games:
    \begin{aligntable}[2]
        \renewcommand{\displayitem}[2]{&#1=#2}
        \begin{align*}
            \nextit{\freevars{\gamesymbat{\ivarseq}{\termseq}}}{\freevars{\termseq}}
            \nextit{\freevars{\gdual{\game}}}{\freevars{\game}}
            \nextit{\freevars{\game\gachoice\gameb}}{\freevars{\game}\cup\freevars{\gameb}}
            \nextit{\freevars{\gtest{\fml}}}{\freevars{\fml}}
            \nextit{\freevars{\game\gcom\gameb}}{\freevars{\game}\cup(\freevars{\gameb}\setminus\mustboundvars{\game})}
            \lastit{\freevars{\garepeat{\game}}}{\freevars{\game}}
        \end{align*}%
    \end{aligntable}%
    Again the crucial property of the free variables is the coincidence lemma.

    \begin{lemmaE}[Coincidence for \fogls][normal]\label{prop:coincidencegl}
        For every \emph{\fogls formula} \(\fml\) and \(\fostate,\fostateb\in\fostates\) such that \(\strestrvar{\fostate}{\freevars{\fml}}=\strestrvar{\fostateb}{\freevars{\fml}}\) then
        \[\foglsem{\fonstr}{\fostate}{\fml} \quad\iff\quad \foglsem{\fonstr}{\fostateb}{\fml}\]
    \end{lemmaE}

    \begin{proofE}
        Recall the definition of \(Z\)-closed sets of states from the proof of \Cref{prop:coincidencelmu}.
        Note that if \(\fostatessubset\) is \(Z\)-closed then it is \(Z'\)-closed for all \(Z'\supseteq Z\).
        By simultaneous induction on the definition of \fogls formulas and games prove the following
        \begin{enumerate}
            \item \(\glfden{\fml}\) is \(\freevars{\fml}\)-closed
            \item \(\glgden{\game}(\fostatessubset)\) is \(Z\)-closed if \(\fostatessubset\) is \(Z\cup\mustboundvars{\game}\)-closed and \(\freevars{\game}\subseteq Z\)
        \end{enumerate}

        \begin{caselist}
            \case{\(\atfml\)}
            Standard for atomic formulas.

            \case{\(\lnot\fml\)}
            As in the proof of \Cref{prop:coincidencelmu} the complement of a \(Z\)-closed set is \(Z\)-closed.

            \case{\(\fml\land\fmlb\)}
            As in the proof of \Cref{prop:coincidencelmu} the intersection of \(Z\)-closed sets is \(Z\)-closed.

            \case{\(\lpossible{\game}\fml\)}
            By the induction hypothesis on \(\fml\) the set \(\glfden{\fml}\) is \(\freevars{\fml}\)-closed and since \(\freevars{\lpossible{\game}\fml}\cup\mustboundvars{\game}\supseteq \freevars{\fml}\) the set \(\glfden{\fml}\) is also \((\freevars{\lpossible{\game}\fml}\cup\mustboundvars{\game})\)-closed.
            So by the induction hypothesis on \(\game\) for \(Z=\freevars{\lpossible{\game}\fml}\) the set \(\glgden{\game}(\glfden{\fml})\)
            ii \(\freevars{\lpossible{\game}\fml}\)-closed.

            \case{\(\gamesymbat{\ivarseq}{\termseq}\)}
            Suppose \( \strestrvar{\fostate}{Z}=\strestrvar{\fostateb}{Z}\) and \(\fostateb\in\glgden{\gamesymbat{\ivarseq}{\termseq}}(\fostatessubset)\).
            By definition of the semantics there is \(\pstatepre\in\fonstrint{\gamesymb}(\ltden[\fostateb]{\termseq})\) such that \(\replpstate[\fostateb]{\pstatepreatvar{\ivarseq}}\subseteq\fostatessubset\).
            Because \(\freevars{\termseq}=\freevars{\game}\subseteq Z\) also \(\ltden[\fostateb]{\termseq}=\ltden[\fostate]{\termseq}\).
            Let \(Z'=Z \cup \ivarseq\) and note that \( \strestrvar{(\replpstate[\fostateb]{\pstatepreatvar{\ivarseq}})}{Z'}=\strestrvar{(\replpstate[\fostate]{\pstatepreatvar{\ivarseq}})}{Z'}\).
            Since \(\fostatessubset\) is \((Z\cup\mustboundvars{\game})\)-closed, it follows that \(\replpstate[\fostate]{\pstatepreatvar{\ivarseq}}\subseteq\fostatessubset\).
            Hence, also \(\fostate\in\glfden{\gamesymbat{\ivarseq}{\termseq}}(\fostatessubset)\).

            \case{\(\gdual{\game}\)}
            Immediate by the induction hypothesis as the complement of a \(Z\)-closed set is \(Z\)-closed. (See proof of \Cref{prop:coincidencelmu}.)

            \case{\(\gtest{\fml}\)}
            The intersection \(\glgden{\gtest{\fml}}(\fostatessubset)=\glfden{\fml}\cap\fostatessubset\) is \(Z\)-closed as the intersection of two \(Z\)-closed sets, since \(\glfden{\fml}\) is \(Z\)-closed by induction hypothesis.

            \case{\(\game\gachoice\gameb\)}
            By induction hypothesis (noting that \(\fostatessubset\) is \(Z\cup(\mustboundvars{\game})\)-closed as \(Z\cup(\mustboundvars{\game\cup\gameb})\subseteq Z\cup \mustboundvars{\game}\)), it follows that \(\glgden{\game}(\fostatessubset)\) is \(Z\)-closed.
            As the union of two \(Z\)-closed sets, \(\glgden{\game\gachoice\gameb}(\fostatessubset)\) is \(Z\)-closed.

            \case{\(\game;\gameb\)}
            Observe that because \(\freevars{\game;\gameb}\subseteq Z\) also \(\freevars{\gameb}\subseteq Z\cup\mustboundvars{\game}\).
            Since \(\fostatessubset\) is \((Z\cup\mustboundvars{\game}\cup\mustboundvars{\gameb})\)-closed, the induction hypothesis on \(\gameb\) yields that \(\glgden{\gameb}(\fostatessubset)\) is \(Z\cup\mustboundvars{\game}\)-closed.
            So by the induction hypothesis on \(\game\) it follows that \(\glgden{\game}(\fostatessubset)\) is \(Z\)-closed.

            \case{\(\garepeat{\game}\)}
            Note that \(\fostatessubset\) is \(Z\)-closed by assumption.
            By the induction hypothesis and as the intersection of \(Z\)-closed sets is \(Z\)-closed, it follows that \(\glgden{\garepeat{\game}}(\fostatessubsetb)\cup\fostatessubset\) is \(Z\)-closed \emph{for all \(\fostatessubsetb\subseteq\fostates\)}.
            Since the set of \(Z\)-closed sets is also closed under arbitrary unions, it follows by induction on the fixpoint iterates of \(\fostatessubsetb\mapsto \glgden{\game}(\fostatessubsetb)\cup\fostatessubset\) that \(\glgden{\garepeat{\game}}(\fostatessubset)\) is \(Z\)-closed.
        \end{caselist}
    \end{proofE}

    Finally also define the bound variables of a game:
    \begin{aligntable}[2]
        \renewcommand{\displayitem}[2]{&#1=#2}
        \begin{align*}
            \nextit{\boundvars{\gamesymbat{\ivarseq}{\termseq}}}{\ivarseq}
            \nextit{\boundvars{\gdual{\game}}}{\boundvars{\game}}
            \nextit{\boundvars{\gtest{\fml}}}{\emptyset}
            \nextit{\boundvars{\game\gachoice\gameb}}{\boundvars{\game}\cup\boundvars{\gameb}}
            \nextit{\boundvars{\garepeat{\game}}}{{\boundvars{\game}}}
            \lastit{\boundvars{\game\gcom\gameb}}{\boundvars{\game}\cup\boundvars{\gameb}}%
        \end{align*}%
    \end{aligntable}%
    The crucial property of bound variables is that \emph{only} bound variables can capture free variables.
    This is made precise by:

    \begin{lemmaE}[Bound Effect][normal]
        For every \fogls game \(\game\)
        \[\glgden{\game}(\fostatessubset)=
            \{\fostate:\fostate\in \glgden{\game}(\fostatessubsetboundeff{\game}{\fostate})\}\]
        where \(\fostatesame{\game}{\fostate}=\{\fostateb\in\fostates : \mforall{\ivarb\notin\boundvars{\game}} \strestrvar{\fostate}{\ivarb}= \strestrvar{\fostateb}{\ivarb}\}\).
    \end{lemmaE}

    \begin{proofE}
        Note \( \{\fostate:\fostate\in \glgden{\game}(\fostatessubsetboundeff{\game}{\fostate})\}\subseteq \glgden{\game}(\fostatessubset)\) by monotonicity of games.
        This shows the \(\subseteq\) inclusion of the lemma.
        Note that \(\fostatesame{\fostate}{\game}\subseteq\fostatesame{\fostate}{\gameb}\) if \(\boundvars{\game}\subseteq\boundvars{\gameb}\).
        For the reverse inclusion show by induction on games~\(\game\) that
        \begin{enumerate}
            \item \(\glgden{\game}(\fostatessubset)=
                  \{\fostate:\fostate\in \glgden{\game}(\fostatessubsetboundeff{\game}{\fostate})\}\)
            \item \(\glgden{\gdual{\game}}(\fostatessubset)=
                  \{\fostate:\fostate\in \glgden{\gdual{\game}}(\fostatessubsetboundeff{\game}{\fostate})\}\)
        \end{enumerate}

        \begin{caselist}
            \case{\(\gamesymbat{\ivarbseq}{\termseq}\)}
            Suppose \(\fostate\in\glgden{\gamesymbat{\ivarbseq}{\termseq}}(\fostatessubset)\).
            By definition of the semantics there is \(U\in\fonstrint{\gamesymb}(\ltden[\fostate]{\termseq})\) such that \(\replpstate[\fostate]{\pstatepreatvar{\ivarbseq}}\in\fostatessubset\) for all \(\pstatepre\in U\).
            Clearly \(\replpstate[\fostate]{\pstatepreatvar{\ivarbseq}}\in\fostatessubsetboundeff{\game}{\fostate}\), since \(\ivarbseq\subseteq\boundvars{\game}\).
            Hence, also \(\fostate\in\glgden{\gamesymbat{\ivarbseq}{\termseq}}(\fostatessubsetboundeff{\game}{\fostate})\).

            For the dual case suppose \(\fostate\notin\glgden{\gamesymbat{\ivarbseq}{\termseq}}(\fostatessubsetcomp)\).
            By definition of the semantics, for all \(U\in\fonstrint{\gamesymb}(\ltden[\fostate]{\termseq})\) there is \(\pstatepre\in U\) with \(\replpstate[\fostate]{\pstatepreatvar{\ivarbseq}}\in\fostatessubset\) .
            Clearly \(\replpstate[\fostate]{\pstatepreatvar{\ivarbseq}}\in\fostatessubsetboundeff{\game}{\fostate}\), since \(\ivarbseq\subseteq\boundvars{\game}\).
            Hence, \(\fostate\in\glgden{\gdualp{\gamesymbat{\ivarbseq}{\termseq}}}(\fostatessubsetboundeff{\game}{\fostate})\).

            \case{\(\gdual{\game}\)}
            Immediate by the inductive hypothesis on \(\game\).

            \case{\(\gtest{\fml}\)}
            Both the case for tests and their dual are straightforward consequences of the definition.

            \case{\(\game\gachoice\gameb\)}
            Suppose \(\fostate\in\glgden{\game\gachoice\gameb}(\fostatessubset)\).
            By induction hypothesis \(\fostate\in\glgden{\game}(\fostatessubsetboundeff{\game}{\fostate})\cup \glgden{\gameb}(\fostatessubsetboundeff{\gameb}{\fostate})\).
            By monotonicity of games, it follows that \(\fostate\in\glgden{\game\gachoice\gameb}(\fostatessubsetboundeff{\game\gachoice\gameb}{\fostate})\).
            The dual case is similar.

            \case{\(\game;\gameb\)}
            Suppose \(\fostate\in\glgden{\game;\gameb}(\fostatessubset)\), then by induction hypothesis \(\fostate\in\glgden{\game}\fostatessubsetboundeff[(\glgden{\gameb}(\fostatessubset)]{\game}{\fostate})\).
            Now say \(\fostateb\in\fostatessubsetboundeff[\glgden{\gameb}(\fostatessubset)]{\game}{\fostate}\), then \(\fostateb\in\fostatesame{\game;\gameb}{\fostate}\) and \(\fostateb\in \glgden{\gameb}(\fostatessubsetboundeff{\gameb}{\fostateb})\).
            By monotonicity it follows that \(\fostateb\in \glgden{\gameb}(\fostatessubsetboundeff{\game;\gameb}{\fostateb})\) and hence  \(\fostateb\in \glgden{\gameb}(\fostatessubsetboundeff{\game;\gameb}{\fostate})\).
            So \(\fostate\in\glgden{\game}(\glgden{\gameb}(\fostatessubsetboundeff[\fostatessubset]{\game}{\fostate}))\).
            The dual case is similar.

            \case{\(\garepeat{\game}\)}
            By induction on ordinals \(\gamma\) prove that the fixpoint iterates satisfy \(\mlfppi{\gamma}{\fostatessubsetb}{\fostatessubset\cup\glgden{\game}(\fostatessubsetb)}\subseteq\glgden{\garepeat{\game}}(\fostatessubset)\).
            The case that \(\gamma=0\) and the case that \(\gamma\) is a limit ordinal are immediate.
            For successor ordinals by the induction hypothesis on \(\gamma\)
            \begin{align*}
                \mlfppi{\gamma+1}{\fostatessubsetb}{\fostatessubset\cup\glgden{\game}(\fostatessubsetb)}
                 & =
                \fostatessubset\cup\glgden{\game}(\mlfppi{\gamma}{\fostatessubsetb}{\fostatessubset\cup\glgden{\game}(\fostatessubsetb)})
                \subseteq
                \fostatessubset\cup\glgden{\game}(\glgden{\garepeat{\game}}(\fostatessubset))
                =\glgden{\garepeat{\game}}(\fostatessubset)
            \end{align*}
            The dual case is similar.
        \end{caselist}
    \end{proofE}

    Also define the bound variables of a formula of \fogls, which are all the variables that are bound \emph{anywhere} in \(\fml\):
    \begin{aligntable}[2]
        \renewcommand{\displayitem}[2]{&#1=#2}
        \begin{align*}
            \nextit{\boundvars{\atfml}}{\emptyset}
            \nextit{\boundvars{\fml\land\fmlb}}{\boundvars{\fml}\cup\boundvars{\fmlb}}
            \nextit{\boundvars{\lnot{\fml}}}{\boundvars{\fml}}
            \lastit{\boundvars{\lpossible{\game}\fml}}{\boundvars{\game}\cup\boundvars{\fml}}%
        \end{align*}%
    \end{aligntable}

    %########################################################

    \subsection{Substitution in \FirstorderGameLogic}

    Define substitution for \fogls by induction on the definition.
    Again the substitution of a variable by itself is defined as \(\fmlreplacevarby{\ivar}{\ivar}\synequiv\fml\).
    For \fogls formulas and games define when \(\ivar\not\equiv\term\) the substitution:
    \begin{aligntable}[2]
        \begin{align*}
            \nextit{\fmlreplacevarpby[\lnot\fml]{\ivar}{\term}}{\lnot\fmlreplacevarby[\fml]{\ivar}{\term}\quad}
            \nextit{\fmlreplacevarpby[\fml\land\fmlb]{\ivar}{\term}}{\fmlreplacevarby[\fml]{\ivar}{\term}\land\fmlreplacevarby[\fmlb]{\ivar}{\term}}
            \nextit{\fmlreplacevarpby[\lpossible{\game}{\fml}]{\ivar}{\term}}{
                \begin{cases}
                    \lpossible{\gamereplacevarby[{\game}]{\ivar}{\term}}\fml                                   & \text{if } \ivar\in\mustboundvars{\game}                                                    \\
                    \lpossible{\gamereplacevarby[{\game}]{\ivar}{\term}}\fmlreplacevarby[{\fml}]{\ivar}{\term} & \text{if } \freevars{\term}\cap\boundvars{\game}=\emptyset\mand\ivar\notin\boundvars{\game} \\
                    \lpossible{\gassign{\ivar}{\term};\game}\fml                                               & \text{otherwise}
                \end{cases}\span\span\span
            }
            \nexti
            \nextit{\gamereplacevarpby[{\gamesymbat{\ivarbseq}{\termbseq}}]{\ivar}{\term}}{\gamesymbat{\ivarbseq}{\fmlreplacevarby[\termb]{\ivar}{\term}}\quad}
            \nextit{\gamereplacevarpby[{\gtest{\fml}}]{\ivar}{\term}}{\fmlreplacevarby[\gtest{\fml}]{\ivar}{\term}}
            \nextit{\gamereplacevarpby[{\game\cup\gameb}]{\ivar}{\term}}{\gamereplacevarby[{\game}]{\ivar}{\term}\cup\gamereplacevarby[{\gameb}]{\ivar}{\term}}
            \nextit{\gamereplacevarpby[{\gdual{\game}}]{\ivar}{\term}}{\gdualp{\gamereplacevarby[{{\game}}]{\ivar}{\term}}}
            \nextit{\gamereplacevarpby[{\game;\gameb}]{\ivar}{\term}}{
                \begin{cases}
                    \gamereplacevarby[{\game}]{\ivar}{\term};\gameb                                    & \text{if } \ivar\in\mustboundvars{\game}                                                    \\
                    \gamereplacevarby[{\game}]{\ivar}{\term};\gamereplacevarby[{\gameb}]{\ivar}{\term} & \text{if } \freevars{\term}\cap\boundvars{\game}=\emptyset\mand\ivar\notin\boundvars{\game} \\
                    \gassign{\ivar}{\term};\game;\gameb                                                & \text{otherwise}
                \end{cases}\span\span
            }
            \nexti
            \lastit{\gamereplacevarpby[{\garepeat{\game}}]{\ivar}{\term}}{
                \begin{cases}
                    \garepeatp{\gamereplacevarby[{\game}]{\ivar}{\term}} & \text{if } \freevars{\term}\cap\boundvars{\game}=\emptyset\mand\ivar\notin\boundvars{\game} \\
                    \gassign{\ivar}{\term};\garepeat{\game}              & \text{otherwise}
                \end{cases}\span\span
            }
        \end{align*}%
    \end{aligntable}%
    \noindent
    For any set \(\fostatessubset\subseteq\fostates\) let \(\fostatessubsetafterdas{\ivar}{\term} = \{\fostate : \streplaceby{\fostate}{\ivar}{\ltden[\fostate]{\term}}\in\fostatessubset\}\).

    \begin{lemmaE}[\fogls Substitution][normal]\label{prop:substgl}
        For every \emph{\fogls formula} \(\fml\), game \(\game\) and term \(\term\):
        \begin{align*}
            \fostate\in \glfden{\fmlreplacevarby[\fml]{\ivar}{\term}} \quad & \iff\quad \streplaceby{\fostate}{\ivar}{\ltden{\term}}\in \glfden{\fml}
        \end{align*}
    \end{lemmaE}

    \begin{proofE}
        By simultaneous induction on formulas and games prove the following:
        \begin{enumerate}
            \item\label{subglfml} \(\glfden{\fmlreplacevarby[\fml]{\ivar}{\term}} = \fostatessubsetafterdasp[\glfden{\fml}]{\ivar}{\term}\)
            \item\label{subglmustbound} if \(\ivar\in\mustboundvars{\game}\):
                  \(\glgden{\gamereplacevarby[\game]{\ivar}{\term}}(\fostatessubset)= \fostatessubsetafterdasp[\glgden{\game}(\fostatessubset)]{\ivar}{\term}\)
            \item\label{subglnoclash} if \(\freevars{\term}\cap\boundvars{\game}=\emptyset\) and \(\ivar\notin\boundvars{\game}\):
                  \[\glgden{\gamereplacevarby[\game]{\ivar}{\term}}(\fostatessubsetafterdas{\ivar}{\term})=
                      \fostatessubsetafterdasp[\glgden{\game}(\fostatessubset)]{\ivar}{\term}\]
        \end{enumerate}
        \begin{caselist}
            \case{\(\atfml\)}
            This is standard substitution for atomic formulas.

            \case{\(\lnot\fml\)}
            Immediate from the induction hypothesis.

            \case{\(\fml\land\fmlb\)}
            Immediate from the induction hypothesis.

            \case{\(\lpossible{\game}\fml\)}
            There are three cases.
            If \(\ivar\in\mustboundvars{\game}\) then by~\pref{subglmustbound} of the induction hypothesis on \(\game\)
            \begin{align*}
                \glfden{\fmlreplacevarby[(\lpossible{\game}\fml)]{\ivar}{\term}}
                =
                \glfden{\gamereplacevarby[\game]{\ivar}{\term}}(\glfden{\fml})
                =
                \fostatessubsetafterdasp[\glfden{\game}(\glfden{\fml})]{\ivar}{\term}.
            \end{align*}
            Suppose then \(\freevars{\term}\cap\boundvars{\game}=\emptyset\) and \(\ivar\notin\boundvars{\game}\).
            Then by~\pref{subglnoclash} of the induction hypothesis on \(\game\) and the induction hypothesis on \(\fml\):
            \begin{align*}
                \glfden{\fmlreplacevarby[(\lpossible{\game}\fml)]{\ivar}{\term}}
                 & =
                \glfden{\gamereplacevarby[\game]{\ivar}{\term}}(\glfden{\fmlreplacevarby[\fml ]{\ivar}{\term}})
                =
                \glfden{\gamereplacevarby[\game]{\ivar}{\term}}(\fostatessubsetafterdasp[\glfden{\fml}]{\ivar}{\term})
                =
                \fostatessubsetafterdasp[\glfden{\game}(\glfden{\fml})]{\ivar}{\term}.
            \end{align*}

            Finally, the case that neither \(\ivar\in\mustboundvars{\game}\) nor (\(\freevars{\term}\cap\boundvars{\game}=\emptyset\) and \(\ivar\notin\boundvars{\game}\))
            is immediate by the definition of deterministic assignment in \Cref{sec:deterministic-assignment}.

            \case{\(\gamesymbat{\ivarbseq}{\termseq}\)}
            Suppose first \(\ivar\in\mustboundvars{\gamesymbat{\ivarbseq}{\termseq}}=\ivarbseq\).
            Then
            \begin{align*}%
                \glgden{\gamesymbat{\ivarbseq}{\fmlreplacevarby[\termbseq]{\ivar}{\term}}}(\fostatessubset)
                 & =
                \{\fostate:\mexists{\pstatepre\in\fonstrint{\gamesymb}(\ltden[\fostate]{\fmlreplacevarby[\termbseq]{\ivar}{\term}})}
                \replpstate[\fostate]{\pstatepreatvar{\ivarbseq}}\subseteq\fostatessubset\}
                \\
                 & =
                \{\fostate:\mexists{\pstatepre\in\fonstrint{\gamesymb}(\ltden[{\streplaceby{\fostate}{\ivar}{\ltden{\term}}}]{\termbseq})}
                \replpstate[\streplaceby{\fostate}{\ivar}{\ltden{\term}}]{\pstatepreatvar{\ivarbseq}}\subseteq\fostatessubset\}
                \\
                 & =
                \fostatessubsetafterdasp[\glgden{\gamesymbat{\ivarbseq}{\termbseq}}(\fostatessubset)]{\ivar}{\term}.
            \end{align*}
            Next suppose \(\ivar\notin\ivarbseq\) and \(\ivarbseq\cap\freevars{\term}=\emptyset\).
            Then
            \begin{align*}%
                \glgden{\gamesymbat{\ivarbseq}{\fmlreplacevarby[\termbseq]{\ivar}{\term}}}(\fostatessubsetafterdas{\ivar}{\term})
                 & =
                \{\fostate:\mexists{\pstatepre\in\fonstrint{\gamesymb}(\ltden[\fostate]{\fmlreplacevarby[\termbseq]{\ivar}{\term}})}
                \replpstate[\fostate]{\pstatepreatvar{\ivarbseq}}\subseteq\fostatessubsetafterdas{\ivar}{\term}\}
                \\
                 & =
                \{\fostate:\mexists{\pstatepre\in\fonstrint{\gamesymb}(\ltden[\fostate]{\fmlreplacevarby[\termbseq]{\ivar}{\term}})}
                \streplaceby{(\replpstate[\fostate]{\pstatepreatvar{\ivarbseq}})}{\ivar}{\ltden[{\replpstate[\fostate]{\pstatepreatvar{\ivarbseq}}}]{\term}}\subseteq\fostatessubset\}
                \\
                 & =
                \{\fostate:\mexists{\pstatepre\in\fonstrint{\gamesymb}(\ltden[{\streplaceby{\fostate}{\ivar}{\ltden{\term}}}]{\termbseq})}
                \replpstate[\streplaceby{\fostate}{\ivar}{\ltden{\term}}]{\pstatepreatvar{\ivarbseq}}\subseteq\fostatessubset\}
                \\
                 & =
                \fostatessubsetafterdasp[\glgden{\gamesymbat{\ivarbseq}{\termbseq}}(\fostatessubset)]{\ivar}{\term}.
            \end{align*}

            \case{\(\gdual{\game}\)}
            This is immediate since \(\fostates\setminus\fostatessubsetafterdasp{\ivar}{\term}=\fostatessubsetafterdas[(\fostates\setminus\fostatessubset)]{\ivar}{\term}\).

            \case{\(\gtest{\fml}\)}
            Only \pref{subglnoclash} is to be shown.
            By induction hypothesis on \(\fml\)
            \begin{align*}
                \glgden{\gamereplacevarby[(\gtest{\fml})]{\ivar}{\term}}(\fostatessubsetafterdas{\ivar}{\term})
                 & =
                \glfden{\fmlreplacevarby[\fml]{\ivar}{\term}}\cap\fostatessubsetafterdas{\ivar}{\term}
                =
                \fostatessubsetafterdasp[\glfden{\fml}\cap\fostatessubset]{\ivar}{\term}
                =
                \fostatessubsetafterdasp[\glgden{\gtest{\fml}}(\fostatessubset)]{\ivar}{\term}.
            \end{align*}

            \case{\(\game\gachoice\gameb\)}
            Again suppose first \(\ivar\in\mustboundvars{\game\gachoice\gameb}\).
            Then \(\ivar\in\mustboundvars{\game}\) and \(\glgden{\gamereplacevarby[\game]{\ivar}{\term}}(\fostatessubset)= \fostatessubsetafterdasp[\glgden{\game}(\fostatessubset)]{\ivar}{\term}\) by the induction hypothesis \pref{subglmustbound} on \(\game\).
            Similarly for \(\gameb\), so that \(\glgden{\gamereplacevarpby[\game\gachoice\gameb]{\ivar}{\term}}(\fostatessubset)= \fostatessubsetafterdasp[\glgden{\game\gachoice\gameb}(\fostatessubset)]{\ivar}{\term}\).

            Next suppose \(\boundvars{\game\gachoice\gameb}\cap\freevars{\term}=\emptyset\) and \(\ivar\notin\boundvars{\game\gachoice\gameb}\).
            Then \(\boundvars{\game}\cap\freevars{\term}=\emptyset\) and \(\ivar\notin\boundvars{\game}\), so the induction hypothesis applies to \(\game\) and it applies to \(\gameb\) for the same reason, yielding the desired equality.

            \case{\(\game;\gameb\)}
            There are three cases.
            Suppose first \(\ivar\in\mustboundvars{\game}\) then \(\ivar\in\mustboundvars{\game;\gameb}\) so only~\pref{subglmustbound} needs to be shown.
            By the induction hypothesis on \(\game\):
            \begin{align*}
                \glgden{\gamereplacevarby[\game]{\ivar}{\term}\gcom\gameb}(\fostatessubset)
                =
                \glgden{\gamereplacevarby[\game]{\ivar}{\term}}(\glgden{\gameb}(\fostatessubset))
                =
                \fostatessubsetafterdasp[\glgden{\game}(\glgden{\gameb}(\fostatessubset))]{\ivar}{\term}.
            \end{align*}
            Suppose next \(\freevars{\term}\cap\boundvars{\game}=\emptyset\) and \(\ivar\notin\boundvars{\game}\).
            Then only~\pref{subglnoclash} needs to be shown.
            By the induction hypothesis on \(\game\) and \(\gameb\):
            \begin{align*}
                \glgden{\gamereplacevarby[\game]{\ivar}{\term};\gamereplacevarby[\gameb]{\ivar}{\term}}(\fostatessubsetafterdas{\ivar}{\term})
                =
                \glgden{\gamereplacevarby[\game]{\ivar}{\term}}(\glgden{\gamereplacevarby[\gameb]{\ivar}{\term}}(\fostatessubsetafterdas{\ivar}{\term}))
                =
                \glgden{\gamereplacevarby[\game]{\ivar}{\term}}(\fostatessubsetafterdasp[\glgden{\gameb}(\fostatessubset)]{\ivar}{\term})
                =
                \fostatessubsetafterdasp[\glgden{\game}(\glgden{\gameb}(\fostatessubset))]{\ivar}{\term}.
            \end{align*}

            Finally, the case that neither \(\ivar\in\mustboundvars{\game}\) nor (\(\freevars{\term}\cap\boundvars{\game}=\emptyset\) and \(\ivar\notin\boundvars{\game}\))
            is immediate by the definition of deterministic assignment in \Cref{sec:deterministic-assignment}.

            \case{\(\garepeat{\game}\)}
            Suppose first \(\freevars{\term}\cap\boundvars{\game}=\emptyset\) and \(\ivar\notin\boundvars{\game}\).
            Then only~\pref{subglnoclash} needs to be shown.
            By definition of the semantics of \(\garepeat{\game}\) to show is that
            \[\mlfpp{\fostatessubsetb}{\glgden{\gamereplacevarby[\game]{\ivar}{\term}}(\fostatessubsetb)\cup\fostatessubsetafterdas{\ivar}{\term}}
                =
                \fostatessubsetafterdasp[\mlfpp{\fostatessubsetb}{\glgden{\game}(\fostatessubsetb)
                        \cup\fostatessubset}]{\ivar}{\term}.\]
            By induction on ordinals \(\gamma\) prove this for the fixpoint approximations.
            For successor ordinals this holds by the induction hypothesis on \(\game\).
            \begin{align*}
                \mlfppi{\gamma+1}{\fostatessubsetb}{\glgden{\gamereplacevarby[\game]{\ivar}{\term}}(\fostatessubsetb)\cup\fostatessubsetafterdas{\ivar}{\term}}
                 & =                \glgden{\gamereplacevarby[\game]{\ivar}{\term}}(\mlfppi{\gamma}{\fostatessubsetb}{\glgden{\gamereplacevarby[\game]{\ivar}{\term}}(\fostatessubsetb)\cup\fostatessubsetafterdas{\ivar}{\term}})  \cup\fostatessubsetafterdas{\ivar}{\term}
                \\
                 & =                \glgden{\gamereplacevarby[\game]{\ivar}{\term}}(
                \fostatessubsetafterdasp[\mlfppi{\gamma}{\fostatessubsetb}{\glgden{\game}(\fostatessubsetb)
                        \cup\fostatessubset}]{\ivar}{\term}
                )  \cup\fostatessubsetafterdas{\ivar}{\term}
                \\
                 & =                \fostatessubsetafterdasp[\glgden{\game}(
                    \mlfppi{\gamma}{\fostatessubsetb}{\glgden{\game}(\fostatessubsetb)
                        \cup\fostatessubset})\cup \fostatessubset]{\ivar}{\term}
                \\
                 & =
                \fostatessubsetafterdasp[\mlfppi{\gamma+1}{\fostatessubsetb}{\glgden{\game}(\fostatessubsetb)
                        \cup\fostatessubset}]{\ivar}{\term}.
            \end{align*}
            The limit step is immediate by the induction hypothesis on the fixpoint approximations.

            Again, the case that \(\freevars{\term}\cap\boundvars{\game}\neq\emptyset\) and \(\ivar\in\boundvars{\game}\) is immediate by the definition of deterministic assignment in \Cref{sec:deterministic-assignment}.
        \end{caselist}
    \end{proofE}

    %########################################################

    \subsection{Free Variables in \FirstorderMuCalculus}
    The set of free variables \(\freevars{\atfml}\) of an atomic \(\gsig\)-formula \(\atfml\) are defined as usual to be all the variables appearing in \(\atfml\).
    The set of \emph{syntactically free} variables of a formula of the \firstordermucalculus is defined by induction on the formula as follows:
    \begin{aligntable}[2]
        \renewcommand{\displayitem}[2]{&#1=#2}
        \begin{align*}
            \nextit{\freevars{\pvar}}{\{\pvar\}\qquad}
            \nextit{\freevars{\fml\land\fmlb}}{\freevars{\fml}\cup\freevars{\fmlb}}
            \nextit{\freevars{\lnot\fml}}{\freevars{\fml}}
            \nextit{\freevars{\flfp{\pvar}{\fml}}}{\freevars{\fml}\setminus\{\pvar\}\span\span}
            \lastit{\freevars{\lpossible{\gamesymbat{\ivarseq}{\termseq}}\fml}}{\freevars{\termseq}\cup(\freevars{\fml}\cap\varcomp{\ivarseq})\span\span}.
        \end{align*}
    \end{aligntable}%
    Note that this contains \pvarname{s} and \ivarname{s}.

    \begin{lemmaE}[Coincidence for \folmus][normal]\label{prop:coincidencelmu}
        For every \emph{\folmus formula}~\(\fml\) \emph{without free \pvarname{s}} and \(\fostate,\fostateb\in\fostates\) such that \(\strestrvar{\fostate}{\freevars{\fml}}=\strestrvar{\fostateb}{\freevars{\fml}}\)
        \[\folmusem{\fonstr}{\fostate}{\fml} \quad\iff\quad \folmusem{\fonstr}{\fostateb}{\fml}.\]
    \end{lemmaE}

    \begin{proofE}
        Say a set \(\fostatessubset\subseteq\fostates\) is \(Z\subseteq\ivars\)-closed if
        \[\{\fostate : \mexists{\fostateb} \; \strestrvar{\fostate}{Z}=\strestrvar{\fostateb}{Z} \mand \fostateb\in\fostatessubset\}\subseteq \fostatessubset.\]
        By induction on \(\fml\) show that \(\lmden{\intp}{\fml}\) is \(Z\)-closed if \(\intp(\pvar)\) is \(Z\)-closed for all \(\pvar\in\pvars\) and \(\freevars{\fml}\cap\ivars\subseteq Z\).
        The interesting cases are for modalities and fixpoint formulas.

        \begin{caselist}
            \case{\(\atfml\)}
            This is standard coincidence for atomic formulas.

            \case{\(\pvar\)}
            By assumption \(\lmden{\intp}{\pvar}=\intp(\pvar)\) is \(Z\)-closed.

            \case{\(\lnot{\fml}\)}
            Holds since the complement of a \(Z\)-closed set is \(Z\)-close.
            Suppose \(\fostatessubset\) is \(Z\)-closed, but \(\fostates\setminus\fostatessubset\) is not.
            Then there is \(\fostate\in\fostatessubset\) and \(\fostateb\in\fostates\setminus\fostatessubset\) such that \(\strestrvar{\fostate}{Z}=\strestrvar{\fostateb}{Z}\).
            Because \(\fostatessubset\) is \(Z\)-closed then also \(\fostateb\in\fostatessubset\). A contradiction.

            \case{\(\fml\land\fmlb\)}
            Holds as the intersection of \(Z\)-closed sets \(\fostatessubset_1,\fostatessubset_2\) is closed.
            Suppose there is \(\fostateb\in\fostatessubset_1\cap\fostatessubset\) such that \(\strestrvar{\fostate}{Z}=\strestrvar{\fostateb}{Z}\).
            Then \(\fostate\in\fostatessubset_1\), because \(\fostatessubset_1\) is \(Z\)-closed.
            , for \(\fostatessubset_2\), so that \(\fostate\in\fostatessubset_1\cap\fostatessubset_2\).

            \case{\(\lpossible{\gamesymbat{\ivarseq}{\termseq}}\fml\)}
            Suppose \( \strestrvar{\fostate}{Z}=\strestrvar{\fostateb}{Z}\) and \(\fostateb\in\lmdenprop{\intp}{\lpossible{\gamesymbat{\ivarseq}{\termseq}}\fml}\).
            By definition of the semantics there is \(\pstatepre\in\fonstrint{\gamesymb}(\ltden[\fostateb]{\termseq})\) such that \(\replpstate[\fostateb]{\pstatepreatvar{\ivarseq}}\subseteq\lmden{\intp}{\fml}\).
            Because \(\freevars{\termseq}\subseteq Z\) also \(\ltden[\fostateb]{\termseq}=\ltden[\fostate]{\termseq}\).
            Let \(Z'=Z \cup \ivarseq\) and note that \( \strestrvar{(\replpstate[\fostateb]{\pstatepreatvar{\ivarseq}})}{Z'}=\strestrvar{(\replpstate[\fostate]{\pstatepreatvar{\ivarseq}})}{Z'}\). Since \(\freevars{\fml}\subseteq Z'\), it follows with the induction hypothesis that \(\replpstate[\fostate]{\pstatepreatvar{\ivarseq}}\subseteq\lmden{\intp}{\fml}\).
            Hence, also \(\fostate\in\lmdenprop{\intp}{\lpossible{\gamesymbat{\ivarseq}{\termseq}}\fml}\).

            \case{\(\flfp{\pvar}{\fml}\)}
            By induction on ordinals \(\gamma\) the fixpoint iterates \(\mlfpi{\gamma}{\fostatessubset}{\lmden{\intreplaceby{\pvar}{\fostatessubset}}{\fml}}\) are shown to be \(Z\)-closed.
            For the successor case \(\mlfpi{\gamma+1}{\fostatessubset}{\lmden{\intp}{\fml}}=\lmden{\intreplaceby{\pvar}{\mlfpi{\alpha}{\fostatessubset}{\lmden{\intp}{\fml}}}}{\fml}\) is \(Z\)-closed by induction hypothesis.
            For the limit case, arbitrary unions of \(Z\)-closed sets are \(Z\)-closed.
        \end{caselist}
    \end{proofE}

    %########################################################

    \subsection{Substitution in \FirstorderMuCalculus}

    Defining substitutions for the \firstordermucalculus is a bit more complex.
    For an atomic formula \(\atfml\) the formula obtained by replacing the variable \(\ivar\) for the term \(\term\) everywhere is denoted \(\fmlreplacevarby[\atfml]{\ivar}{\term}\) and similarly for terms.
    Importantly the substitution of a variable by itself is defined as \(\fmlreplacevarby{\ivar}{\ivar}\synequiv\fml\).
    This extends to \folmus formulas when \(\ivar\notsynequiv\term\) as follows:
    \begin{aligntable}[2]
        \begin{align*}
            \nextit{\fmlreplacevarpby[\pvar]{\ivar}{\term}}{\lpossible{\gassign{\ivar}{\term}}\pvar}
            \nextit{\fmlreplacevarpby[\lnot\fml]{\ivar}{\term}}{\lnot\fmlreplacevarby[\fml]{\ivar}{\term}}
            \nextit{\fmlreplacevarpby[\fml\land\fmlb]{\ivar}{\term}}{\fmlreplacevarby[\fml]{\ivar}{\term}\land\fmlreplacevarby[\fmlb]{\ivar}{\term}}
            \nextit{\fmlreplacevarpby[\lpossible{\gamesymbat{\ivarbseq}{\termbseq}}\fml]{\ivar}{\term}}{\begin{cases}
                                                                                                                \lpossible{\gamesymbat{\ivarbseq}{\fmlreplacevarby[\termbseq]{\ivar}{\term}}}\fml                                 & \text{if }\ivar\in\ivarbseq                                                \\
                                                                                                                \lpossible{\gamesymbat{\ivarbseq}{\fmlreplacevarby[\termbseq]{\ivar}{\term}}}\fmlreplacevarby[\fml]{\ivar}{\term} & \text{if }\ivar\notin\ivarbseq\mand\freevars{\term}\cap\ivarbseq=\emptyset \\
                                                                                                                \lpossible{\gassign{\ivar}{\term}}\lpossible{\gamesymbat{\ivarbseq}{\termbseq}}\fml                               & \text{otherwise}                                                           \\
                                                                                                            \end{cases}}
            \lastit{\fmlreplacevarpby[{\flfp{\pvar}{\fml}}]{\ivar}{\term}}{\lpossible{\gassign{\ivar}{\term}}\flfp{\pvar}{\fml}}
        \end{align*}%
    \end{aligntable}%

    \begin{lemmaE}[\folmus Substitution][normal]\label{prop:substlmu}
        For \emph{\folmus formulas} \(\fml\) and terms~\(\term\):
        \[\fostate\in\lmden{\intp}{\fmlreplacevarby{\ivar}{\term}} \quad\iff\quad \streplaceby{\fostate}{\ivar}{\ltden{\term}}\in\lmden{\intp}{\fml}.\]
    \end{lemmaE}

    \begin{proofE}
        By induction on \(\fml\).

        \begin{caselist}
            \case{\(\atfml\)}
            This is standard substitution for atomic formulas.

            \case{\(\pvar\)}
            By the definition of deterministic assignment in \Cref{sec:deterministic-assignment}.

            \case{\(\lnot\fml\)}
            Immediate from the induction hypothesis.

            \case{\(\fml\land\fmlb\)}
            Immediate from the induction hypothesis.

            \case{\(\lpossible{\gamesymbat{\ivarbseq}{\termbseq}}\fml\)}
            Suppose first \(\ivar\in\ivarbseq\).
            Then the pairwise equivalences hold:
            {\renewcommand{\iff}{\quad\text{iff}\quad}% 
            \begin{align*}%
                \fostate\in\lmden{\intp}{\lpossible{\gamesymbat{\ivarbseq}{\fmlreplacevarby[\termbseq]{\ivar}{\term}}}\fml}
                \iff &
                \mexists{\pstatepre\in\fonstrint{\gamesymb}(\ltden[\fostate]{\fmlreplacevarby[\termbseq]{\ivar}{\term}})}
                \replpstate[\fostate]{\pstatepreatvar{\ivarbseq}}\subseteq\lmden{\intp}{\fml}
                \\
                \iff &
                \mexists{\pstatepre\in\fonstrint{\gamesymb}(\ltden[{\streplaceby{\fostate}{\ivar}{\ltden{\term}}}]{\termbseq})}
                \replpstate[\streplaceby{\fostate}{\ivar}{\ltden{\term}}]{\pstatepreatvar{\ivarbseq}}\subseteq\lmden{\intp}{\fml}
                \\
                \iff &
                \streplaceby{\fostate}{\ivar}{\ltden{\term}}\in\lmden{\intp}{\lpossible{\gamesymbat{\ivarbseq}{\termbseq}}\fml}.
            \end{align*}
            }%
            Next suppose \(\ivar\notin\ivarbseq\) and \(\freevars{\term}\cap\ivarbseq=\emptyset\).
            Then the pairwise equivalences hold:
            {\renewcommand{\iff}{\quad\text{iff}\quad}% 
            \begin{align*}%
                \fostate\in\lmden{\intp}{\lpossible{\gamesymbat{\ivarbseq}{\fmlreplacevarby[\termbseq]{\ivar}{\term}}}\fmlreplacevarby[\fml]{\ivar}{\term}}
                % \\
                \iff &
                \mexists{\pstatepre\in\fonstrint{\gamesymb}(\ltden[\fostate]{\fmlreplacevarby[\termbseq]{\ivar}{\term}})}
                \replpstate[\fostate]{\pstatepreatvar{\ivarbseq}}\subseteq\lmden{\intp}{\fmlreplacevarby[\fml]{\ivar}{\term}}
                \\
                \iff &
                \mexists{\pstatepre\in\fonstrint{\gamesymb}(\ltden[{\streplaceby{\fostate}{\ivar}{\ltden{\term}}}]{\termbseq})}
                \streplaceby{(\replpstate[\fostate]{\pstatepreatvar{\ivarbseq}})}{\ivar}{\ltden[{\replpstate[\fostate]{\pstatepreatvar{\ivarbseq}}}]{\term}}
                \subseteq\lmden{\intp}{\fml}
                \\
                \iff &
                \mexists{\pstatepre\in\fonstrint{\gamesymb}(\ltden[{\streplaceby{\fostate}{\ivar}{\ltden{\term}}}]{\termbseq})}
                \replpstate[\streplaceby{\fostate}{\ivar}{\ltden{\term}}]{\pstatepreatvar{\ivarbseq}}\subseteq\lmden{\intp}{\fml}
                \\
                \iff &
                \streplaceby{\fostate}{\ivar}{\ltden{\term}}\in\lmden{\intp}{\lpossible{\gamesymbat{\ivarbseq}{\termbseq}}\fml}.
            \end{align*}
            }%
            Finally the case that \(\ivar\notin\ivarbseq\) and \(\freevars{\term}\cap\ivarbseq\neq\emptyset\)
            is immediate by the definition of deterministic assignment in \Cref{sec:deterministic-assignment}.

            \case{\(\flfp{\pvar}{\fml}\)}
            By the definition of deterministic assignment in \Cref{sec:deterministic-assignment}.
        \end{caselist}
    \end{proofE}

    In the \firstordermucalculus substitutions of \pvarname{s} is also crucial.
    Note that a substitution \(\fmlreplacepvarby[\fml]{\pvar}{\fmlc}\) is only possible if \(\pvar\) does not appear in \(\fml\) in the scope of a fixpoint operator binding any \pvarname that is free in \(\fmlc\).
    This will always be assumed implicitly.

    \begin{aligntable}[2]
        \begin{align*}
            \nextit{\fmlreplacepvarby[\pvarb]{\pvar}{\fmlc}}{
                \begin{cases}
                    \fmlc  & \text{if }\pvar=\pvarb \\
                    \pvarb & \text{otherwise}       \\
                \end{cases}
            }
            \nextit{\fmlreplacepvarpby[{\flfp{\pvarb}{\fml}}]{\pvar}{\fmlc}}{
                \begin{cases}
                    \flfp{\pvarb}{\fml}                                  & \text{if }\pvar=\pvarb \\
                    \flfp{\pvarb}{\fmlreplacepvarby[\fml]{\pvar}{\fmlc}} & \text{otherwise}       \\
                \end{cases}
            }
            \nextit{\fmlreplacepvarpby[\lnot\fml]{\pvar}{\fmlc}}{\lnot\fmlreplacepvarby[\fml]{\pvar}{\fmlc}}
            \nextit{\fmlreplacepvarpby[\lpossible{\gamesymbat{\ivarbseq}{\termbseq}}\fml]{\pvar}{\fmlc}}{\lpossible{\gamesymbat{\ivarbseq}{\termbseq}}\fmlreplacepvarby[\fml]{\pvar}{\fmlc}}
            \lastit{\fmlreplacepvarpby[\fml\land\fmlb]{\pvar}{\fmlc}}{\fmlreplacepvarby[\fml]{\pvar}{\fmlc}\land\fmlreplacepvarby[\fmlb]{\pvar}{\fmlc}}
        \end{align*}%
    \end{aligntable}%

    \begin{lemmaE}[\folmus Substitution][normal]\label{prop:pvarsubst}
        For \emph{\folmus formulas} \(\fml,\fmlc\):
        \[\lmden{\intp}{\fmlreplacepvarby{\pvar}{\fmlc}}=\lmden{\intreplaceby{\pvar}{\lmden{\intp}{\fmlc}}}{\fml}.\]
    \end{lemmaE}

    \begin{proofE}
        By a straightforward induction on \folmus formulas \(\fml\).

        \begin{caselist}
            \case{\(\pvarb\)} Immediate.

            \case{\(\lnot\fml\)} Immediate.

            \case{\(\fml\land\fmlb\)} Immediate.

            \case{\(\lpossible{\gamesymbat{\ivarbseq}{\termbseq}}\fml\)}
            By the induction hypothesis the pairwise equivalences follow:
            \begin{align*}%
                \fostate\in\lmden{\intp}{\fmlreplacepvarpby[\lpossible{\gamesymbat{\ivarbseq}{\termbseq}}\fml]{\pvar}{\fmlc}}
                % \\
                \ciff &
                \mexists{\pstatepre\in\fonstrint{\gamesymb}(\ltden[\fostate]{\termbseq})}
                \replpstate[\fostate]{\pstatepreatvar{\ivarbseq}}\subseteq\lmden{\intp}{\fmlreplacepvarby{\pvar}{\fmlc}}
                \\
                \ciff &
                \mexists{\pstatepre\in\fonstrint{\gamesymb}(\ltden[\fostate]{\termbseq})}
                \replpstate[\fostate]{\pstatepreatvar{\ivarbseq}}\subseteq\lmden{\intreplaceby{\pvar}{\lmden{\intp}{\fmlc}}}{\fml}
                \\
                \ciff &
                \fostate\in\lmden{\intreplaceby{\pvar}{\lmden{\intp}{\fmlc}}}{\lpossible{\gamesymbat{\ivarbseq}{\termbseq}}\fml}.
            \end{align*}

            \case{\(\flfp{\pvarb}{\fml}\)}
            Suppose first \(\pvar=\pvarb\).
            Then:
            \begin{align*}
                \lmden{\intp}{\flfp{\pvarb}{\fml}}
                 & =
                \Intersection\{\fostatessubset : \lmden{\intreplaceby{\pvarb}{\fostatessubset}}{\fml}\subseteq \fostatessubset\}
                =
                \Intersection\{\fostatessubset : \lmden{\intreplaceby[(\intreplaceby{\pvar}{\lmden{\intp}{\fmlc}})]{\pvarb}{\fostatessubset}}{\fml}\subseteq \fostatessubset\}
                \\
                 &
                =
                \lmden{\intreplaceby{\pvar}{\lmden{\intp}{\fmlc}}}{\flfp{\pvarb}{\fml}}.
            \end{align*}
            If \(\pvar\neq\pvarb\), then since \(\pvar\) is free for \(\fmlc\) in \(\fml\):
            {
            \belowdisplayskip=-5pt
            \begin{align*}
                \lmden{\intp}{\flfp{\pvarb}{\fmlreplacepvarby[\fml]{\pvar}{\fmlc}}}
                 &
                =
                \Intersection\{\fostatessubset : \lmden{\intreplaceby{\pvarb}{\fostatessubset}}{\fmlreplacepvarby[\fml]{\pvar}{\fmlc}}\subseteq \fostatessubset\}
                =
                \Intersection\{\fostatessubset : \lmden{\intreplaceby[(\intreplaceby{\pvar}{\lmden{\intp}{\fmlc}})]{\pvarb}{\fostatessubset}}{\fml}\subseteq \fostatessubset\}
                \\
                 &
                =
                \lmden{\intreplaceby{\pvar}{\lmden{\intp}{\fmlc}}}{\flfp{\pvarb}{\fml}}.\qedhere
            \end{align*}
            }
        \end{caselist}%
    \end{proofE}

    %########################################################

    \subsection{Vectorial Assignment}\label{sec:vectorialassignment}

    It is convenient to have simultaneous vectorial assignments \(\gassign{\ivarseq}{\termseq}\), which assign the value of \(\term_i\) to \(\ivar_i\) for all~\(i\).
    The intended semantics is that \(\glgden{\gassign{\ivarseq}{\termseq}}(\fostatessubset)=\{\fostate:\streplaceby{\fostate}{\ivarseq}{\ltden{\termseq}}\in\fostatessubset\}\).
    Vectorial assignments are definable in \fogls and \folmus.
    In \fogls the game \(\gassign{\ivarb_1}{\term_1};\ldots;\gassign{\ivarb_\ell}{\term_\ell};\gassign{\ivar_1}{\ivarb_1};\ldots;\gassign{\ivar_\ell}{\ivarb_\ell}\), where \(\ivarb_1,\ldots,\ivarb_\ell\) are fresh, defines the vectorial assignment \(\gassign{\ivarseq}{\termseq}\).
    Analogously in \folmus the formula \(\lpossible{\gassign{\ivarseq}{\termseq}}\fml\) can be defined as
    \[\lpossible{\gassign{\ivarb_1}{\term_1}}\ldots\lpossible{\gassign{\ivarb_\ell}{\term_\ell}}\lpossible{\gassign{\ivar_1}{\ivarb_1}}\ldots\lpossible{\gassign{\ivar_\ell}{\ivarb_\ell}}\fml.\]

    %########################################################

    \section{Properties of Proof Calculi}\label{sec:properties-of-proof-calculi}

    \begin{lemma}
        The \excontextaxiomname axiom \irref{axcontext} is derivable in \fogls from \irref{univ}.
    \end{lemma}

    \begin{proofE}
        Derive the extended context axiom \irref{axcontext} by induction on \(\game\).
        The non-trivial case is for games of the form \(\garepeat{\game}\):
        {\renewcommand{\linferPremissSeparation}{\hspace{0.8cm}}
        \begin{sequentdeduction}
            \linfer
            {\linfer[diaind]
                {
                    \linfer
                    {
                        \linfer[axcontext]
                        {
                            \linfer[mon]
                            {
                                \linfer[loop]
                                {
                                    *
                                }
                                {((\fml\land\fmlb)\lor\lpossible{\game}\lpossible{\garepeat{\game}}(\fmlb\land\fml))\limply\lpossible{\garepeat{\game}}(\fmlb\land\fml)}
                            }
                            {((\fml\land\fmlb)\lor\lpossible{\game}(\fmlb\land\lpossible{\garepeat{\game}}(\fmlb\land\fml)))\limply\lpossible{\garepeat{\game}}(\fmlb\land\fml)}
                        }
                        {((\fml\land\fmlb)\lor(\fmlb\land\lpossible{\game}(\lnot\fmlb\lor\lpossible{\garepeat{\game}}(\fmlb\land\fml))))\limply\lpossible{\garepeat{\game}}(\fmlb\land\fml)}
                    }
                    {(\fml\lor\lpossible{\game}(\lnot\fmlb\lor\lpossible{\garepeat{\game}}(\fmlb\land\fml)))\limply(\lnot\fmlb\lor\lpossible{\garepeat{\game}}(\fmlb\land\fml))}
                }
                {\lpossible{\garepeat{\game}}\fml\limply(\lnot\fmlb\lor\lpossible{\garepeat{\game}}(\fmlb\land\fml))}
            }
            {(\fmlb\land \lpossible{\garepeat{\game}}\fml)\limply\lpossible{\garepeat{\game}}(\fmlb\land\fml)}
        \end{sequentdeduction}
        }%
        The instance of \irref{axcontext} is by the induction hypothesis on \(\game\).
    \end{proofE}

    \noindent
    Recall \Cref{prop:mon-box-derivable}, which says that the two rules \irref{monstrong} and \irref{diaindstrong} are derivable for $\freevars{\fmlc}\cap(\boundvars{\game})=\emptyset$.

    \begin{proofE}[Proof of \Cref{prop:mon-box-derivable}]
        The derivation of \irref{monstrong} is immediate from \irref{axcontext} and \irref{mon}.
        To derive \irref{diaindstrong} note first that \[\provfogl{(\fmlc\land\lpossible{\garepeat{\game}}\fml)\limply \lpossible{\garepeat{(\gtest{\fmlc};\game)}}(\fmlc\land\fml)}.\]
        This can be derived as follows:
        {\renewcommand{\linferPremissSeparation}{\hspace{0.8cm}}\footnotesize
        \begin{sequentdeduction}
            \linfer
            {
                \linfer[diaind]
                {\linfer[test]
                    {
                        \linfer[axcontext+test]
                        {
                            \linfer[loop]
                            {
                            }
                            {((\fmlc\land\fml)\lor\lpossible{\gtest{\fmlc};\game}\lpossible{\garepeat{(\gtest{\fmlc};\game)}}(\fmlc\land\fml))\limply\lpossible{\garepeat{(\gtest{\fmlc};\game)}}(\fmlc\land\fml)}
                        }
                        {((\fmlc\land\fml)\lor(\fmlc\land\lpossible{\game}(\lnot\fmlc\lor\lpossible{\garepeat{(\gtest{\fmlc};\game)}}(\fmlc\land\fml))))\limply\lpossible{\garepeat{(\gtest{\fmlc};\game)}}(\fmlc\land\fml)}
                    }
                    {(\fml\lor\lpossible{\game}(\lnot\fmlc\lor\lpossible{\garepeat{(\gtest{\fmlc};\game)}}(\fmlc\land\fml)))\limply(\lnot\fmlc\lor\lpossible{\garepeat{(\gtest{\fmlc};\game)}}(\fmlc\land\fml))}
                }
                {\lpossible{\garepeat{\game}}\fml\limply (\lnot\fmlc\lor\lpossible{\garepeat{(\gtest{\fmlc};\game)}}(\fmlc\land\fml))}
            }
            {(\fmlc\land\lpossible{\garepeat{\game}}\fml)\limply \lpossible{\garepeat{(\gtest{\fmlc};\game)}}(\fmlc\land\fml)}
        \end{sequentdeduction}
        }%

        Then \irref{monstrong} can be derived thus
        {\belowdisplayskip=-7pt
        {\renewcommand{\linferPremissSeparation}{\hspace{0.8cm}}
        \begin{sequentdeduction}
            \linfer
            {
                \linfer[diaind]
                {\linfer[test]
                    {
                        \linfer
                        {
                            \fmlc\limply((\fml\lor\lpossible{\game}\fmlb)\limply\fmlb)
                        }
                        {(\fmlc\land(\fml\lor\lpossible{\game}\fmlb))\limply\fmlb}
                    }
                    {((\fmlc\land\fml)\lor\lpossible{\gtest{\fmlc};\game}\fmlb)\limply\fmlb}
                }
                {\lpossible{\garepeat{(\gtest{\fmlc};\game)}}(\fmlc\land\fml)\limply\fmlb}
            }
            {\fmlc\limply(\lpossible{\garepeat{\game}}\fml\limply\fmlb)}
        \end{sequentdeduction}
        }%
        }
    \end{proofE}

    %########################################################

    \section{Least Fixpoint Logic and \folmusort}\label{sec:lfp-and-folmu}

    The \firstordermucalculus can also be seen as the extension of least fixpoint logic with state-changing modalities that describe \emph{atomic games}.
    Least fixpoint logic (\LFP) \cite{DBLP:journals/bsl/DawarG02} is the extension of first-order logic with the addition that \(\loglfplfp{\relsymb}{\ivarseq}{\fml}{\termseq}\) is a formula, whenever \(\relsymb\) is a \(k\)-ary relation symbol, \(\ivarseq\) and \(\termseq\) are \(k\)-sequences of \ivarname{s} and terms respectively and \(\fml\) is an \LFP-formula in which \(\relsymb\) appears only positively.
    The formula \(\loglfplfp{\relsymb}{\ivarseq}{\fml}{\termseq}\) is interpreted as the least fixpoint:
    \[\lfpden[\fonstr]{\loglfplfp{\relsymb}{\ivarseq}{\fml}{\termseq}} = \{\fostate :\ltden{\termseq}\in \mlfp{A}{\{\fonstrel : \streplaceby{\fostate}{\ivarseq}{\fonstrel}\in \lfpden[\fonstr_\relsymb^A]{\fml}\}}\}\]
    where \(\fonstr_\relsymb^A\) is \(\fonstr\) with the difference that \(\fonstrint[\fonstr_\relsymb^A]{\relsymb}=A\).
    The remaining first-order connectives are interpreted as in first-order logic.
    For a first-order signature \(\gsig\) let \(\gsig^*\) be the \gamesignature extending \(\gsig\) only by the \gamesymbol~\(\gndassignsymb\).

    \begin{theoremE}[\folmus and \LFP][see full proof]\label{prop:lfp-is-follmu}
        Least fixpoint logic in the signature~\(\gsig\) and the \firstordermucalculus in \(\gsig^*\) are equiexpressive.
    \end{theoremE}

    \begin{proof}
        Without loss of generality, assume that the formulas in question mention only \ivarname{s} from the fixed finite list of variables~\(\ivarsall\) of length \(\ell\).
        Any \folmus[\gsig^*]-formula can be written equivalently in \LFP by replacing \(\flfp{\pvar}{\fml}\) by \(\loglfplfp{\relsymb_X}{\ivarsall}{\fml}{\ivarsall}\) and \(\pvar\) by \(\relsymb_X(\ivarsall)\) and nondeterministic assignments \(\lpossible{\gndassign{\ivar}}\fml\) by quantifiers \(\lexists{\ivar}\fml\).
        This translation is easily seen to be correct.

        The converse translation from an \LFP formula \(\fml\) to \folmus is proved next.
        Without loss of generality, assume \(\freevars{\fmlb}\subseteq \ivarseq\) for every least fixpoint subformula of the form \(\loglfplfp{\relsymb}{\ivarseq}{\fmlb}{\termseq}\) so that the fixpoints always range over \emph{all} free variables of \(\fmlb\).
        Moreover, assume without loss of generality that every relation symbol in an \LFP formula is either used only for fixpoints or never used as a fixpoint.
        Also assume that every fixpoint relation symbol \(\relsymb\) is only bound once as \(\loglfplfp{\relsymb}{\ivarseq}{\fmlb_\relsymb}{\termseq}\) with the \emph{fixed} formula \(\fmlb_\relsymb\).
        By bound renaming assume without loss of generality that no free variable in a subformula \(\fmlb_\relsymb\) is bound anywhere in the formula \(\fml\).
        Fix for any fixpoint relation symbol \(\relsymb\) a \pvarname \(X_\relsymb\).

        Define the \folmus[\gsig^*] formula \(\lfptofml{\fml}\) corresponding to \(\fml\) by making the following substitutions
        \begin{align*}
             & \lfptofmlp{\relsymb(\termseq)}\synequiv\lpossible{\gassign{\ivarseq}{\termseq}}\pvar_{\relsymb}
             & \lfptofmlp{\loglfplfp{\relsymb}{\ivarseq}{\fmlb}{\termseq}}\synequiv\lpossible{\gassign{\ivarseq}{\termseq}}\flfp{\pvar_{\relsymb}}{{\lfptofmlp{\fmlb_\relsymb}}},
        \end{align*}
        if \(\relsymb\) is associated to the fixpoint \(\loglfplfp{\relsymb}{\ivarseq}{\fmlb}{\termseq}\).
        For \(\fostatessubset\subseteq\fostates\) let \(\strestrvar{\fostatessubset}{\ivarbseq}=\{\strestrvar{\fostate}{\ivarbseq} :\fostate\in\fostatessubset\}\) for all \(\ivarbseq\).
        Let \(\fonstr_\intp\) be the \fonstructure \(\fonstr\) where every fixpoint \(\relsymb\) is interpreted by \(\fonstrint[{\fonstr_\intp}]{\relsymb}=\strestrvar{\intp(\pvar_\relsymb)}{\freevars{\fmlb_\relsymb}}\).
        For any \(U\subseteq\fonstrdom^{\realnorm{\ivarbseq}}\) let \(\allexp{U} = \{\fostate:\strestrvar{\fostate}{\ivarbseq}\in U\}\).
        A set of states \(\fostatessubset\subseteq\fostates\) is \emph{\(\ivarbseq\)-closed} if \(\fostatessubset = \allexpp{\strestrvar{\fostatessubset}{\ivarbseq}}\).
        An \interp \(\intp\) is \emph{suitable} if every \(\intp(\pvar_\relsymb)\) is \(\freevars{\fml_\relsymb}\)-closed.
        Now by induction on \(\fml\) it is not hard to prove that \(\lfpden[\fonstr_\intp]{\fml}=\lmden{\intp}{\lfptofml{\fml}}\) and this set is \(\freevars{\fml}\)-closed for all suitable \(\intp\).
    \end{proof}

    \section{Properties of Continuous Evolutions}
    The following lemma is of technical utility in the context of \Cref{prop:ode-as-fixpoint-theorem}.

    \begin{lemma}\label{prop:bound-converge-uniformly}
        Let $F\colon\reals^\ell\to\reals^\ell$ be continuously differentiable and \(g_n\colon X\to\reals^\ell\) be a uniformly convergent sequence of bounded functions.
        Then \(F\circ g_n\) also converges uniformly.
    \end{lemma}

    \begin{proof}
        Let \(g=\lim_{n\to\infty}g_n\).
        By uniform convergence, there is some \(K\in\reals\) and some \(N\in\naturals\) such that \(\realnorm{g(x)}\leq K\) and \(\realnorm{g_n(x)}\leq K\) for all \(x\in X\) and all \(n\geq N\).
        Now note that \(F\) is \(L\)-Lipschitz continuous on the closed ball \(\{x\in\reals^\ell:\realnorm{x}\leq K\}\).
        Then \(\supnorm[X]{F\circ g_n-F\circ g}\leq L \supnorm[X]{g_n-g}\xrightarrow{n\to\infty} 0\) for all \(n\geq N\).
    \end{proof}

    For any set \(D\subseteq \reals^\ell\) let \(\setepsminus{D}{\varepsilon}=\{x : \epsnbhd[\varepsilon]{x}\subseteq D\}\).
    The following lemma shows how for compact initial regions and bounded time, certain behavior of differential equations is robust.

    \begin{lemma}\label{prop:ode-continuity-properties}
        Let \(F:\reals^\ell\to\reals^\ell\) be continuously differentiable and \(E\subseteq \reals^\ell\) be a subset.
        Moreover, let~\(D\subseteq \reals^\ell\) be open and \(C\subseteq \reals^\ell\) be compact.

        \begin{enumerate}
            \item If \(E\) is open, there are \(\varepsilon>0\), \(T\geq 0\) and a bounded set \(V\subseteq\reals^\ell\) such that\label{comp1}
                  \[C\subseteq \{\xel: \lexists{\yel\in D}  \contreachin{F}{\xel}{\yel}{E}\}\mimply C\subseteq \{\xel : \lexists{t\leq T,\yel\in \setepsminus{D}{\varepsilon}}  \contreachin[t]{F}{\xel}{\yel}{\setepsminus{E}{\varepsilon}\cap V}\}.\]
            \item If \(E\) is compact and \(T\geq 0\), there is \(\varepsilon>0\) such that \label{comp2}
                  \[C\subseteq\{\xel: \lforall{\yel,t\in[0, T]}  \contreachin[t]{F}{\xel}{\yel}{E}\mimply \yel\in D\}\mimply C\subseteq \{\xel : \lforall{\yel,t\in[0,T]}  \contreachin[t]{F}{\xel}{\yel}{E}\mimply \yel\in \setepsminus{D}{\varepsilon}\}.\]
        \end{enumerate}
    \end{lemma}

    \begin{proof}
        Let \(X_E=\{(\xel,t)\in C \times [0,\infty) : \mexists{\yel}  \contreachin[t]{F}{\xel}{\yel}{E}\}\), where the subscript \(E\) is dropped for \(E=\reals^\ell\).
        Define the \(E\)-flow \(\Phi:X_{\reals^\ell}\to\reals^\ell\) by \(\Phi(\xel,t)=z\), where \(z\) is the unique element with \(\contreach[t]{F}{\xel}{z}\).

        For \pref{comp1}: For every \(\xel\in C\), there is some \(t_\xel\) such that \(\Phi(\xel,t_\xel)\in D\) and \(\Phi(\xel,s)\in E\) for all \(s\in[0,t_\xel]\).
        By openness of \(D\), there is \(\delta_\xel>0\) such that \(\Phi(\xel,t_\xel)\in \setepsminus{D}{\delta_\xel}\).
        By continuity of the flow, there is an open neighborhood \( U_\xel\) containing \(\xel\) such that \(\Phi(\yel,t_\xel)\in \setepsminus{D}{\delta_\xel}\) for all \(\yel\in U_\xel\cap X_{\reals^\ell}\).
        Now \(\Phi(\xel,[0,t_\xel])\) is compact and contained in \(E\).
        Thus, there is \(\varepsilon_\xel>0\) such that \(\Phi(\xel,[0,t_\xel])\subseteq \setepsminus{E}{{\varepsilon_\xel}}\) and, without loss of generality, assume \(\varepsilon_\xel<\delta_\xel\).
        By joint continuity of \(\Phi\) and the tube lemma, there is a bounded open neighborhood \(V_\xel\subseteq U_\xel\) containing \(\xel\) such that \(\Phi(\yel,s)\in \setepsminus{E}{\varepsilon_\xel}\) for all \(\yel\in V_\xel\) and all \(s\in [0,t_\xel]\).
        By compactness of \(C\), there are finitely many \(\xel_1,\ldots, \xel_k\in C\) such that the \(V_{\xel_i}\) cover \(C\).
        Let \(T=\max_{1\leq i\leq k} t_{\xel_i}\), \(V=\bigcup_{1\leq i \leq k} \Phi(\closure{V_{\xel_i}},[0,t_{\xel_i}])\) and \(\varepsilon=\min_{1\leq i\leq k} \varepsilon_{\xel_i}\).
        These are as required.
        In particular, \(V\) is bounded as a finite union of continuous images of compact sets.

        For \pref{comp2}:
        Let \(K=X_E\cap (C\times [0,T])\).
        Then \(K\) is bounded because \(C\) is bounded.
        Moreover, \(K\) is closed.
        To see this, consider a sequence \((\xel_i,t_i)\in K\) with \(\lim_{i\to\infty}(\xel_i,t_i)=(\xel,t)\). Clearly \(\xel\in C\) and \(t\in [0,T]\).
        It suffices to show that \((\xel,t)\in X_E\), so suppose this were not the case.
        Then there must be some \(s\leq t\) such that \((\xel,s)\in X_{\reals^\ell}\) and \(\Phi(\xel,s)\notin E\).
        (If the solution does not exist for time \(t\), then it leaves the compact set \(E\) first.)
        First, suppose \(s=t\).
        Then \(E\notni\Phi(\xel,t)=\lim_{i\to\infty}\Phi(\xel_i,t_i)\in E\) which is a contradiction.
        Thus, \(s<t\).
        By continuity of the flow, there is some open neighborhood \(U\ni \xel\) such that \(\Phi(\yel,s)\notin E\) for all \(\yel\in U\).
        Choosing \(i\) large enough such that \(\xel_i\in U\) and \(s\leq t_i\) yields \(\Phi(\xel_i,s)\notin E\), which contradicts \((\xel_i,t_i)\in X_E\).
        Hence \(K\) is compact.

        By continuity of the flow, \(\Phi(K)\) is compact, and by assumption \(\Phi(K)\subseteq D\).
        Thus, \(\mathrm{dist}(\Phi(K),\reals^\ell\setminus D) > \varepsilon\) for some \(\varepsilon>0\), since \(D\) is open and \(\Phi(K)\) is compact.
        Consequently, \(\Phi(K)\subseteq \setepsminus{D}{\varepsilon}\).
    \end{proof}

    \ifinlineproofs
    \else
        \section{Proofs}
        \label{sec:proofs}
        \printProofs[proofs]
    \fi
\fi

\end{document}